\documentclass{amsart}
\usepackage{geometry}
\usepackage{amsthm, amssymb}
\usepackage{color}
\usepackage{algorithm}
\usepackage{algorithmic}
\usepackage{graphicx,subfigure}
\usepackage{tikz}
\usepackage{wrapfig}
\usepackage{tabulary}
\usepackage{booktabs,url}

\newcommand{\be}{\begin{equation}}
\newcommand{\ee}{\end{equation}}
\newcommand{\EE}{\mathbb{E}}

\newcommand{\RR}{\mathbb{R}}

\newtheorem{thm}{Theorem}
\newtheorem{defn}[thm]{Definition}
\newtheorem{rmk}[thm]{Remark}
\newtheorem{prop}[thm]{Proposition}
\newtheorem{lem}[thm]{Lemma}
\newtheorem*{lem*}{Lemma}
\newtheorem*{thm*}{Theorem}
\newtheorem{cor}[thm]{Corollary}

\newtheorem{Assumption}[thm]{Assumption}

\newcommand{\argmin}{\operatornamewithlimits{argmin}}

\begin{document}

\title[Vector nonlocal Euclidean median]{Vector Nonlocal Euclidean Median: Fiber Bundle Captures The Nature of Patch Space} 

\author{Chen-Yun Lin}
\address{Chen-Yun Lin\\
Department of Mathematics\\
University of Toronto}
\email{cylin@math.toronto.edu}

\author{Arin Minasian}
\address{Arin Minasian\\
Department of Electrical Engineering\\
University of Toronto}
\email{arin.minasian@mail.utoronto.ca}

\author{Xin Jessica Qi}
\address{Xin Jessica Qi\\
Department of Mathematics\\
University of Toronto}
\email{jxin.qi@alum.utoronto.ca}

\author{Hau-Tieng Wu}
\address{Hau-Tieng Wu\\
Department of Mathematics\\
University of Toronto}
\email{hauwu@math.toronto.edu}

\maketitle

\begin{abstract}
We extensively study the rotational group structure inside the patch space by introducing the fiber bundle structure. The rotational group structure leads to a new image denoising algorithm called the \textit{vector non-local Euclidean median} (VNLEM). The theoretical aspect of VNLEM is studied, which explains why the VNLEM and the traditional non-local mean/non-local Euclidean median (NLEM) algorithm work. The numerical issue of the VNLEM is improved by taking the orientation feature in the commonly applied scale-invariant feature transform (SIFT), and a theoretical analysis of the robustness of the orientation feature in the SIFT is provided. The VNLEM is applied to an image database of 1,361 images and a comparison with the NLEM is provided. Different image quality assessments based on the error-sensitivity or the human visual system are applied to evaluate the performance. The results confirmed the potential of the VNLEM algorithm.
\end{abstract}

\section{Introduction}
\label{intro}
Image denoising is a long-lasting challenge in the image processing field. Much effort has been invested in this problem in the past decades. While there are several approaches to handle this problem, like the variational approach, the wavelet approach, the partial differential equation approach, etc (see \cite{Buades_Coll_Morel:2010,Jain_Tyagi:2016} for an overall survey), we focus on the idea of nonlocal filtering and its associated theoretical analysis in this paper.

Based on the idea that pixels spatially far apart in an image can be similar or even the same, Buades et al. pioneered the nonlocal mean (NLM) filters \cite{buades01} to denoise a noisy image. The motivation for the NLM could be summarized by taking the patch space into account \cite{Lee_Pedersen_Mumford:2003,Carlsson2007,Peyre2008,Peyre2009,singer2009NLM,Chaudhury_Singer:2012,Chaudhury_Singer:2013,Perea_Carlsson:2014, Osher_Shi_Zhu:2016,Yin_Gao_Lu_Dubechies:2016}. For the $i$-th pixel on a given image $I$, we could associate  it with a patch $P_i$ of size $q\times q$, where $q$ is the patch size determined by the user. Mathematically, a patch is the restriction of the image on a subset around $i$, so that the $i$-th pixel of $I$ is the central pixel of $P_i$. Denote the set of all patches as $\mathcal{X}_I$. The main assumption is that $\mathcal{X}_I$ is located on, or could be well approximated by, a low dimensional geometric object, like a manifold. By viewing a $q\times q$ patch as a $q^2$-dim vector, the manfiold is a subset of the Euclidean space $\mathbb{R}^{q^2}$, and we could endow an induced Riemannian metric on the manifold from $\mathbb{R}^{q^2}$. Under this assumption, two patches with a similar intensity, while they might be far apart in the image, are close in the intrinsic geometry of the manifold.

Under this low-dimensional and nonlinear patch space structure, the NLM algorithm is introduced \cite{Buades_Coll_Morel:2010,singer2009NLM}. In brief, in the NLM algorithm, to denoise a given pixel $i$, we find the $m\in \mathbb{N}$ nearest neighboring patches of the patch $P_i$, denoted as $\mathcal{N}_i$, with respect to the Euclidean distance and then denoise the $i$-th pixel by evaluating the mean of all central pixels of those patches in $\mathcal{N}_i$. It has been well known that the NLM algorithm leads to better edge preservation \cite{Jain_Tyagi:2016}, and this improvement is directly related to the diffusion process on the nonlinear geometric structure \cite{singer2009NLM}. The NLM algorithm can be understood as reducing the noise influence on the patch space via the diffusion process \cite{singer2009NLM}.
Several generalizations of NLM follow based on this diffusion idea. By noting that the mean operator is sensitive to the outliers, the authors in \cite{Chaudhury_Singer:2012} considered replacing the mean in the NLM by the median, which leads to the nonlocal Euclidean median (NLEM). In brief, after finding the neighbors of $P_i$, the $i$-th pixel is denoised by evaluating the median of all central pixels of those patches in $\mathcal{N}_i$. It is shown in \cite{Chaudhury_Singer:2012} that the NLEM could tolerate more noises inside the noisy patches. This idea has been applied to the single-channel blind source separation problem to reconstruct the ``wave-shape function'' and decompose the fetal electrocardiogram signal from the maternal abdominal electrocardiogram signal \cite{Su_Wu:2016}. Furthermore, by noticing that the mean operator is equivalent to minimizing a functional based on the $L^2$ norm and the median operator is equivalent to minimizing a functional based on the $L^1$ norm, and by the need of enhancing the sparsity structure, the nonlocal patch regression (NLPR) is considered in \cite{Chaudhury_Singer:2013}, which replaces the $L^1$ norm in the associated functional by the $L^p$ norm, where $0<p<1$. 
We mention that the above model and algorithm have been applied to different fields, like the medical imaging problem \cite{Chan_Fulton_Feng_Meikle:2010} and the inpainting problem  \cite{Gepshtein2013,Osher_Shi_Zhu:2016,Yin_Gao_Lu_Dubechies:2016,Zhou_et_al:2012}.

As successful as the patch space model and those diffusion-based algorithms are, there are structures in the patch space that we can consider to further improve the algorithm and theoretical problems we need to answer.
From the model perspective, there are structures in the patch space not considered in the past, particularly the rotational group structure. Since the central pixel of a patch is fixed after rotation, two patches could be viewed the same, or \textit{rotationally invariant}, if they are the same up to a rotation. Therefore, in the patch space model, we could take the rotational group into account. From the theoretical viewpoint, to the best of our knowledge, a study explaining why neighbors could be well approximated from the noisy patches in NLM/NLEM/NLPR was not available. Also, a discussion and explanation of how the patch size should be chosen is lacking. Furthermore, in the literature, the denoising performance is commonly evaluated by the ``error-based'' measurements, like the signal-to-noise ratio (SNR) or the peak SNR, but it has been well known that those error-based quantities might capture only partial information of the image quality, and more needs to be considered.

In this paper, we aim to advance the progress on these problems. We take the rotational group into account, and model the patch space by a fiber bundle. In this model, the set of rotationally invariant patches is modeled by a fiber that is diffeomorphic to $SO(2)$, and the collection of the set of rotationally invariant patches (or the orbits coming from the $SO(2)$ action), denoted as $\mathcal{X}_I/SO(2)$, is parametrized by a manifold, which is the base manifold of the fiber bundle. We then generalize the NLM/NLEM/NLPR algorithm by taking the fiber bundle structure into account, and call the new algorithm the vector nonlocal Euclidiean median (VNLEM). In the VNLEM, the rotationally invariant distance (RID) associated with the fiber bundle structure is considered so that two rotationally invariant patches will have RID equal to $0$. With the RID, we could determine the neighboring patches, and then evaluate the median value of the central pixels of all neighboring patches. Note that this leads to a dimensional reduction of the patch space, since we work with a 1-dim lower base manifold. Hence, we get more samples for the denoise purpose in the VNLEM, when compared with the NLM/NLEM/NLPR. 
From the theoretical perspective, we study how accurate we could estimate the neighborhood from the noisy patches, and provide a quantification in Theorem \ref{Theorem:RIDNeighborsEstimate}. In brief, we show that with high probability, which depends on the patch size and the noise level, we could accurately determine the neighborhood from the noisy patches under the RID or the Euclidean distance. By noting that the probability we could determine the correct neighbors depends on the patch size, we could explain why the patch space approach leads to a better denoising result compared with the pixel-based NLM or NLEM algorithm. On the other hand, we also discuss that the patch size cannot be too large, or the patch space will be too ``complicated'' so that the diffusion algorithm might fail.
From the algorithmic perspective, we need to handle the numerical problem for the VNLEM. Note that with the RID distance, the base manifold is no longer embedded in the ordinary Euclidean space, and the ordinary fast nearest neighbor search algorithms cannot be applied. As far as we know, there is no fast algorithm available to determine the neighbors under the RID metric. Our solution is via a relaxation step. We consider the commonly applied scale-invariant feature transform (SIFT) features to estimate candidates for the neighbors. Then we run the RID to determine the true neighbors. To guarantee the applicability of this relaxation, in addition to discussing the relationship between the RID distance and the neighbors determined by the SIFT features, we show that the SIFT features are robust to noise.
Finally, in addition to the ordinary error-based measurements, we consider image quality measurements that take the human visual system into account to evaluate the performance of the proposed VNLEM algorithm. The result is reported on a large scale image database consisting of 1,361 images.

The paper is organized as the following. In Section \ref{sec:model}, we introduce a rotationally invariant distance since the patch space allows a canonical rotation action. Then, we propose a principal bundle model for the patch space of an image and provide both continuous and discrete versions of our model. 
In Section \ref{sec:alg}, we give the VNLEM algorithm and discuss how we deal with numerical issues to make the algorithm computationally affordable. One main step is to use orientations in the SIFT algorithm to approximate rotation angles between patches. 
In Section \ref{sec:kNN}, we show that with high probability, we could accurately determine nearest neighbors of a clean patch through finding nearest neighbors of the associated noisy patch in the noisy patch space. In Section \ref{sec:SIFT}, we provide a mathematical definition of the orientation feature in the SIFT and show why such approximations are reliable when patches are noisy. 
The performance evaluation measurements are summarized in Section \ref{sec:img_quality}.
In Section \ref{sec:stats}, we show our numerical results.
In Appendix \ref{sec:DM}, we give a brief review of diffusion geometry which is used for dimension reduction to help find good nearest neighbors for image denoising. In Appendix \ref{sec:manifold_model}, we discuss a possible model under which we could approximate a patch space by a manifold .

\begin{table}[htbp]\caption{Table of commonly used notation throughout the paper}
\begin{center}
\begin{tabular}{r c p{10cm} }
\toprule
$\| \cdot \|$ & & $\ell^2$ norm \\
$\Phi_t^{(m)}$ & & DM with the diffusion time $t>0$ and the first $m$ non-trivial eigenvectors\\
$D_t^{(m)}(\cdot, \cdot)$ & & diffusion distance (DD) \\
$SO(2)$ & & rotation group\\
$O \in SO(2)$ & & rotation matrix\\
$d_{\texttt{RID}}(\cdot,\cdot)$ & & rotationally invariant distance\\
$\iota^{(\texttt{c})} : \mathbb{R}^2 \rightarrow \mathbb{R}$ & & continuous clean image\\
$p^{(\texttt{c})}_{\mathbf{x}} : \mathbb{R}^2 \rightarrow \mathbb{R}$ & & continuous round clean patch centered at $\mathbf{x}\in\mathbb{R}^2$ with radius $r$\\ 
$\iota^{(\texttt{n})} : \mathbb{R}^2 \rightarrow \mathbb{R}$ & & continuous noisy image\\
$p^{(\texttt{n})}_{\mathbf{x}} : \mathbb{R}^2 \rightarrow \mathbb{R}$ & & continuous round noisy patch centered at $\mathbf{x}\in\mathbb{R}^2$ with radius $r$\\ 
$\mathcal{X}^{(\texttt{c})}$ & & the patch space of image $\iota^{(\texttt{c})}$ \\
$I^{(\texttt{c})} \in \mathbb{R}^{N \times N}$ & & discrete clean image of size $N \times N$\\
$\xi (\cdot) $ & & $\xi(\cdot) \sim \mathcal{N}(0,1)$ i.i.d. Gaussian white noise \\$I^{(\texttt{n})} = I^{(\texttt{c})} + \sigma \xi $ & & discrete noisy image of size $N \times N$\\
$P_i^{(\texttt{c})} \in \mathbb{R}^{q\times q}$ & & clean patch of size $q \times q$ centered at $i$-pixel \\
$P_i^{(\texttt{n})} = P_i^{(\texttt{c})} + \sigma \xi$ & & noisy patch of size $q \times q$ centered at $i$-pixel \\
$\mathcal{X}^{(\texttt{c})}_I$ & & the patch space of image $I^{(\texttt{c})}$\\
$\mathcal{X}^{(\texttt{n})}_I$ & & the patch space of image $I^{(\texttt{n})}$\\
$G(\mathbf{s})$ & & Gaussian function of scale $1$ in the SIFT algorithm\\
$L(\mathbf{x}) = L_p(\mathbf{x})$  & & Gaussian smoothed patch of patch $p$\\
$L^{(\texttt{c})}$ & & Gaussian smoothed clean patch\\
$L^{(\texttt{n})}$ & & Gaussian smoothed noisy patch\\
$\theta^{(c)*}$ & & orientation assignment of a clean patch \\
$\theta^{(n)*}$ & & orientation assignment of a noisy patch\\
$C^\infty_c(\mathbb{R}^2)$ & & the space of $C^\infty$ functions with compact supports defined on $\mathbb{R}^2$\\
$\mathcal{D}'(\mathbb{R}^2)$ & & the space of distributions defined on $\mathbb{R}^2$\\
\bottomrule
\end{tabular}
\end{center}
\label{tab:TableOfNotationForMyResearch}
\end{table}

\section{Mathematical Model}
\label{sec:model}

We start from recalling the patch space commonly used in the image processing field \cite{Lee_Pedersen_Mumford:2003,Donoho2005,Carlsson2007,Peyre2008,Peyre2009,singer2009NLM,Chaudhury_Singer:2012,Chaudhury_Singer:2013,Perea_Carlsson:2014, Osher_Shi_Zhu:2016,Yin_Gao_Lu_Dubechies:2016}. Take a grayscale image $I$ of size $N$-pixels wide and $N$-pixels long, which is a real function defined on $\mathbb{Z}_N^2$, where  $\mathbb{Z}_N=\{1,2,\ldots,N\}$.
Usually we represent $I$ as a $N\times N$ matrix with real entries. In general, we could consider an image with different width and length, but to simplify the discussion, we limit our focus on square images in this paper. 
Call a point $(\alpha,\beta)\in \mathbb{Z}_N^2$ the $(\alpha,\beta)$-th pixel of the image, and $I(i,j)$ is the intensity of the grayscale image $I$.
Take an odd natural number $q$. For each pixel $(\alpha,\beta)\in \mathbb{Z}_N^2$, we associate it with a patch $P_{(\alpha,\beta)}\in \mathbb{R}^{q\times q}$, which is defined as
\begin{align}
&P_{(\alpha,\beta)}(k,l)\\
:=&\,\left\{
\begin{array}{ll}
I((\alpha+k-(q-1)/2,\beta+l-(q-1)/2))&\mbox{when }((\alpha+k-(q-1)/2,\beta+l-(q-1)/2))\in \mathbb{Z}_N^2\\
0& \mbox{otherwise},
\end{array}
\right.\nonumber
\end{align}
for $k,l=1,\ldots,q$. Specifically, we use the notation $P_{(\alpha,\beta)}(c)$ to indicate the central point of the $(\alpha,\beta)$-th patch, $P_{(\alpha,\beta)}((q+1)/2,(q+1)/2)$.

To express the notation in a compact format, we stack the columns of the matrix $I$ into a vector $I^\vee\in\mathbb{R}^{N^2}$, where the superscript $\vee$ means the {\it vector form}, that is, the $((\ell-1)N+1)$-th to the $(\ell N)$-th entries in $I^\vee$ is the $\ell$-th column of $I$, where $\ell=1,\ldots,N$. Similarly, denote $P_{(\alpha,\beta)}^\vee\in \mathbb{R}^{q^2}$ to be the vector form of the $(\alpha,\beta)$-th patch.
When there is no danger of confusion, we ignore the superscript $\vee$ and use the notation $I$ to represent the grayscale image in the matrix form and in the column form interchangeably, and denote $P_i:=P_{(\alpha,\beta)}^\vee\in \mathbb{R}^{q^2}$, where $i=(\alpha-1)q+\beta$. 
We have the following definition for the discretized patch space.
For a given grayscale image $I\in\mathbb{R}^{N\times N}$ and the patch size $q$, where $q$ is an odd integer number, the discretized patch space is defined as
\begin{equation}
\mathcal{X}_I:=\{P_i\}_{i=1}^{N^2}\subset \mathbb{R}^{q^2}.
\end{equation}

Besides the general consensus that we could approximate the patch space by the manifold model \cite{Lee_Pedersen_Mumford:2003,Peyre2008,Peyre2009,singer2009NLM,Chaudhury_Singer:2012,Chaudhury_Singer:2013,Osher_Shi_Zhu:2016,Yin_Gao_Lu_Dubechies:2016}, the structure of the patch space is less discussed, except the discussion of the Klein bottle in \cite{Carlsson2007,Perea_Carlsson:2014}. 
Inspired by the fact that two image patches might be the same up to a rotation, in \cite{Qi2015}, the rotation structure naturally existed in the patch space was taken into account. In addition to \cite{Qi2015}, the same orientation idea was considered in  \cite{Zimmer2008,Grewenig2011,Guizard2015,Sreehari2015,Zhang_He_Du:2016}, and has been applied to the neuroimaging analysis \cite{Manjon2012,Guizard2015}.

In this section, we introduce a fiber bundle structure for the patch space. To make clear how to incorporate the $SO(2)$ group into the model, and how to numerically rotate a patch, we start from a continuous setup. We define the continuous patch space as a topological space with the $SO(2)$ group structure. Then we discuss how to obtain the discrete model from a continuous one. Recall the definition of a fiber bundle with the group structure \cite{Gilkey:1980}. 

\begin{defn} [Fiber bundle with group structure $G$]
Let $F$ and $M$ be manifolds. A fiber bundle $E$ with fiber $F$ over $M$ consists of a topological space $E$ together with a map $\pi: E: \rightarrow M$ satisfying the local triviality condition. 
Let $G$ be a Lie group, for example the rotation group $SO(2)$. Let the map $. \, : G \times F \rightarrow F$  be a smooth left action of $F$. That is, the map (the action of $G$) $.\, : (x,g) \mapsto g.x$ from $G \times F$ to $F$
satisfies
\be 
e.g = g \mbox{ for all $g\in G$,}
\ee
where $e$ is the neutral element of $G$ and
\be
\label{compatibility}
(gh).x = g.(h.x)\mbox{ for all $x\in M$ and $g,h \in G$}.
\ee
\end{defn}
This group $G$ is often called the gauge group or the transition group. 
If the fiber is equal to the structure group, then $(\pi, E, M, G)$ is called a principal $G$ bundle. (See, for example, Definition 5.7 in \cite{Gallier2005}). This definition is equivalent to the definition of principal bundle requiring $G$ as a right action on $E$. (See, for example, Propositions 5.5 and 5.6 in \cite{Gallier2005}).

\subsection{Continuous Model}
Let $\iota^{(\texttt{c})}: \RR^2 \rightarrow \mathbb{R}$ be a $L^2$ function that represents an image. 
The round patches of a finite radius from the image $\iota^{(\texttt{c})}$ is defined below.  

\begin{defn}[Continuous patch and continuous patch space]\label{Definition:ContinuousPatchSpace}
Fix a $L^2$ image $\iota^{(\texttt{c})}$ and $r>0$. A patch centered at $\mathbf{x} \in \mathbb{R}^2$, denoted as $p^{(\texttt{c})}_{\mathbf{x}}$, is defined as 
\be
p^{(\texttt{c})}_{\mathbf{x}}(\mathbf{s}) = \iota^{(\texttt{c})}(\mathbf{x+s})\psi(\mathbf{s}), 
\ee
where $\mathbf{s} \in \mathbb{R}^2$ and $\psi\in C_c^\infty$ defined as
\be
\psi(\mathbf{s})=\left\{
\begin{array}{ll}
1&\mbox{when }\|\mathbf{s}\|_{\mathbb{R}^2}\leq 3r/2\\
0&\mbox{when }\|\mathbf{s}\|_{\mathbb{R}^2}> 2r
\end{array}
\right..
\ee
The (continuous) patch space associated with $\iota^{(\texttt{c})}$ is denoted as 
\be
\mathcal{X}^{(\texttt{c})} := \{ p^{(\texttt{c})}_{\mathbf{x}} : \mathbf{x} \in\mathbb{R}^2\} \subset L^2(\mathbb{R}^2).  
\ee
\end{defn}

The superscript ``$(c)$'' in the definition indicates that the image is clean. In the literature, it is common to assume that the patch space $\mathcal{X}^{(\texttt{c})}$ is located on, or could be approximated by, a low dimensional manifold \cite{Lee_Pedersen_Mumford:2003,singer2009NLM,Chaudhury_Singer:2012,Perea_Carlsson:2014, Osher_Shi_Zhu:2016}. We will make this assumption in this paper. 
%
To further capture the structure of the patch space, that is, two image patches might be the same up to a rotation, we could consider a $SO(2)$ action on the patch space.
For $O \in SO(2)$ and $\mathbf{s}\in \mathbb{R}^2$ expressed by a column vector, we define the action on $p_{\mathbf{x}}$ as
\be
       (O. p^{(\texttt{c})}_{\mathbf{x}})(\mathbf{s}) = p^{(\texttt{c})}_{\mathbf{x}}(O^{-1}\mathbf{s}).
\ee
This is a left group action since for any $O_1, O_2 \in SO(2)$,
\begin{align}
O_2.( O_1.p^{(\texttt{c})}_{\mathbf{x}})(\mathbf{s})  = O_1.p^{(\texttt{c})}_{\mathbf{x}}(O_2^{-1}\mathbf{s}) = p^{(\texttt{c})}_{\mathbf{x}}(O_1^{-1}O_2^{-1}\mathbf{s}) = (O_2O_1).p^{(\texttt{c})}_{\mathbf{x}}(\mathbf{s})
\end{align}

Since each fiber, including a patch and its rotated patches, can be identified as $S^1$, and $SO(2)$ is diffeomorphic to $S^1$,  the patch space $\mathcal{X}^{(\texttt{c})}$ could be viewed as a principal $SO(2)$ bundle. 

By identifying patches up to a rotation, we have the quotient space $\mathcal{X}^{(\texttt{c})}/SO(2)$.  
We make the following assumption 
\begin{Assumption}
The patch space $\mathcal{X}^{(\texttt{c})}$ is a subset of the fiber bundle $E=\mathcal{X}^{(\texttt{c})}$ with $\pi: E \rightarrow M$ and a left $SO(2)$ group action so that the quotient space $M:=\mathcal{X}^{(\texttt{c})}/SO(2)$ is a manifold. The $SO(2)$ action preserves the fiber that is diffeomorphic to $SO(2)$. 
\end{Assumption}

We consider the rotationally invariant distance (RID) to measure the similarity between patches. 
\begin{defn}
Let $p^{(\texttt{c})}_{\mathbf{x}_1}, p^{(\texttt{c})}_{\mathbf{x}_2}$ be two patches in $\mathcal{X}^{(\texttt{c})}$. The rotational invariant distance is defined as 
\be
d_{\texttt{RID}}(p^{(\texttt{c})}_{\mathbf{x}_1}, p^{(\texttt{c})}_{\mathbf{x}_2}) = \min_{O \in SO(2)} \left( \int_{\mathbb{R}^2} \left| p^{(\texttt{c})}_{\mathbf{x}_1}(\mathbf{s}) - O.p^{(\texttt{c})}_{\mathbf{x}_2}(s) \right|^2 d\mathbf{s}\right)^{1/2}. 
\ee
\end{defn}

The minimum in the RID could be achieved since $SO(2)$ is compact. Also note that the definition is equivalent to 
\be
d_{\texttt{RID}}(p^{(\texttt{c})}_{\mathbf{x}_1}, p^{(\texttt{c})}_{\mathbf{x}_2}) = \min_{O_1, O_2 \in SO(2)} \left( \int_{\mathbb{R}^2} \left| O_1.p^{(\texttt{c})}_{\mathbf{x}_1}(\mathbf{s}) - O_2.p^{(\texttt{c})}_{\mathbf{x}_2}(\mathbf{s}) \right|^2 d\mathbf{s} \right)^{1/2}, 
\ee
since 
\begin{align}
 \int_{\mathbb{R}^2} | O_1.p^{(\texttt{c})}_{\mathbf{x}_1}(\mathbf{s}) - O_2.p^{(\texttt{c})}_{\mathbf{x}_2}(\mathbf{s}) |^2 d\mathbf{s} 
 =& \, \int_{\mathbb{R}^2} | p^{(\texttt{c})}_{\mathbf{x}_1}(O_1^{-1}\mathbf{s}) - p^{(\texttt{c})}_{\mathbf{x}_2}(O_2^{-1}\mathbf{s}) |^2 d\mathbf{s} \nonumber \\
 =& \, \int_{\mathbb{R}^2} | p^{(\texttt{c})}_{\mathbf{x}_1}(\mathbf{s}) - p^{(\texttt{c})}_{\mathbf{x}_2}(O_2^{-1} O_1\mathbf{s}) |^2 d\mathbf{s}
\end{align}
by a change of variables and the fact that $O_1, O_2 \in SO(2)$.

Consider an isotropic homogeneous generalized Gaussian random field $\Phi$ with a finite variance defined on $\mathbb{R}^2$ to model the noise \cite[Chapter III.5]{Gelfand:1964}.\footnote{Recall the definition of $\Phi$. For any $\phi\in C_c^\infty(\mathbb{R}^2)$, $\Phi(\phi)$ is a random variable with mean $0$ and finite variance. For functions $\phi_1(x),\ldots,\phi_m(x)\in C_c^\infty(\mathbb{R}^2)$, any vector $v\in\mathbb{R}^2$, and any rotation or reflection $O$ of $\mathbb{R}^2$, the $m$-dimensional random variables
\be
(\Phi(\phi_1(x)),\ldots,\Phi(\phi_m(x))),\quad(\Phi(\phi_1(x+v)),\ldots,\Phi(\phi_m(x+v))),
\mbox{ and }(\Phi(\phi_1(O.x)),\ldots,\Phi(\phi_m(O.x)))\nonumber
\ee
are identically distributed.}
The noisy image is defined as
\be
\iota^{(\texttt{n})}=\iota^{(\texttt{c})}+\Phi\in \mathcal{D}'(\mathbb{R}^2)\,,\label{Definition:NoisyImage}
\ee
where we assume that the spectral measure of $\Phi$ \cite[p.264]{Gelfand:1964} is the same as $\sigma^2 d\xi$, $\sigma>0$, and $d\xi$ is the Lebesgue measure on $\mathbb{R}^2$. The superscript ``$(n)$'' in the definition indicates that the image is noisy.
By the same way as that in Definition \ref{Definition:ContinuousPatchSpace}, the noisy patch centered as $\mathbf{x}$ is denoted as 
\be
p^{(\texttt{n})}_{\mathbf{x}}=p^{(\texttt{c})}_{\mathbf{x}}+\psi \Phi, \label{definition:ContinuousNoisyPatch}
\ee
and the noisy patch space is denoted as $\mathcal{X}^{(\texttt{n})}$. Note that since $\Phi$ is a generalized random field and the cut-off function $\psi$ in (\ref{Definition:ContinuousPatchSpace}) is in $C_c^\infty(\mathbb{R}^2)$, the noisy patches are well-defined. However, in general the RID cannot be defined for two noisy patches in the continuous setup, since the noisy patches are distributions.

\subsection{Discrete Model}
We use a mollifier to create discrete patches from the continuous patches $p^{(\texttt{c})}_{\mathbf{x}}$ or $p^{(\texttt{n})}_{\mathbf{x}}$. Note that if an image is regular enough, like continuous, then the discretization could be easily achieved by evaluating the image at the designed grid points. For a $L^2$ or more general image, however, we need a mollifier to achieve this discretization. Note that we could consider a more general model for an image, like a distribution, but to simplify the discussion we focus on the $L^2$ image.

First, consider the following discretization map.
\begin{defn}
Let $\eta$ be a mollifier on $\mathbb{R}^2$, that is, 
\begin{itemize}
\item $\eta\in C_c^\infty$ with the unitary $L^2$ norm;
\item $\lim_{\epsilon \rightarrow 0} \eta^{\epsilon}(\mathbf{y}) = \delta$ in the weak sense, where $\displaystyle \eta^{\epsilon}(\mathbf{y}) := \frac{1}{\epsilon^2}\eta(\frac{\mathbf{y}}{\epsilon})$ and $\delta$ is the Dirac delta measure. 
\end{itemize}
For a fixed $\epsilon>0$ and a set of grid points $\mathcal{G}:=\{\mathbf{x}_i\}_{i=1}^n\subset \mathbb{R}^2$, we consider a discretization map
\begin{eqnarray}
\mathcal{D}^{\epsilon}_{\mathcal{G}} : L^2(\mathbb{R}^2) &\rightarrow& \mathbb{R}^{n} \nonumber\\
	f &\mapsto& 
	\begin{bmatrix} 
	f \star \eta^{\epsilon}_{\mathbf{x}_1}(\mathbf{x}_1)\\ 
	\vdots \\ 
	f \star \eta^{\epsilon}_{\mathbf{x}_{n}}(\mathbf{x}_n) 
	\end{bmatrix}\in \mathbb{R}^{n}, 
\end{eqnarray}
where $\eta^\epsilon_{\mathbf{x}_i}(\mathbf{y}) := \eta^{\epsilon}(\mathbf{x}_i-\mathbf{y})$ and $\star$ means the convolution. 
\end{defn}

In this work, to fulfill the conventional definition of a discrete patch, we consider the discrete patch to be defined on a square that inscribes the $D_r$, which is a disk centered at the origin with radius $r>0$.   
For a fixed odd integer $q$, we consider a $q \times q$ square sampling grid $\mathcal{G}_q:=\{\mathbf{x}_a\}_{a=1}^{q^2}\subset\mathbb{R}^2$, where $\mathbf{x}_a=\left( \frac{r}{\sqrt{2}} + (\alpha -1) \frac{\sqrt{2}r}{q-1},  -\frac{r}{\sqrt{2}} + (\beta -1) \frac{\sqrt{2}r}{q-1} \right)^T\in\mathbb{R}^2$ and $(\alpha, \beta)$ is the associated index such that $a=(\alpha-1)q+\beta$. See Figure \ref{fig:grid} for example. A discrete patch corresponding to $p^{(\texttt{c})}_{\mathbf{x}}$ is defined as
\be
P^{(\texttt{c})}_{\mathbf{x}} := \mathcal{D}^{\epsilon}_{\mathcal{G}_q}p^{(\texttt{c})}_{\mathbf{x}} \in \mathbb{R}^{q^2},
\ee  
where $\epsilon$ is assumed to be much smaller than $\sqrt{2}r/(q-1)$ and the superscript $(c)$ indicates that the image is clean. Note that the distance between two vertically or horizontally consecutive grid points is $\sqrt{2}r/(q-1)$. See Figure \ref{fig:grid} for an example. 
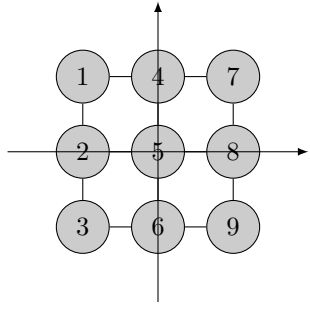
\begin{figure}
\begin{tikzpicture}[darkstyle/.style={circle,draw,fill=gray!40,minimum size=20}]
  \foreach \x in {0,...,2}
    \foreach \y in {0,...,2} 
       {\pgfmathtruncatemacro{\label}{ 3-\y + 3 *\x}
       \node [darkstyle]  (\x\y) at (1*\x,1*\y) {\label};} 

  \foreach \x in {0,...,2}
    \foreach \y [count=\yi] in {0,...,1}  
      \draw (\x\y)--(\x\yi) (\y\x)--(\yi\x) ;
 \draw[thin,-latex] (-1,1) -- (3,1);
 \draw[thin,-latex] (1,-1) -- (1,3);
\end{tikzpicture}
 \caption{Given $q=3$, we consider a $3 \times 3$ grid centered at the origin. Here we denote $\mathbf{x}_a  = a$, where $a = 1, \cdots, 9$.} \label{fig:grid}
\end{figure}
The discrete grayscale image associated with a continuous $L^2$ image $\iota^{(\texttt{c})}$ is 
\be
I^{(\texttt{c})} :=\mathcal{D}^{\epsilon}_{\mathcal{G}_N}\iota^{(\texttt{c})} \in \mathbb{R}^{N^2},
\ee
where $\mathcal{G}_N = \{ \mathbf{x}_i \}_{i=1}^{N^2}$ denotes a uniform sampling grid.

Putting these definitions together, the $i$-th patch associated with the discrete grayscale image are related by
\be
P^{(\texttt{c})}_i:=\mathcal{D}^{\epsilon}_{\mathcal{G}_q}p^{(\texttt{c})}_{\mathbf{x}_i},
\ee
and we denote the discrete patch space associated with $I$ by 
\be
\mathcal{X}_I^{(\texttt{c})}:=\{P^{(\texttt{c})}_{i}\}_{i=1}^{N^2}.
\ee

We would like to define $SO(2)$ actions on a discretized patch. Recall that for a given discretized patch $P^{(\texttt{c})}_{\mathbf{x}}$, the numerical rotation of $P^{(\texttt{c})}_{\mathbf{x}}$ by $O\in SO(2)$ is carried out by
\be
\mathcal{D}^{\epsilon}_{\mathcal{G}_q}(O.\mathcal{I} P^{(\texttt{c})}_{\mathbf{x}})=\mathcal{D}^{\epsilon}_{\mathcal{G}_q}(O.\mathcal{I} \mathcal{D}^{\epsilon}_{\mathcal{G}_q} p^{(\texttt{c})}_{\mathbf{x}}),
\ee
where $\mathcal{I}$ is the selected interpolation operator. Here $\mathcal{I}$ could be viewed as a deconvolution operator trying to recover $p^{(\texttt{c})}_{\mathbf{x}}$ from $\mathcal{D}^{\epsilon}_{\mathcal{G}_q}p^{(\texttt{c})}_{\mathbf{x}}$. Suppose $\mathcal{I}\mathcal{D}^{\epsilon}_{\mathcal{G}_q}$ is the identity operator, then the numerical rotation of  $P^{(\texttt{c})}_{\mathbf{x}}$ by $O\in SO(2)$ becomes $\mathcal{D}^{\epsilon}_{\mathcal{G}_q}(O.p^{(\texttt{c})}_{\mathbf{x}})$. In general, however, this is not true, unless the function $\iota^{(\texttt{c})}$ has a special structure so that we can find $\mathcal{I}$. 
As we utilize interpolation to rotate a discrete patch, the discrepancy between the numerical rotation of $\mathcal{D}_{\mathcal{G}_q}^{\epsilon}p^{(\texttt{c})}_{\mathbf{x}}$ and $\mathcal{D}^{\epsilon}_{\mathcal{G}_q}( O.p^{(\texttt{c})}_{\mathbf{x}})$ depends on the rotational angle, the underlying function, and the selected interpolation algorithm. In other words, the discretization and rotation operations are not interchangeable. 

Since the numerical rotation performance is not the main focus of this work, to simplify the discussion, we further assume that the numerical impact of numerical rotation of a discrete patch is negligible, and hence
$\mathcal{I}\mathcal{D}^{\epsilon}_{q}$ is the identity operator. Thus, we have the following definition.
\begin{defn}
We define the $SO(2)$ group action on a discrete patch $P^{(\texttt{c})}_{\mathbf{x}}$ by
\be
O.P^{(\texttt{c})}_{\mathbf{x}} := \mathcal{D}^{\epsilon}_q (O.p^{(\texttt{c})}_{\mathbf{x}}), \mbox{ for any $O \in SO(2)$}.
\ee
\end{defn}

Next, we discuss the discretization procedure for the noisy patches. 
Indeed, since the mollifier is a function in $C_c^\infty$, the noisy image $\iota^{(\texttt{n})}$ could be discretized as
\begin{eqnarray}
\mathcal{D}^{\epsilon}_{\mathcal{G}_N} : D' &\rightarrow& \mathbb{R}^{N^2} \nonumber\\
	\iota^{(\texttt{n})}&\mapsto& 
	I^{(\texttt{n})}:=\begin{bmatrix}  
	\iota^{(\texttt{n})} \star \eta^{\epsilon}_{\mathbf{x}_1}(\mathbf{x}_1)\\ 
	\vdots \\ 
	\iota^{(\texttt{n})}\star \eta^{\epsilon}_{\mathbf{x}_{N^2}}(\mathbf{x}_{N^2}) 
	\end{bmatrix} =I^{(\texttt{c})}+\begin{bmatrix} 
	\Phi(\eta^{\epsilon}_{\mathbf{x}_1})\\ 
	\vdots \\ 
	\Phi(\eta^{\epsilon}_{\mathbf{x}_{N^2}}) 
	\end{bmatrix} , \label{Definition:DiscretizationNoisyImage}
\end{eqnarray}
where $\begin{bmatrix} 
	\Phi(\eta^{\epsilon}_{\mathbf{x}_1})
	\ldots
	\Phi(\eta^{\epsilon}_{\mathbf{x}_{N^2}}) 
	\end{bmatrix}^T$ 
is a Gaussian random vector by the definition of $\Phi$. Precisely, by the assumption of $\Phi$ in (\ref{Definition:NoisyImage}) and the chosen $\epsilon$ in the discretization operator, $\Phi(\eta^{\epsilon}_{\mathbf{x}_i})$ is a Gaussian random variable, $\mathbb{E}\Phi(\eta^{\epsilon}_{\mathbf{x}_i})=0$ for $i=1,\ldots,N^2$, and $\Phi(\eta^{\epsilon}_{\mathbf{x}_1}),\ldots,\Phi(\eta^{\epsilon}_{\mathbf{x}_{N^2}})$ are uncorrelated and hence independent since we have
\begin{align}
\mathbb{E}[\Phi(\eta^{\epsilon}_{\mathbf{x}_i})\Phi(\eta^{\epsilon}_{\mathbf{x}_j})]=\,&\int |\hat{\eta^{\epsilon}}(\xi)|^2e^{i2\pi (\mathbf{x}_i-\mathbf{x}_j)\cdot \xi}\sigma^2d \xi\\
=\,&\sigma^2\int \eta^{\epsilon}(\mathbf{y})\eta^{\epsilon}(\mathbf{x}_i-\mathbf{x}_j-\mathbf{y})d\mathbf{y}=\sigma^2\delta_{ij}\nonumber,
\end{align}
where $\delta_{ij}$ is the Kronecker delta and  $i,j=1,\ldots,N^2$.
The $i$-th patch associated with the noisy discrete grayscale image is
\be
P^{(\texttt{n})}_i:=\mathcal{D}^{\epsilon}_{\mathcal{G}_q}p^{(\texttt{n})}_{\mathbf{x}_i}=P^{(\texttt{c})}_{i}+\sigma\xi_i,\label{Definition:DiscretizationNoisyImagePatch}
\ee
where $\xi_i$ is a Gaussian random vector, and $\xi_i(a)\sim \mathcal{N}(0,1)$ and $\mathbb{E}(\xi_i(a)\xi_i(b))=\delta_{ab}$ for all $a,b=1,\ldots,q^2$. It is clear that for two non-overlapping patches $P^{(\texttt{n})}_i$ and $P^{(\texttt{n})}_j$, the associated noises $\xi_i$ and $\xi_j$ are independent. The discrete patch space associated with the noisy image $I^{(\texttt{n})}$ is denoted by 
\be
\mathcal{X}_I^{(\texttt{n})}:=\{P^{(\texttt{n})}_{i}\}_{i=1}^{N^2}\subset \mathbb{R}^{q^2}.\label{Definition:DiscretizationNoisyImagePatchSpace}
\ee

Again, we assume that the error incurred by the numerical rotation is negligible, and have the following definition for the noisy patches.
\begin{defn}\label{Definition:DiscreteRotation}
We define the $SO(2)$ group action on a discrete patch $P^{(\texttt{n})}_{i}$ by
\be
O.P^{(\texttt{n})}_{i} := \mathcal{D}^{\epsilon}_{\mathcal{G}_q} (O.p^{(\texttt{n})}_{\mathbf{x}_i}), \mbox{ for any $O \in SO(2)$}.
\ee
\end{defn}

Note that due to the isotropic assumption of the homogeneous random field $\Phi$, the distribution of the noise in a noisy patch is fixed after rotation. The RID in the discrete setup is thus defined as the following.

\begin{defn}
The rotation invariant distance between two patches, $P_{i}$ and $P_{j}$, which could be clean or noisy, is defined as
\be
d_{\texttt{RID}}(P_{i}, P_{j}) := \min_{O\in SO(2)} \|P_{i} - O.P_{j}\|,
\ee
where $\| \cdot \|$ denotes the $\ell^2$ norm. 
\end{defn}

Before closing this section, we show an example to demonstrate the benefit of introducing the frame bundle structure to the patch space. In Fig.~\ref{PatchSpaceModelDemonstration2}, we show the 49 nearest neighbors of the patch $P$ indicated by the white box shown in Fig.~\ref{PatchSpaceModelDemonstration} determined by the RID and $L^2$ distances respectively. Clearly, with the RID, more nearest neighbor patches that are similar to $P$ are identified. In other words, we reduce the dimension of the patch space by wiping out the fiber associated with the rotationally invariant patches.

\begin{figure}[ht]
\includegraphics[width=0.30\textwidth]{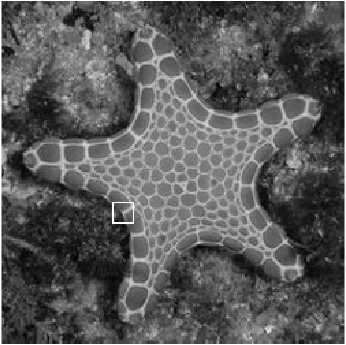}
\caption{\label{PatchSpaceModelDemonstration}Starfish. The selected patch $P$ is indicated by the white box.}
\end{figure}

\begin{figure}[ht]
\subfigure[$L^2$ distance]{
\includegraphics[width=.43\textwidth]{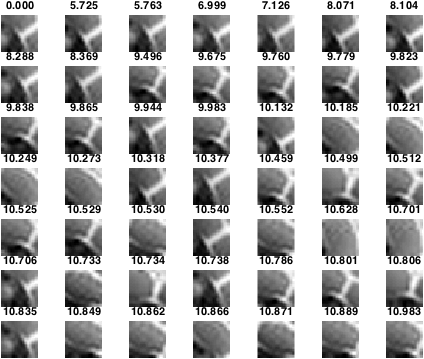}\hspace{30pt}}
\subfigure[RID]{
\includegraphics[width=.43\textwidth]{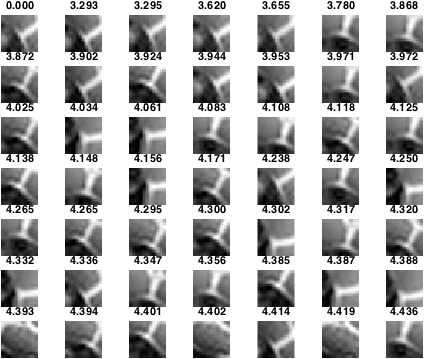}}
\caption{\label{PatchSpaceModelDemonstration2}Left: the first 49 nearest neighbors of the patch $P$ shown in Figure \ref{PatchSpaceModelDemonstration} with respect to the $L^2$ distance, including $P$. The patch $P$ is shown in the left top subfigure, and the $L^2$ distance is shown on the top of each patch. Right: the first 49 nearest neighbors of the patch $P$ with respect to the RID, including $P$. The patch $P$ is shown in the left top subfigure, and the RID is shown on the top of each patch.}
\end{figure}

\section{Vector non-local Euclidean median Algorithm}
\label{sec:alg}

Take a clean grayscale image denoted as $I^{(\texttt{c})}\in \RR^{N \times N}$. 
Assume that the image has been normalized to be of mean $0$ and standard deviation $1$; that is, 
\be
\mu_I:=\frac{1}{N^2}\sum_{i=1}^{N^2}I^{(\texttt{c})}(i)=0\quad\mbox{and}\quad\sigma_I:=\left(\frac{1}{N^2}\sum_{i=1}^{N^2}(I^{(\texttt{c})}(i)-\mu_I)^2\right)^{1/2}=1. \label{Definition:SigmaI}
\ee
Take the associated noisy image $I^{(\texttt{n})}$ defined in (\ref{Definition:DiscretizationNoisyImage}), the noisy patches $P^{(\texttt{n})}_i$ defined in (\ref{Definition:DiscretizationNoisyImagePatch}) with $\sigma>0$, and the noisy patch space defined in (\ref{Definition:DiscretizationNoisyImagePatchSpace}). We now introduce the VNLEM algorithm.

\subsection{VNLEM algorithm}

The goal of the \textit{denoising problem} is finding an algorithm that will recover $I^{(\texttt{c})}$ from $I^{(\texttt{n})}$ as accurately as possible. We will come back to the notion of accuracy in Section \ref{sec:img_quality}.
In this paper, we consider the following \textit{vector nonlocal Euclidean median} (VNLEM) algorithm, which is a generalization of the NLM, the NLEM \cite{Chaudhury_Singer:2012}, and the NLPR \cite{Chaudhury_Singer:2013}.
The basic idea is to combine the fiber bundle structure underlying the patch space in order to improve the performance of the NLM and NLEM algorithms. 
With the RID, define the affinity matrix $W\in \mathbb{R}^{N^2\times N^2}$ by
\begin{align}
W_{ij} = \exp(- d_{\texttt{RID}}^2(P^{(\texttt{n})}_i, P^{(\texttt{n})}_j) / \epsilon),\label{Affinity:True_RID}
\end{align}
where $i,j=1,\ldots,N^2$ and $\epsilon>0$ is the pre-determined bandwidth.
For each patch $P^{(\texttt{n})}_i$, we identify its $N_1\in \mathbb{N}$ nearest neighbors in the sense of the RID,
and we represent this set as $N_{\texttt{RID}}(i)$. The denoised image, denoted as $\tilde{I}^{(\texttt{VNLEM})}\in \mathbb{R}^{N\times N}$, is calculated by 
\begin{align}
\tilde{I}^{(\texttt{VNLEM})}(i) = \tilde{P}^{(\texttt{VNLEM})}_i(c),\quad \mbox{where }%
\tilde{P}^{(\texttt{VNLEM})}_i:=\argmin_{P\in \mathbb{R}^{q\times q}} \sum_{P_j \in N_{\texttt{RID}}(i)} W_{ij} \| P- O_{ij}.P^{(\texttt{n})}_j\| ^{\gamma},\label{Definition:VNLEM}
\end{align}
where $i=1,\ldots,N^2$, $0 < \gamma \leq 1$, and $O_{ij}$ is the rotation that achieves $d_{\texttt{RID}}^2(P^{(\texttt{n})}_i, P^{(\texttt{n})}_j)$. Note that when $\gamma=1$, this is equivalent to taking the median over $\{P^{(\texttt{n})}_k(c)\}_{k\in N_{\texttt{RID}}(i)}$. When $0<\gamma<1$, this is equivalent to the NLPR proposed in \cite{Chaudhury_Singer:2013}. We call the algorithm VNLEM with $0<\gamma\leq1$.

Under the manifold assumption, we could apply the diffusion map (DM) algorithm to further improve the VNLEM algorithm. We summarize the DM and the theory behind it in Appendix \ref{sec:DM}. 
With the affinity matrix $W$, the graph Laplacian and the associated transition matrix $A$ could be established. By taking the top $K$ eigenvalues and eigenvectors of $A$, we then embed each patch into a low-dimensional space and calculate the diffusion distance (DD) to evaluate the true neighbors of each patch. For each patch $P^{(\texttt{n})}_i$, we identify its $N_2\in \mathbb{N}$ nearest neighbors in the sense of DD. 
We represent this set of nearest neighbors of $P^{(\texttt{n})}_i$ as $N_{\text{DD}}(i)$. Based on the robustness property of the DM \cite{ElKaroui:2010a,ElKaroui_Wu:2015b}, this step acts as an additional filtering procedure to dismiss the patches in the initial nearest neighbors set determined by the RID.   
The denoised image, denoted as $\tilde{I}^{(\texttt{VNLEM-DD})}\in \mathbb{R}^{N\times N}$, is calculated by  
\begin{align}
\tilde{I}^{(\texttt{VNLEM-DD})}(i) = \tilde{P}^{(\texttt{VNLEM-DD})}_i(c),\quad  \mbox{where }%
\tilde{P}^{(\texttt{VNLEM-DD})}_i:=\argmin_{P\in \mathbb{R}^{q\times q}} \sum_{P^{(\texttt{n})}_j \in N_{\texttt{DD}}(i)} W_{ij} \| P-O_{ij}. P^{(\texttt{n})}_j\|^\gamma\label{Definition:VNLEM-DD}
\end{align}
and $i=1,\ldots,N^2$.

While it seems a straightforward generalization of the NLM/NLEM/NLPR by replacing the Euclidean distance by the RID, there are numerical issues we have to handle, and we discuss three of them below.

\subsubsection{Numerical techniques to speed up the computation -- search window}

First, note that the established $W$ is a dense matrix, which is not feasible to handle when the image size is large. Furthermore, obtaining the pairwise distances between all pairs of patches is a computationally intense and time-consuming task. One practical solution to this issue is to only consider the nearest neighbors of any given patch when forming the affinity matrix. By finding a pre-assigned number of nearest neighbors, we could simultaneously reduce the computational time and the memory required to save $W$. However, to the best of our knowledge, there is no available efficient nearest neighbor searching algorithm for the RID. 

To handle this numerical issue, we consider the \textit{search window} scheme \cite{Chaudhury_Singer:2012} by limiting our algorithm to consider only patches that are within a given search window centered around the reference patch. Precisely, for a patch $P^{(\texttt{n})}_i$, we consider the search window of size $(2N_2+1)\times (2N_2+1)$ that is centered at the 
\be
S_i:=\{P^{(\texttt{n})}_j|\, j\in\{1,\ldots,N^2\},\,\mbox{the difference of $i$ and $j$ is bounded by $N_2$ in both $x$ and $y$ axes}\},
\ee
where $N_2\in\mathbb{N}$ so that $(2N_2+1)\times (2N_2+1)>N_1$.
That is, when we establish $W$, we only consider the $(2N_2+1)\times (2N_2+1)$ patches whose centers are within $N_2$ pixels away from the center of $P^{(\texttt{n})}_i$ in both the $x$-axis and $y$-axis.
With this search window, we form the affinity matrix by the following:
\begin{align}
 W_{ij} = \left\{
 \begin{array}{ll}
 \exp(-d_{\texttt{RID}}^2(P^{(\texttt{n})}_i, P^{(\texttt{n})}_j) / \epsilon), & \mbox{for } P^{(\texttt{n})}_j \in S_i
  \\
  0, & \text{otherwise} 
\end{array}
\right..\label{Algorithm:Defintion:Wmatrix}
\end{align}
 
\subsubsection{Numerical techniques to improve the RID evaluation -- SIFT}

Note that finding the RID between two given patches incurs huge computational costs in its general form. Also, in general the patch is square and the size is limited, like $11\times 11$ or $13\times 13$, performing a direct numerical rotation might lead to a non-negligible error and deviate the estimated RID. To alleviate these two troubles and facilitate the derivation of the affinity matrix, we use the scale invariant feature transform (SIFT) \cite{Lowe:2004} to approximate the RID. We mention that the central moments are used in \cite{Zimmer2008,Grewenig2011,Guizard2015} and the curvelet transform is used in \cite{Zhang_He_Du:2016} to capture the rotational feature.

SIFT is an algorithm to extract the local features in an image. The particular feature extracted by the SIFT that we have interest in is the orientation feature. In short, for each pixel, based on the local image gradient direction, an orientation angle is calculated and assigned as the local feature. We will use this feature orientation to approximate the RID distances between the patches. Denote the orientation for the local feature centered in $P^{(\texttt{n})}_i$ as $\theta^{(\texttt{n})}_i$. The relative angle between $P^{(\texttt{n})}_i$ and $P^{(\texttt{n})}_j$ achieving the RID is approximated by $\theta^{(\texttt{n})}_i - \theta^{(\texttt{n})}_j$. We then rotate $P^{(\texttt{n})}_j$ by $R_{\theta^{(\texttt{n})}_i }.R^{-1}_{\theta^{(\texttt{n})}_j}.P^{(\texttt{n})}_j$, where $R_{\theta}\in SO(2)$ means the rotation by $\theta$ degrees.

Note that the SIFT provides an approximation of the angular relationship between two patches, which allows us to approximate the RID between two patches. However, it might not be accurate. To improve the accuracy, we perform the exhaustive search over a small range centered around the estimated angular relationship $\theta^{(\texttt{n})}_i - \theta^{(\texttt{n})}_j$:
\begin{align}
\theta_{ij}:=\argmin_{\theta \in \{\theta : \vert \theta - (\theta^{(\texttt{n})}_i - \theta^{(\texttt{n})}_j) \vert < \theta_l   \}} \| U_kP^{(\texttt{n})}_i - R_{\theta}.U_kP^{(\texttt{n})}_j \|,\label{approximate_Rotation_RID}
\end{align} 
where $U_k:\mathbb{R}^{q^2}\to \mathbb{R}^{k^2q^2}$ is the chosen upsampling operator that increases the sampling rate of the patch $P^{(\texttt{n})}_i$ by $k\in\mathbb{N}$ times and $\theta_l>0$ is the parameter chosen by the user,
and hence
\begin{align}
\tilde{d}_{\texttt{RID}}(P^{(\texttt{n})}_i, P^{(\texttt{n})}_j):= \| P^{(\texttt{n})}_i - R_{\theta_{ij}}.P^{(\texttt{n})}_j \|,\label{approximate_accurate_RID}
\end{align} 
which is used as an approximation of the RID distance.
Note that $U_k$ is applied to improve the accuracy of numerical rotation.
To estimate the value of the clean image at pixel $i$, we use the following affinity weights in~\eqref{Definition:VNLEM-DD}:
\begin{align}
W'_{ij} = \exp(-\tilde{d}^2_{\texttt{RID}}(P^{(\texttt{n})}_i, P^{(\texttt{n})}_j)/\epsilon),
\label{accurate_RID}
\end{align}
where $j \in S_i$.
Note that this method for finding RID offers a trade-off between computational time and accuracy of the result. With the estimated RID and the estimated affinity matrix $W'$, we could run the DM and get the estimated DD.

\subsubsection{The proposed VNLEM algorithm}

The proposed algorithms, after taking the above modifications into account, are summarized in Algorithm~\ref{Algorithm_VNLEM}. 
We call the modified denoising scheme in (\ref{Definition:VNLEM}) based on the approximated RID (\ref{accurate_RID}) the \textit{VNLEM} algorithm, and call the modified denoising scheme in (\ref{Definition:VNLEM-DD}) based on the estimated DD the \textit{VNLEM with DD} (VNLEM-DD). 
We denote $\tilde{I}^{(\texttt{VNLEM})}$ as the denoised image by VNLEM and $\tilde{I}^{(\texttt{VNLEM-DD})}$ as the denoised image by VNLEM with DD.

\begin{algorithm}[h] 
\begin{algorithmic}
\STATE \textbf{Input :} Noisy image $I^{(\texttt{n})}$, patch size $q\in\mathbb{N}$, the number of nearest neighbors $N_1\in\mathbb{N}$, the search window size $N_2\in\mathbb{N}$, the kernel bandwidth $\epsilon>0$, the DM embedding dimension $m\in\mathbb{N}$, the diffusion time $t>0$, and the power $0<\gamma\leq 1$.
\STATE \textbf{Output :} Denoised image $\tilde{I}$.
\STATE[pre-1] Pad the image array with a border of $\lceil  q/ 2 \rceil$ pixels.
\STATE[pre-2] Create the patch space $\mathcal{X}^{(\texttt{n})} := \{P^{(\texttt{n})}_i\}_{i = 1}^{N ^ 2}\subset\mathbb{R}^{q^2}$, where the center of $P^{(\texttt{n})}_i$ is $I^{(\texttt{n})}(i)$.
\STATE[pre-3] Find SIFT orientation feature for each patch and form an affinity matrix $W$ using these orientations from the search window $S_i$ of size $(N_2 + 1) \times (N_2 + 1)$ according to equation (\ref{Algorithm:Defintion:Wmatrix}).

\STATE[VNLEM. Step 1] For each $i$, find $N_1$ nearest neighbours from $S_i$ according to $W$.
\STATE[VNELM. Step 2] Find the more accurate estimation of RID, $\tilde{d}_{\texttt{RID}}$ in (\ref{approximate_accurate_RID}), and form $N_{\texttt{RID}}(i)$ that contains $\lceil N_1 / 2\rceil$ patches that are closer to patch $i$ according to $\tilde{d}_{\texttt{RID}}$, where $\lceil x\rceil$ means the smallest integer greater than or equal to $x\in \mathbb{R}$.

\STATE[VNLEM. Step 3] For each $i$, set $\tilde{I}^{(\texttt{VNLEM})}(i)$ to be the center point of 
$$\argmin_{P\in \mathbb{R}^{q\times q}} \sum_{P^{(\texttt{n})}_j \in N_{\texttt{RID}}(i)} W'_{ij} \| P- O_{ij}.P^{(\texttt{n})}_j\|^\gamma.$$

\STATE[VNLEM-DD. Step 1] Form the eigenvalue decomposition of $D^{-1} W$, where $D\in \mathbb{R}^{N^2\times N^2}$ is the diagonal matrix determined by $D_{ii}=\sum_{j=1}^{N^2}W_{ij}$. 

\STATE[VNLEM-DD. Step 2] Embed $P^{(\texttt{n})}_i$ into $\mathbb{R}^{m}$ by $\Phi_t^{(m)}(P^{(\texttt{n})}_i) = (\lambda_2^t \phi_2(i), \ldots, \lambda_{m+1}^t \phi_{m+1}(i))$ and evaluate the DD between patches. 

\STATE[VNLEM-DD. Step 3] For each $P_i^{(\texttt{n})}$, find $N_1$ nearest neighbours in terms of DD.

\STATE[VNLEM-DD. Step 4] Find $\lceil N_1 / 2\rceil$ closest patches with respect to $\tilde{d}_{\texttt{RID}}$ among the $N_1$ patches from the previous step to form $N_{\texttt{DD}}(i)$.

\STATE[VNLEM-DD. Step 5] For each $i$, set $\tilde{I}^{(\texttt{VNLEM-DD})}(i)$ to be the center point of $$\argmin_{P\in \mathbb{R}^{q\times q}} \sum_{P^{(\texttt{n})}_j \in N_{\texttt{DD}}(i)} W'_{ij} \| P- O_{ij}.P^{(\texttt{n})}_j\|^\gamma.$$

\end{algorithmic}
\caption{Vector non-local Euclidean median algorithm. }
\label{Algorithm_VNLEM}
\end{algorithm}

\section{Theoretical Analysis}

In this section, we provide a theoretical analysis to study the proposed VNLEM algorithm. The first part concerns how the VNLEM algorithms work. Precisely, we claim that the clean patch neighbors could be accurately evaluated from the noisy patch neighbors with high probability. This theorem also explains why the traditional nonlocal mean/median algorithm work. The second part concerns how accurate the orientation feature determined by the SIFT could help us to accurately approximate the RID.

\subsection{Finding good neighborhoods}
\label{sec:kNN}

In this section, we show that through finding nearest neighborhoods of noisy patches, it is with high probability that we would find ``correct'' nearest neighborhoods of clean patches as well. 

Note that  the rotation group action on patches can be expressed as 
\be
O.P^{(\texttt{n})}_j= O.P^{(\texttt{c})}_j + \sigma O.\xi_j,
\ee
where $O \in SO(2)$. 
When two patches $P^{(\texttt{n})}_i$ and $P^{(\texttt{n})}_j$ do not overlap, the noises of $P^{(\texttt{n})}_i$ and $P^{(\texttt{n})}_j$ are independent. 
However, when $P^{(\texttt{n})}_i$ and $P^{(\texttt{n})}_j$ overlap, the associated noises are not independent, and we need to control the dependence.
 To achieve this, we introduce the following sets that are associated with $P^{(\texttt{n})}_i$ and $P^{(\texttt{n})}_j$:
\be
K_{\texttt{O}}(O) := \{ (a, b)|\, a, b \in \{1,\cdots, q^2\} \mbox{ such that } \xi_i(a) = [O.\xi_j](b)   \}\,,
\ee 
which is associated with the overlapped pixels of $P^{(\texttt{n})}_i$ and $P^{(\texttt{n})}_j$ and is dependent on the rotation $O$, but whose cardinality does not depend on $O$;
\be
K_{\texttt{S}}(O) := \{ (a,b)|\, a, b \in \{1,\cdots, q^2\} \mbox{ such that } a\neq b,\,\xi_i(a) = O.\xi_j(b) \mbox{ and } \xi_i(b) = [O.\xi_j](a) \},
\ee 
which is associated with the ``swapped'' pixel indices after rotation and depends on $O\in SO(2)$;
and
\be 
K_{\texttt{I}}(O) := \{a\in \{1,\cdots, q^2\} |\, \xi_i(a) = [O.\xi_j](a) \},
\ee
which is associated with the overlapped pixels of $P^{(\texttt{n})}_i$ and $P^{(\texttt{n})}_j$ with ``identical indices'' after rotation and depends on $O\in SO(2)$. By definition, the cardinalities of $K_{\texttt{S}}(O)$ and $K_{\texttt{I}}(O)$ both depend on $O$.

Note that $K_{\texttt{S}}(O)\subset K_{\texttt{O}}(O)$ and $K_{\texttt{I}}(O)\subset K_{\texttt{O}}(O)$, and when $P^{(\texttt{n})}_i$ and $P^{(\texttt{n})}_j$ do not overlap, $K_{\texttt{O}}(O)$, $K_{\texttt{S}}(O)$ and $K_{\texttt{I}}(O)$ are all empty sets. Also note that $K_{\texttt{I}}(O)$ would only have at most one element. We mention that for the NLEM, since there is no rotation, $K_{\texttt{I}}(O)$ and $K_{\texttt{S}}(O)$ will be empty.

To better illustrate the sets $K_{\texttt{O}}(O)$, $K_{\texttt{S}}(O)$, and $K_{\texttt{I}}(O)$, see Figures \ref{Fig:OverlapContinuousSetup180} to \ref{Fig:OverlapDiscreteSetupS180} when the rotation is of 180 degree.  
It is easier to visualize the overlap in the continuous setup. Let $\mathbf{x}$ be a point in the overlap region of $p^{(\texttt{n})}_i$ and $p^{(\texttt{n})}_j$. If there exists $O \in SO(2)$ so that after rotation the point would be in the same relative position in $p^{(\texttt{n})}_i$ and $O.p^{(\texttt{n})}_j$, then the distances from $\mathbf{x}$ to the boundaries of $p^{(\texttt{n})}_i$ and $p^{(\texttt{n})}_j$ must be the same. See Figure \ref{Fig:OverlapContinuousSetup180} for an illustration. 
Therefore, when two patches have overlap, only the points on the line segment connecting the intersection points of the boundary circles. In the discrete setup, the same consideration holds. See Figure \ref{Fig:OverlapDiscreteSetup180} for an illustration. Therefore, for each rotation action, there would be at most one pixel pixel being lined up, and hence $|K_{\texttt{I}}(O)| \leq 1$. 

On the other hand, if $K_{\texttt{S}}(O)$ is not empty, there exist two points $\mathbf{x}$ and $\mathbf{y}$ in the intersection so that their corresponding rotated points $\mathbf{x}'$ and $\mathbf{y}'$ are at the same relative positions but swapped. This is only possible when the rotation angle is $\pi$ in the continuous setup. See Figure \ref{Fig:OverlapContinuousSetup90}.
Note that the overlap region is symmetric about the centre $\mathbf{x}$. When the rotation angle is $\pi$ in the discrete setup, the overlapping region, except $\mathbf{x}$, are points in $K_{\texttt{S}}(O)$, and hence $|K_{\texttt{S}}(O)| \leq |K_{\texttt{O}}(O)|$. See Figure \ref{Fig:OverlapDiscreteSetupS180}. Clearly, in general we have a rough bound $|K_{\texttt{O}}(O)|\leq q(q-1)$.

\begin{figure}[h]
\begin{tikzpicture}
\draw[blue, ultra thick] (1+5/3,0) circle [radius=5/3];
\draw[gray, ultra thick] (1+11/3,0) circle [radius=5/3];
\draw[<->, red, ultra thick, dotted] (1+8/3,0) -- (1+10/3,0);
\draw[<->, red, ultra thick, dotted] (1+2,0) -- (1+8/3,0);
\draw[fill] (1+8/3,0) circle [radius=0.05];
\node at (1+4/3,0) {$p^{(\texttt{n})}_i$};
\node at (1+12/3,0) {$p^{(\texttt{n})}_j$};
\node at (1+8/3, 0.3) {$\mathbf{x}$};
\draw[blue, ultra thick] (8+5/3,0) circle [radius=5/3];
\draw[gray, ultra thick] (10+11/3,0) circle [radius=5/3];
\node at (8+4/3,0) {$p^{(\texttt{n})}_i$};
\node at (10+10/3,0) {$O.p^{(\texttt{n})}_j$};
\draw[fill] (8+8/3,0) circle [radius=0.05];
\draw[fill] (10+14/3,0) circle [radius=0.05];
\node at (7.2,3.4/3) {after};
\draw[dashed, ->] (4 + 8/3,0) -- (5+8/3,0);
\node at (7.2,2.2/3) {$180^\circ$};
\node at (7.2,1/3) {rotation};
\end{tikzpicture}
\caption{\label{Fig:OverlapContinuousSetup180} Illustration of two overlapping patches under the continuous setup.}
\end{figure}
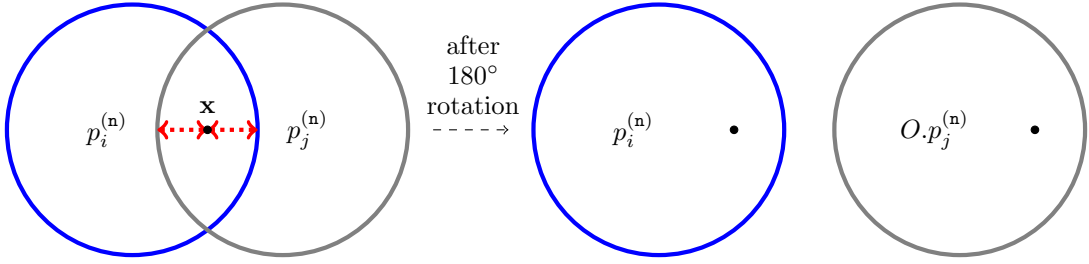

\begin{figure}[h] 
\begin{tikzpicture}
\draw[step=1cm,color=blue, ultra thick] (1,0) grid (4,3);
\draw[step=1cm,color=gray, ultra thick] (3,0) grid (6,3);
\node at (.5, 0) { {\color{blue}$P^{(\texttt{n})}_i$ }};
\node at (6.5, 0) {$P^{(\texttt{n})}_j$ };
\node at (3.5,2.5) {$\mathbf{x}_7$};
\node at (3.5,1.5) {$\mathbf{x}_8$};
\node at (3.5,0.5) {$\mathbf{x}_9$};
\draw[step=1cm,color=blue,  ultra thick] (8,0) grid (11,3);
\draw[step=1cm,color=gray,  ultra thick] (12,0) grid (15,3);
\node at (7.5, 0) { {\color{blue}$P^{(\texttt{n})}_i$ }};
\node at (15.7, 0) {$O.P^{(\texttt{n})}_j$ };
\node at (10.5,2.5) {$\mathbf{x}_7$};
\node at (10.5,1.5) {$\mathbf{x}_8$};
\node at (10.5,0.5) {$\mathbf{x}_9$};
\node at (14.5,2.5) {$\mathbf{x}_9$};
\node at (14.5,1.5) {$\mathbf{x}_8$};
\node at (14.5,0.5) {$\mathbf{x}_7$};
\node at (7.1,2.6) {after};
\draw[dashed, ->] (3.6 + 8/3,1.5) -- (5.1+8/3,1.5);
\node at (7.1,2.2) {$180^\circ$};
\node at (7.1,1.8) {rotation};
\end{tikzpicture}
\caption{\label{Fig:OverlapDiscreteSetup180} Illustration of the set $K_{\texttt{I}}(O)$ of two overlapping patches under the discrete setup. In this example there are three overlapped pixels. Clearly, $K_{\texttt{I}}(O) = \{8\}$ since $P^{(\texttt{n})}_i(8) = O.P^{(\texttt{n})}_j(8).$}
\end{figure}
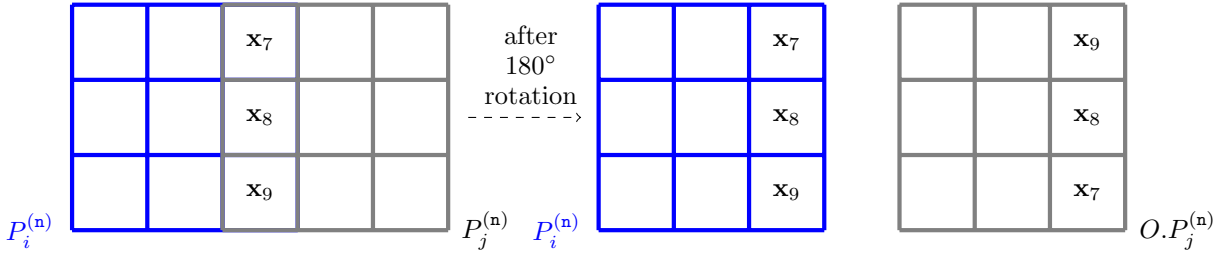

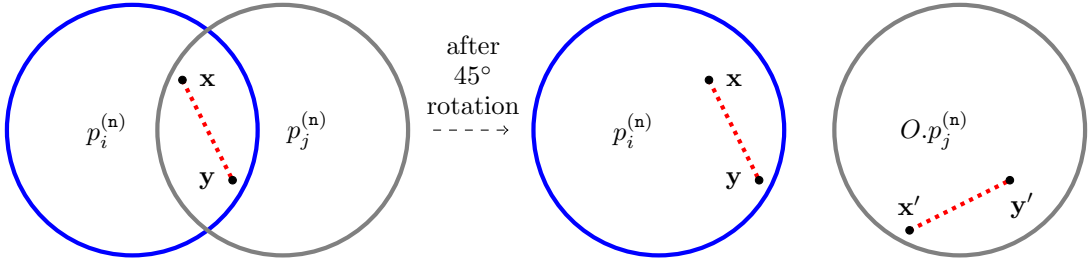
\begin{figure}[h]
\begin{tikzpicture}
\draw[blue, ultra thick] (1+5/3,0) circle [radius=5/3];
\draw[gray, ultra thick] (1+11/3,0) circle [radius=5/3];
\draw[-, red, ultra thick, dotted] (10/3,2/3) -- (12/3,-2/3);
\draw[fill] (10/3,2/3) circle [radius=0.05];
\draw[fill] (12/3,-2/3) circle [radius=0.05];
\node at (1+4/3,0) {$p^{(\texttt{n})}_i$};
\node at (1+12/3,0) {$p^{(\texttt{n})}_j$};
\node at (11/3, 2/3) {$\mathbf{x}$};
\node at (11/3, -2/3) {$\mathbf{y}$};
\draw[blue, ultra thick, shift={(7,0)}] (1+5/3,0) circle [radius=5/3];
\draw[-, red, ultra thick, dotted, shift={(7,0)}] (10/3,2/3) -- (12/3,-2/3);
\draw[fill, shift={(7,0)}] (10/3,2/3) circle [radius=0.05];
\draw[fill, shift={(7,0)}] (12/3,-2/3) circle [radius=0.05];
\node at (7+11/3, 2/3) {$\mathbf{x}$};
\node at (7+11/3, -2/3) {$\mathbf{y}$};
\draw[gray, ultra thick, shift={(9,0)}] (1+11/3,0) circle [radius=5/3];
\draw[-, red, ultra thick, dotted,shift={(9,0)}, rotate around={90:(14/3,0)}] (10/3,2/3) -- (12/3,-2/3);
\draw[fill, shift={(9,0)}, rotate around={90:(14/3,0)}] (10/3,2/3) circle [radius=0.05];
\draw[fill, shift={(9,0)}, rotate around={90:(14/3,0)}] (12/3,-2/3) circle [radius=0.05];
\node at (7.2,3.4/3) {after};
\draw[dashed, ->] (4 + 8/3,0) -- (5+8/3,0);
\node at (7.2,2.2/3) {$45^\circ$};
\node at (7.2,1/3) {rotation};
\node at (13, -1) {$\mathbf{x}'$};
\node at (14.5, -1) {$\mathbf{y}'$};
\node at (8+4/3,0) {$p^{(\texttt{n})}_i$};
\node at (10+10/3,0) {$O.p^{(\texttt{n})}_j$};\end{tikzpicture}
\caption{Illustration of the set $K_{\texttt{S}}(O)$ two overlapping patches under the continuous setup. In this case, $K_{\texttt{S}}(O)$ is empty. }
\label{Fig:OverlapContinuousSetup90}
\end{figure}

\begin{figure}[h]
\begin{tikzpicture}
\draw[step=1cm,color=blue, ultra thick] (1,0) grid (4,3);
\draw[step=1cm,color=gray, ultra thick] (2,0) grid (5,3);
\node at (.5, 0) { {\color{blue}$P^{(\texttt{n})}_i$ }};
\node at (5.5, 0) {$P^{(\texttt{n})}_j$ };
\node at (2.5,2.5) {$\mathbf{x}_4$};
\node at (2.5,1.5) {$\mathbf{x}_5$};
\node at (2.5,0.5) {$\mathbf{x}_6$};
\node at (3.5,2.5) {$\mathbf{x}_7$};
\node at (3.5,1.5) {$\mathbf{x}_8$};
\node at (3.5,0.5) {$\mathbf{x}_9$};
\draw[step=1cm,color=blue,  ultra thick] (8,0) grid (11,3);
\draw[step=1cm,color=gray,  ultra thick] (12,0) grid (15,3);
\node at (7.5, 0) { {\color{blue}$P^{(\texttt{n})}_i$ }};
\node at (15.7, 0) {$O.P^{(\texttt{n})}_j$ };
\node at (9.5,2.5) {$\mathbf{x}_4$};
\node at (9.5,1.5) {$\mathbf{x}_5$};
\node at (9.5,0.5) {$\mathbf{x}_6$};
\node at (10.5,2.5) {$\mathbf{x}_7$};
\node at (10.5,1.5) {$\mathbf{x}_8$};
\node at (10.5,0.5) {$\mathbf{x}_9$};
\node at (13.5,2.5) {$\mathbf{x}_9$};
\node at (13.5,1.5) {$\mathbf{x}_8$};
\node at (13.5,0.5) {$\mathbf{x}_7$};
\node at (14.5,2.5) {$\mathbf{x}_6$};
\node at (14.5,1.5) {$\mathbf{x}_5$};
\node at (14.5,0.5) {$\mathbf{x}_4$};
\node at (6.5,2.2) {after};
\node at (6.5,1.8) {$180^\circ$ rotation};
\draw[dashed, ->] (5.5,1.5) -- (7.5,1.5);
\end{tikzpicture}
\caption{\label{Fig:OverlapDiscreteSetupS180} Illustration of the set $K_{\texttt{S}}(O)$ of two overlapping patches under the discrete setup. In this example there are six overlapped pixels. By definition, we have $K_{\texttt{S}}(O) = \{ (4,9), (9,4), (5,8), (8,5), (6,7), (7,6)\}$ and $|K_{\texttt{S}}(O)|=6$. }
\end{figure}
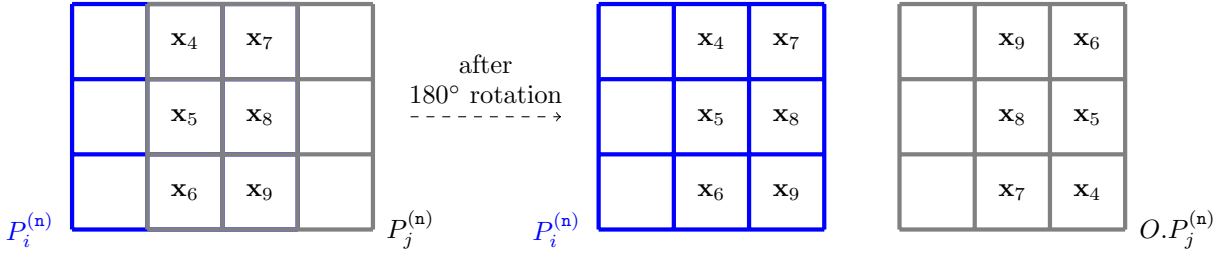

\begin{lem}
\label{lemma:patche_stat}
Fix $O \in SO(2)$. Take two patches $P^{(\texttt{n})}_i,P^{(\texttt{n})}_j\in\mathcal{X}_I^{(\texttt{n})}\subset \mathbb{R}^{q^2}$, where $\mathcal{X}_I^{(\texttt{n})}$ is defined in (\ref{Definition:DiscretizationNoisyImagePatchSpace}). 
Then, we have
\be
\mathbb{E}(\| P^{(\texttt{n})}_i - O.P^{(\texttt{n})}_j \|^2 ) = \| P \|^2 + 2 \sigma^2 (q^2 - |K_{\texttt{I}}(O)|)
\ee
and
\begin{align}
\mathrm{Var}(\| P^{(\texttt{n})}_i - O.P^{(\texttt{n})}_j \|^2 ) = &\, 8\sigma^2 \left(   \| P \|^2 -  \sum_{(a,b) \in K_{\texttt{O}}(O)} P(a)P(b) + \sum_{a \in K_{\texttt{I}}(O)} P(a)^2 \right)  \nonumber \\
&\quad +  4\sigma^4 \left( 2q^2  + |K_{\texttt{S}}(O)| + |K_{\texttt{O}}(O)| -3 |K_{\texttt{I}}(O)|\right)\,,
\end{align}
where $P := P_i^{(\texttt{c})} - O.P_j^{(\texttt{c})}$. Particularly, when $O$ is the identity, we have
\be
\mathbb{E}(\| P^{(\texttt{n})}_i - P^{(\texttt{n})}_j \|^2 ) = \| P \|^2 + 2 \sigma^2 q^2
\ee
and
\begin{align}
\mathrm{Var}(\| P^{(\texttt{n})}_i - P^{(\texttt{n})}_j \|^2 ) = &\, 8\sigma^2 \left(   \| P \|^2 -  \sum_{(a,b) \in K_{\texttt{O}}} P(a)P(b) \right) +  4\sigma^4 \left( 2q^2 + |K_{\texttt{O}}|\right)\,.
\end{align}
\end{lem}

\begin{rmk}
Before proving the Lemma, we have some comments. First, since $2 \sigma^2 (q^2 - |K_{\texttt{I}}(O)|)>0$, $\| P^{(\texttt{n})}_i - O.P^{(\texttt{n})}_j \|^2$ is a biased estimator of $\| P^{(\texttt{c})}_i - O.P^{(\texttt{c})}_j \|^2$. Second, by the lemma, if $P^{(\texttt{n})}_i,P^{(\texttt{n})}_j\in\mathcal{X}_I^{(\texttt{n})}\subset \mathbb{R}^{q^2}$ do not overlap, we have
\begin{align}
\mathbb{E}(\| P^{(\texttt{n})}_i - O.P^{(\texttt{n})}_j \|^2 ) = \| P \|^2 + 2 \sigma^2 q^2 \mbox{ and }
\mathrm{Var}(\| P^{(\texttt{n})}_i - O.P^{(\texttt{n})}_j \|^2 ) = 8\sigma^2\| P \|^2  +  8\sigma^4 q^2\,.
\end{align}
Third, when $O$ is the identity, the Lemma could be applied to study the NLEM. 
Finally, since the sign of $\sum_{(a,b) \in K_{\texttt{O}}} P(a)P(b)$ is not controlled, this term leads to the complicated behavior of the $L^2$ distance or RID when two patches overlap. Indeed, when two patches overlap, depending on the clean patches' structure, the RID estimator from the noisy patches might be biased toward overlapped patches. Thus, if the overlapped patches are included in the denoising process, the search of the nearest neighboring patches might be biased to the ``local patches'' that have overlaps.
\end{rmk}

\begin{proof}
Since $i,j$ are fixed, to simplify the notation, we write 
\be
P^{(\texttt{n})}_i(a) - O.P^{(\texttt{n})}_j (a)
 = P^{(\texttt{c})}_i(a) -O.P^{(\texttt{c})}_j(a) + \sigma (\xi_i - O.\xi_j) \nonumber 
 = P(a) + \sigma(\xi(a)- \xi'(a))\,,
\ee
where $\xi':=O.\xi_j$. 
Hence, we can express the $\ell^2$ norm of $P^{(\texttt{n})}_i - O.P^{(\texttt{n})}_j$ as 
\begin{align}
\|P^{(\texttt{n})}_i - O.P^{(\texttt{n})}_j \|^2 
&= \| P\|^2 + 2 \sigma P^T (\xi - \xi') + \sigma^2 \| \xi - \xi'\|^2 \nonumber \\
&= \| P\|^2 + 2 \sigma P^T (\xi - \xi') + \sigma^2 \left(\| \xi \|^2 + \| \xi' \|^2 - 2 \xi^T \xi' \right)\,.\label{Proof:Lemma1:Eq1}
\end{align}
By the assumption, $\xi(a), a = 1,\ldots,q^2$ are i.i.d. Gaussian random variables and as well as $\xi'(a), a = 1,\ldots,q^2$. We also know that $\xi(a)$ and $\xi'(a)$ are independent Gaussian random variables when $a\notin K_{\texttt{I}}(O)$. By a direct calculation, we have $\mathbb{E}(\xi^T \xi')= |K_{\texttt{I}}(O)|$. Combining this with the facts that 
\be
\mathbb{E}(\|\xi\|^2)=\mathbb{E}(\|\xi'\|^2)=q^2 \mbox{ and } \mathbb{E}(\xi(a) - \xi'(a) )=0\,,
\ee
we obtain 
\be
\mathbb{E}(\|P^{(\texttt{n})}_i - O.P^{(\texttt{n})}_j \|^2) = \| P \|^2 + 2 \sigma^2 (q^2 - |K_{\texttt{I}}(O)|).
\ee

To compute the variance, we write $\mathrm{Var}(\| P^{(\texttt{n})}_i - O.P^{(\texttt{n})}_j \|^2 )$ as the following by expanding (\ref{Proof:Lemma1:Eq1}):
\begin{align}
&\mathrm{Var}(\| P^{(\texttt{n})}_i - O.P^{(\texttt{n})}_j \|^2 ) \\
=\,& \mathrm{Var}\left( 2 \sigma \sum_{a=1}^{q^2} P(a)(\xi(a) - \xi'(a))\right) + \mathrm{Var}\left( \sigma^2 \sum_{a=1}^{q^2} (\xi(a) - \xi'(a))^2\right) \nonumber \\
& \quad + 2 \mathrm{Cov}\left( 2\sigma \sum_{a=1}^{q^2} P(a)(\xi(a) - \xi'(a)) , \sigma^2 \sum_{a=1}^{q^2} (\xi(a) - \xi'(a))^2\right) \nonumber \\
=\,& 4\sigma^2 \cdot (\texttt{I}) + \sigma^4 \cdot (\texttt{II}) + 4\sigma^3\cdot (\texttt{III}),\nonumber
\end{align}
where $(\texttt{I}):=\mathrm{Var}\left(\sum_{a=1}^{q^2} P(a)(\xi(a) - \xi'(a))\right)$, $(\texttt{II}):=\mathrm{Var}\left(\sum_{a=1}^{q^2} (\xi(a) - \xi'(a))^2\right)$, and \\$(\texttt{III}):=\mathrm{Cov}\left(\sum_{a=1}^{q^2} P(a)(\xi(a) - \xi'(a)) , \sum_{a=1}^{q^2} (\xi(a) - \xi'(a))^2\right)$.
We compute $(\texttt{I}), (\texttt{II})$, and $(\texttt{III})$ below.
\begin{align}
(\texttt{I}) 
&= \mathrm{Var}\left( \sum_{a=1}^{q^2} P(a)(\xi(a) - \xi'(a))\right) \nonumber \\
&= \sum_{a=1}^{q^2} P^2(a) \mathrm{Var}(\xi(a) - \xi'(a)) + \sum_{a,b \in \{1,\cdots, q^2\}, a\neq b} P(a)P(b)\mathrm{Cov}\left( \xi(a) - \xi'(a), \xi(b) - \xi'(b)\right) \nonumber \\
&= 2 \| P \|^2 - 2 \sum_{a,b \in \{1, \cdots, q^2\}, a\neq b} P(a)P(b) \mathrm{Cov}(\xi(a),\xi'(b))
\end{align}
where the second equality comes from a direct expansion, and the last equality holds since
\begin{align}
\mathrm{Cov}\left( \xi(a) - \xi'(a), \xi(b) - \xi'(b)\right) 
=  - \left[ \mathrm{Cov}\left( \xi(a), \xi'(b) \right) + \mathrm{Cov}\left( \xi(b), \xi'(a) \right) \right]\,,
\end{align}
which comes from the fact that $\{\xi(a)\}_{a=1}^{q^2}$ are independent and  $\{\xi'(a)\}_{a=1}^{q^2}$ are independent.
Since only overlapped pixels lead to non-zero $\mathrm{Cov}(\xi(a),\xi'(b))$, we have
\be
\sum_{a,b \in \{1, \cdots, q^2\}, a\neq b} P(a)P(b)  \mathrm{Cov}(\xi(a),\xi'(b)) =  \sum_{(a,b) \in K_{\texttt{O}}(O)} P(a)P(b) - \sum_{a \in K_{\texttt{I}}(O)} P(a)^2 \,,
\ee
where we subtract $\sum_{a \in K_{\texttt{I}}(O)} P(a)^2$ since $a\neq b$. As a result,
\be
\label{eq:I}
(\texttt{I}) =  2 \| P \|^2 - 2 \left( \sum_{(a,b) \in K_{\texttt{O}}(O)} P(a)P(b) - \sum_{a \in K_{\texttt{I}}(O)} P(a)^2 \right).
\ee
Next,
\begin{align}
\label{eq:II_0}
(\texttt{II}) 
&= \mathrm{Var}\left(  \sum_{a=1}^{q^2} (\xi(a) - \xi'(a))^2\right) = \mathrm{Var}\left( \|\xi\|^2+\|\xi'\|^2-2\xi^T\xi'\right) \nonumber \\
&= \mathrm{Var}(\| \xi \|^2) + \mathrm{Var}(\| \xi' \|^2) + 4 \mathrm{Var}\left(\sum_{a=1}^{q^2} \xi(a) \xi'(a)\right) +2 \mathrm{Cov}\left(\sum_{a=1}^{q^2} \xi(a)^2, \sum_{b=1}^{q^2} \xi'(b)^2\right) \nonumber\\
&\quad 
- 4 \mathrm{Cov}\left(\sum_{a=1}^{q^2} \xi(a)^2, \sum_{b=1}^{q^2} \xi(b)\xi'(b)\right)
- 4 \mathrm{Cov}\left(\sum_{a=1}^{q^2} \xi'(a)^2, \sum_{b=1}^{q^2} \xi(b)\xi'(b)\right)\,.
\end{align}
We now calculate $(\texttt{II})$ term by term. By a direct expansion, we have
\begin{align}
\label{eq:II_1}
\mathrm{Var}\left(\sum_{a=1}^{q^2} \xi(a) \xi'(a)\right)
&= \sum_{a=1}^{q^2} \mathrm{Var}(\xi(a) \xi'(a)) +  \sum_{a,b \in \{1,\cdots, q^2\}, a \neq b} \mathrm{Cov}\left(\xi(a)\xi'(a), \xi(b)\xi'(b) \right) \nonumber\\
&= \Big[\sum_{a \notin K_{\texttt{I}}(O)} \mathrm{Var}(\xi(a)) \mathrm{Var}(\xi'(a)) + \sum_{a \in K_{\texttt{I}}(O)} \mathrm{Var}(\xi(a)^2)\Big] +\sum_{(a,b) \in K_{\texttt{S}}(O)} \mathrm{Var}(\xi(a)\xi(b))  \nonumber \\
&=[(q^2-|K_{\texttt{I}}(O)|)+2|K_{\texttt{I}}(O)|]+|K_{\texttt{S}}(O)|\nonumber\\
&= q^2 + |K_{\texttt{I}}(O)| + |K_{\texttt{S}}(O)|,  
\end{align}
where the second equality holds since $\mathrm{Cov}\left(\xi(a)\xi'(a), \xi(b)\xi'(b) \right)\neq 0$ only when $(a,b)\in K_{\texttt{S}}(O)$ and the third equality holds since 
$\mathrm{Var}(\xi(a)^2)=2$ and $\mathrm{Var}(\xi(a)\xi(b))=\mathbb{E}\xi(a)^2\mathbb{E}\xi(b)^2=1$ due to the independence.
Similarly, we have by a direct calculation
\begin{align}
\label{eq:II_2}
\mathrm{Cov}\left(\sum_{a=1}^{q^2} \xi(a)^2, \sum_{b=1}^{q^2} \xi'(b)^2\right)
= \sum_{(a,b) \in K_{\texttt{O}}(O)} \mathrm{Cov}\left(\xi(a)^2, \xi'(b)^2\right) = 2 |K_{\texttt{O}}(O)|\,.
\end{align}
By the linearity of the covariance, 
\begin{align}
&\mathrm{Cov}\left(\sum_{a=1}^{q^2} \xi(a)^2, \sum_{b=1}^{q^2} \xi(b)\xi'(b)\right) \label{eq:II_3}\\
=&\, \sum_{a = 1}^{q^2} \mathrm{Cov}( \xi(a)^2,\xi(a)\xi'(a) )+\sum_{a,b\in \{1,\ldots,q^2\},a\neq b} \mathrm{Cov}( \xi(a)^2,\xi(b)\xi'(b) )\nonumber\\
=&\, 2|K_{\texttt{I}}(O)|+\sum_{a\neq b,\,b\in K_{\texttt{I}}(O)} \mathrm{Cov}( \xi(a)^2,\xi(b)\xi'(b) )+\sum_{a\neq b,\,b\notin K_{\texttt{I}}(O)} \mathrm{Cov}( \xi(a)^2,\xi(b)\xi'(b) )=2|K_{\texttt{I}}(O)|\nonumber\,,
\end{align}
where 
\be
\sum_{a\neq b,\,b\in K_{\texttt{I}}(O)} \mathrm{Cov}( \xi(a)^2,\xi(b)\xi'(b) )=\sum_{a\neq b,\,b\in K_{\texttt{I}}(O)} \mathrm{Cov}( \xi(a)^2,\xi(b)^2)=0
\ee
 and 
 \be
 \sum_{a\neq b,\,b\notin K_{\texttt{I}}(O)} \mathrm{Cov}( \xi(a)^2,\xi(b)\xi'(b) )=0,
 \ee
 since $\mathbb{E}(X) = \mathbb{E}(X^3) =0$ for $X \sim \mathcal{N}(0,1)$. 
 To be more precise, when $a\neq b$ and $b\in K_{\texttt{I}}(O)$, $\mathrm{Cov}( \xi(a)^2,\xi(b)^2)=0$ due to the independence assumption; when $a\neq b$ and $b\notin K_{\texttt{I}}(O)$, $\xi(b)$ and $\xi'(b)$ are independent, so 
 \be
 \mathrm{Cov}( \xi(a)^2,\xi(b)\xi'(b) )= \mathbb{E}(\xi(a)^2\xi(b)\xi'(b))-\mathbb{E}\xi(a)^2\mathbb{E}(\xi(b)\xi'(b))=\mathbb{E}(\xi(a)^2\xi(b)\xi'(b)),
 \ee
 which is $0$ no matter $\xi(a)^2$ is independent of $\xi(b)\xi'(b)$ or not. Note that by our assumption, $\xi(a)$ is dependent on $\xi(b)$ (or $\xi'(b)$) if and only if  $\xi(a)$ is the same as $\xi(b)$ (or $\xi'(b)$).
 
Similarly, we have
\be
\label{eq:II_4}
\mathrm{Cov}\left(\sum_{a=1}^{q^2} \xi'(a)^2, \sum_{b=1}^{q^2} \xi(b)\xi'(b)\right) = 2|K_{\texttt{I}}(O)|.
\ee
By
substituting (\ref{eq:II_1}), (\ref{eq:II_2}), (\ref{eq:II_3}), and (\ref{eq:II_4}) into (\ref{eq:II_0}), we obtain
\be
\label{eq:II}
(\texttt{II}) = 8q^2  + 4(|K_{\texttt{S}}(O)| + |K_{\texttt{O}}(O)| -3 |K_{\texttt{I}}(O)|).
\ee
Finally, we have
\begin{align}
\label{eq:III}
(\texttt{III}) =\, &\mathrm{Cov}\left(  \sum_{a=1}^{q^2} P(a)(\xi(a) - \xi'(a)) , \sum_{a=1}^{q^2} (\xi(a) - \xi'(a))^2\right) = \sum_{a,b=1}^{q^2} P(a)\mathbb{E}(\xi(a) - \xi'(a))(\xi(b) - \xi'(b))^2 \nonumber\\
=\, & \sum_{a,b=1}^{q^2} P(a)\mathbb{E}[\xi(a)\xi(b)^2-2\xi(a)\xi(b)\xi'(b)+\xi(a)\xi'(b)^2-\xi'(a)\xi(b)^2+2\xi'(a)\xi(b)\xi'(b)-\xi'(a)\xi'(b)^2] = 0,
\end{align}
since $\mathrm{Cov}(X,Y^2) = 0$ for independent $X,Y \sim \mathcal{N}(0,1)$ and $\mathrm{Cov}(X,YZ) = 0$ for independent $X,Y,Z\sim \mathcal{N}(0,1)$.

Combining equations (\ref{eq:I}), (\ref{eq:II}), and (\ref{eq:III}), we conclude that
\begin{align}
&\mathrm{Var}(\| P^{(\texttt{n})}_i - O.P^{(\texttt{n})}_j \|^2 )\\
 = &\,8\sigma^2 \left(   \| P \|^2 -  \sum_{(a,b) \in K_{\texttt{O}}(O)} P(a)P(b) + \sum_{a \in K_{\texttt{I}}(O)} P(a)^2 \right) +  4\sigma^4 \left( 2q^2  + |K_{\texttt{S}}(O)| + |K_{\texttt{O}}(O)| -3 |K_{\texttt{I}}(O)|\right).\nonumber
\end{align}
\end{proof}

With this Lemma, we are ready to show that the RID between two clean patches could be well approximated by noisy patches; particularly, if two clean patches are close enough, their RID can be well approximated by the associated noisy patches.  

\begin{thm}\label{Theorem:RIDNeighborsEstimate}
Take two patches $P^{(\texttt{n})}_i,P^{(\texttt{n})}_j\in\mathcal{X}_I^{(\texttt{n})}\subset \mathbb{R}^{q^2}$. Suppose that $d_{\texttt{RID}} (P_i^{(\texttt{c})}, P_j^{(\texttt{c})}) < \epsilon$ and that $ \sigma q <  \epsilon$.
Then
\begin{align}
\label{Theorem1:RIDbound:1}
\mathrm{Pr}\left( d_{\texttt{RID}} (P_i^{(\texttt{n})}, P_j^{(\texttt{n})})  < 2 \epsilon \right) 
>  
\left(1 + 8\sigma^2\frac{2   \epsilon^2 +  \sigma^2 ( 2q^2-q )}{\left(3\epsilon^2  - 2 \sigma^2 (q^2 - 1) \right)^2} 
\right)^{-1}\,,
\end{align}
which increases when $q$ decreases. 
  
Suppose that $d_{\texttt{RID}} (P_i^{(\texttt{c})}, P_j^{(\texttt{c})}) > 2\epsilon$. Then
\be
\mathrm{Pr}\left(d_{\texttt{RID}} (P_i^{(\texttt{n})}, P_j^{(\texttt{n})}) >  \epsilon \right) >
\left( 1 + 8 \sigma^2\frac{2 \| P \|^2  +  \sigma^2  ( 2q^2 -q)}{\left(3\|P\|^2/4+ 2 \sigma^2 (q^2 - 1)  \right)^2}\right)^{-1}\,,
\label{Theorem1:RIDbound:2}
\ee
which increases when $q$ increases. 

In particular, when the patches $P^{(\texttt{n})}_i,P^{(\texttt{n})}_j\in\mathcal{X}_I^{(\texttt{n})}\subset \mathbb{R}^{q^2}$ have no overlap, we have
if $d_{\texttt{RID}} (P_i^{(\texttt{c})}, P_j^{(\texttt{c})}) < \epsilon$,
then
\be
\label{noOL:bound1}
\mathrm{Pr}\left( d_{\texttt{RID}} (P_i^{(\texttt{n})}, P_j^{(\texttt{n})})  < 2 \epsilon \right) >
\left( 1 + 8\sigma^2 \frac{\epsilon^2 + \sigma^2 q^2}{(3\epsilon^2 - 2 \sigma^2 q^2)^2} \right)^{-1};
\ee 
if $d_{\texttt{RID}} (P_i^{(\texttt{c})}, P_j^{(\texttt{c})}) > \epsilon$, then
\be
\mathrm{Pr}\left(d_{\texttt{RID}} (P_i^{(\texttt{n})}, P_j^{(\texttt{n})}) >  \epsilon \right) 
> \left( 1 + 8 \sigma^2\frac{ \| P \|^2  +  \sigma^2 q^2}{\left(3\|P\|^2/4+ 2 \sigma^2 q^2  \right)^2}\right)^{-1}.
\label{noOL:bound2}
\ee
\end{thm}

\begin{rmk}
When two patches $P^{(\texttt{n})}_i$ and $P^{(\texttt{n})}_j$ are disjoint, the bound (\ref{noOL:bound1}) suggests that 
the smaller patch size is better. However, the bound (\ref{noOL:bound2}) suggests the opposite. We thus need to choose a suitable patch size $q$ that balances (\ref{noOL:bound1}) and (\ref{noOL:bound2}).

When two patches $P^{(\texttt{n})}_i$ and $P^{(\texttt{n})}_j$ overlap, the analysis of choosing the patch size becomes very complicated.
It would depend on the image, the minimiser rotation $O$, and how the patches overlap. See (\ref{ineq:complexity}) for example. The worst bounds (\ref{Theorem1:RIDbound:1}) and (\ref{noOL:bound1}) we have in Theorem \ref{Theorem:RIDNeighborsEstimate} also suggest that $q$ should not be too small but also not too large.  
In practice, we found that an odd value $q$ between $7$ and $15$ leads to a good performance, but the optimal $q$ depends on the image. 

We mention that if $O$ is the identity in Theorem \ref{Theorem:RIDNeighborsEstimate}, the same argument explains why the $L^2$ distance between two clean patches could be well approximated by noisy patches, and hence better understand the NLEM algorithm.
\end{rmk}

\begin{proof} Suppose $d_{\texttt{RID}} (P_i^{(\texttt{c})}, P_j^{(\texttt{c})}) < \epsilon$ and $\sigma q < \epsilon$. Suppose $O \in SO(2)$ is such that $d_{\texttt{RID}} (P_i^{(\texttt{c})}, P_j^{(\texttt{c})}) = \| P_i^{(\texttt{c})} - O.P_j^{(\texttt{c})} \| = \|P \|^2$. 
We start by preparing a bound. Applying Lemma \ref{lemma:patche_stat} and the bounds of $|K_{\texttt{S}}(O)| $ and $|K_{\texttt{I}}(O)|$, we have
\begin{align}
&\frac{\mathrm{Var} (\| P^{(\texttt{n})}_i - O.P^{(\texttt{n})}_j \|)}{\left(4\epsilon^2 -  \mathbb{E} (\| P^{(\texttt{n})}_i - O.P^{(\texttt{n})}_j \|) \right)^2} \\
=&\, \frac{\displaystyle 8\sigma^2 \left(   \| P \|^2 -  \sum_{(a,b) \in K_{\texttt{O}}, a,b \notin K_{\texttt{I}}} P(a)P(b)  \right) +  4\sigma^4 \left( 2q^2  + |K_{\texttt{O}}| + |K_{\texttt{S}}(O)| -3 |K_{\texttt{I}}(O)|\right)}{\left(4\epsilon^2 - \| P \|^2 - 2 \sigma^2 (q^2 - |K_I|) \right)^2} \nonumber\\
\leq&\,
\frac{16\sigma^2   \epsilon^2 +  8\sigma^4 ( 2q^2-q )}{\left(3\epsilon^2  - 2 \sigma^2 (q^2 - 1) \right)^2} \nonumber\,,
\end{align}
since 
\be
\left| \sum_{(a,b) \in K_{\texttt{O}}} P(a)P(b) \right| \leq \|P\|^2, \,|K_{\texttt{O}}| \leq q(q-1),\, |K_{\texttt{S}}| \leq q(q-1), \mbox{ and } |K_{\texttt{I}}| \leq 1.\label{Proof:Theorem:BoundsForOverlap}
\ee

Recall the one-sided Chebychev's inequality for a random variable $X$ with a finite second moment: $\mathrm{Pr}(X\geq \mathbb{E}X+a)\leq \frac{\mathrm{Var}(X)}{\mathrm{Var}(X)+a^2}$, for $a >0$. Applying the one-sided Chebyshev's inequality and the inequality above, we obtain
\begin{align}
&\mathrm{Pr}\left( d_{\texttt{RID}} (P_i^{(\texttt{n})}, P_j^{(\texttt{n})} ) < 2\epsilon \right) 
\geq  \mathrm{Pr}\left( \| P^{(\texttt{n})}_i - O.P^{(\texttt{n})}_j \|^2 < (2\epsilon)^2 \right) \\
 > &\, \left(1+ \frac{\mathrm{Var} (\| P^{(\texttt{n})}_i - O.P^{(\texttt{n})}_j \|)}{\left(4\epsilon^2 -  \mathbb{E} (\| P^{(\texttt{n})}_i - O.P^{(\texttt{n})}_j \|) \right)^2}  \right)^{-1} 
\geq
\left(1 + \frac{16\sigma^2   \epsilon^2 +  8\sigma^4 ( 2q^2-q )}{\left(3\epsilon^2  - 2 \sigma^2 (q^2 - 1) \right)^2} .\nonumber
\right)^{-1}.
\end{align}
When patches $P^{(\texttt{n})}_i$ and $P^{(\texttt{n})}_j$ do not overlap,
\be
\frac{\mathrm{Var} (\| P^{(\texttt{n})}_i - O.P^{(\texttt{n})}_j \|)}{\left(4\epsilon^2 -  \mathbb{E} (\| P^{(\texttt{n})}_i - O.P^{(\texttt{n})}_j \|) \right)^2}
= \frac{8\sigma^2 \|P\|^2 + 8\sigma^4q^2}{(4 \epsilon^2 - \|P\|^2 - 2\sigma^2 q^2)^2}\,,
\ee
which implies (\ref{noOL:bound1}).

Now, suppose $d_{\texttt{RID}} (P_i^{(\texttt{c})}, P_j^{(\texttt{c})}) > 2 \epsilon$. For any $O \in SO(2)$, we apply the assumption $d_{\texttt{RID}} (P_i^{(\texttt{c})}, P_j^{(\texttt{c})}) > 2 \epsilon$ and Lemma \ref{lemma:patche_stat}. We obtain
\begin{align}
\label{ineq:complexity}
&\frac{\mathrm{Var} (\| P^{(\texttt{n})}_i - O.P^{(\texttt{n})}_j \|)}{\left(\mathbb{E} (\| P^{(\texttt{n})}_i - O.P^{(\texttt{n})}_j \|) - \epsilon^2 \right)^2  } \\
\leq&\, \frac{\displaystyle 8\sigma^2 \left(   \| P \|^2 -  \sum_{(a,b) \in K_{\texttt{O}}, a,b \notin K_{\texttt{I}}} P(a)P(b)  \right)   +  4\sigma^4 \left( 2q^2  + |K_{\texttt{O}}| + |K_{\texttt{S}}| -3 |K_{\texttt{I}}|\right)}{\left(3\|P\|^2/4+ 2 \sigma^2 (q^2 - |K_I|)  \right)^2}, \nonumber
\end{align}
where $P = P^{(\texttt{c})}_i - O.P^{(\texttt{c})}_j$.
Due to the bounds shown in (\ref{Proof:Theorem:BoundsForOverlap}), we can further obtain the following bound which is independent of the rotation $O$:
\be
\frac{\mathrm{Var} (\| P^{(\texttt{n})}_i - O.P^{(\texttt{n})}_j \|)}{\left(\mathbb{E} (\| P^{(\texttt{n})}_i - O.P^{(\texttt{n})}_j \|) - \epsilon^2 \right)^2  } \\
\leq 
\frac{16\sigma^2  \| P \|^2    +  8\sigma^4 ( 2q^2 -q)}{\left(3\|P\|^2/4+ 2 \sigma^2 (q^2 - 1)  \right)^2}. \nonumber
\ee 
Applying the one-sided Chebyshev's inequality, we can obtain a universal lower bound
\begin{align}
\mathrm{Pr}\left( \| P^{(\texttt{n})}_i - O.P^{(\texttt{n})}_j  \|^2 > \epsilon^2 \right) 
&> \left(  1 + \frac{\mathrm{Var} (\| P^{(\texttt{n})}_i - O.P^{(\texttt{n})}_j \|)}{\left(\mathbb{E} (\| P^{(\texttt{n})}_i - O.P^{(\texttt{n})}_j \|) - \epsilon^2 \right)^2  }  \right)^{-1}\label{Bound:Theorem:OverlapLowerBound} \\
&\geq \left( 1 + 4 \sigma^2\frac{4 \| P \|^2  +  \sigma^2  \left( 4q^2 -2q\right)}{\left(3\|P\|^2/4+ 2 \sigma^2 (q^2 - 1)  \right)^2}\right)^{-1}.\nonumber
\end{align}
Since the lower bound (\ref{Bound:Theorem:OverlapLowerBound}) holds for any rotation $O$, we therefore have
\be
\mathrm{Pr}\left( d_{\texttt{RID}}(P^{(\texttt{n})}_i , P^{(\texttt{n})}_j )> \epsilon \right)
>  \left( 1 + 4 \sigma^2\frac{4 \| P \|^2  +  \sigma^2  \left( 4q^2 -2q\right)}{\left(3\|P\|^2/4+ 2 \sigma^2 (q^2 - 1)  \right)^2}\right)^{-1}.
\ee
When patches $P^{(\texttt{n})}_i$ and $P^{(\texttt{n})}_j$ do not overlap,
\be
\frac{\mathrm{Var} (\| P^{(\texttt{n})}_i - O.P^{(\texttt{n})}_j \|)}{\left(\mathbb{E} (\| P^{(\texttt{n})}_i - O.P^{(\texttt{n})}_j \|) - \epsilon^2 \right)^2  }
= \frac{8\sigma^2 \|P\|^2 + 8\sigma^4}{(\|P\|^2 - \epsilon^2 + 2\sigma^2q^2)^2}
\ee
which implies (\ref{noOL:bound2}).
\end{proof}

Some discussions are needed for this Theorem. 
The quantity $\sigma^2 q^2$ could be understood as the ``total energy'' of the added noise, and the condition $ \sigma^2 q^2<\epsilon^2$ means that the RID estimated from two noisy patches is controlled by the square root of the energy of the noise. With this energy viewpoint, we could apply the technique developed in \cite[Theorem 2.1]{ElKaroui_Wu:2015b}. However, we carried out the proof in the above way to emphasize the main purpose of the RID, and to find the true neighbors and the dependence on the patch size.

Second, we mention that when $d_{\texttt{RID}} (P_i^{(\texttt{c})}, P_j^{(\texttt{c})})<\epsilon$ and $P_i^{(\texttt{c})}$ is of size $q\times q$, then the difference of the central pixels of $P_i^{(\texttt{c})}$ and $P_j^{(\texttt{c})}$ is bounded by $\epsilon$ in the worst case. Indeed, we have
\be
d^2_{\texttt{RID}} (P_i^{(\texttt{c})}, P_j^{(\texttt{c})})=\|P_i^{(\texttt{c})}-O.P_j^{(\texttt{c})}\|^2=\sum_{a=1,a\neq c}^{q^2}|P_i^{(\texttt{c})}(a)-(O.P_j^{(\texttt{c})})(a)|^2+|P_i^{(\texttt{c})}(c)-P_j^{(\texttt{c})}(c)|^2<\epsilon^2,
\ee
where $O\in SO(2)$ is the rotation that achieves the RID, and the second equality holds since $(O.P_j^{(\texttt{c})})(c)=P_j^{(\texttt{c})}(c)$. In practice, $|P_i^{(\texttt{c})}(c)-P_j^{(\texttt{c})}(c)|$ could be smaller than $\epsilon$ when $\sum_{a=1,a\neq c}^{q^2}|P_i^{(\texttt{c})}(a)-(O.P_j^{(\texttt{c})})(a)|^2$ is not zero. When $O$ is the identity, the same argument holds for the NLEM algorithm, and this explains why we could have a better denoising result by using the patches.

\subsection{Orientation Assignment in SIFT}
\label{sec:SIFT}
SIFT is a method for extracting features that are invariant to image scale and rotation \cite{Lowe:2004}. The idea was originally from \cite{Burt_Adelson:1983,Crowley_Stern:1984,Lindeberg:2012}, and became popular after \cite{Lowe:2004}. For a given image, SIFT detects points of interest, called keypoints, under the scale-space model \cite{Lindeberg:1994,Lindeberg:1998}. It then assigns the \textit{orientation feature} at each keypoint. Since the centre point of each patch is our point of interest, we only use orientation feature assignments in the SIFT algorithm and skip the keypoint detection. For more details and different variations of SIFT, we refer the interested reader to \cite{Lindeberg:2012} for a review.
In this subsection, we give theoretical proofs to show why orientation assignments in SIFT is robust to noise and hence can be used to approximate rotation angles between patches.

\begin{defn}
Let $p:\RR^2 \rightarrow \RR$ be a round patch of an image $I: \RR^2 \rightarrow \RR$ defined in Definition \ref{Definition:ContinuousPatchSpace}. The Gaussian smoothed patch, denoted as $L_p$, is defined as the convolution of the Gaussian function $G(\mathbf{s},1)$ with the patch $p$
\be
L_p(\mathbf{x}) = G*p(\mathbf{x}) = \iint_{\RR^2} G(\mathbf{s},1)p(\mathbf{x-s}) d{\mathbf{s}},
\ee
where the Gaussian $G(\mathbf{s}, \ell) := \frac{1}{2\pi \ell^2} e^{-\frac{\| \mathbf{s}\|^2}{2 \ell^2}}$ and $\ell>0$ denotes the scale in the SIFT algorithm. 
\end{defn}
We only consider the case when $\ell=1$. Similar argument could be carried over for $\ell\neq 1$. From now on, we write $G(\mathbf{s},1)$ as $G(\mathbf{s})$ to simplify the notation. 
Suppose $L_p$ and $L_{O.p}$ denote the Gaussian smoothed patches of $p$ and its rotated patch $O.p$ by the angle $\phi$. 
Then since $O \in SO(2)$ and $G(\mathbf{s})$ is rotationally invariant, we have
\be
L_{O.p}(\mathbf{x}) = L_p(O^{-1}\mathbf{x}) 
\ee
and the orientation features of $p$ and $O.p$ will differ by the angle $\phi$ as well.

Denote $S^1$ to be the unit circle in $\mathbb{R}^2$ with the metric induced from the canonical metric of $\mathbb{R}^2$. We now give a mathematical definition of the orientation feature in the SIFT algorithm. 
\begin{defn}[Orientation]
\label{def:orient}
Given a patch $p:D_r \rightarrow \RR$, let $\Psi$ be the map 
\begin{eqnarray}
\label{def:psi}
\Psi : D_r &\rightarrow& S^1 \nonumber \\
\mathbf{x} &\mapsto& \theta,
\end{eqnarray}
where $\theta$ is an angle between $\nabla L_p(\mathbf{x})$ and the positive $x$-axis. 
For a fixed positive number $\delta<\pi$, the \textit{orientation feature} is defined as an angle $\theta ^* \in S^1$ such that
\be
\iint_{\Psi^{-1}(N_{\theta_*})} \|\nabla L(\mathbf{x})\| d\mathbf{x} = \max_{\theta '} \iint_{\Psi^{-1}(N_{\theta'})} \|\nabla L(\mathbf{x})\| d\mathbf{x},
\ee
where $N_{\theta'}:=\{\theta|\,d_{S^1}(\theta ,\theta ')<\delta\}$ and $d_{S^1}$ is the distance with respect to the canonical metric on $S^1$.
\end{defn}

In general, there might be more than one orientation feature associated with one patch. To simplify the discussion, we assume that there is only one orientation feature.

\begin{Assumption}\label{Assumption:OneOrientation}
Take $\epsilon>0$ and $\delta>0$ associated with the orientation feature. Assume
\be
\iint_{\Psi^{-1}(N_{\theta_*})} \|\nabla L(\mathbf{x})\| d\mathbf{x} > \iint_{\Psi^{-1}(N_{\theta'})} \|\nabla L(\mathbf{x})\| d\mathbf{x} + \epsilon  \pi r^2
\ee
for any $\theta ' \in S^1$ such that $d_{S^1}(\theta ',\theta _*)>\delta$. 
\end{Assumption}

With the orientation features of patches, we could define the following ``SIFT distance'' between patches.
\begin{defn}
The ``SIFT distance'' between patches $\tilde{p}$ and $p$ is defined as
\be
d_{\textup{SIFT}}(\tilde{p}, p) := \|\tilde{p} - R(\tilde{\theta})R(\theta)^{-1}.p\|,
\ee
where $\tilde{\theta}\in S^1$ and $\theta\in S^1$ are the orientation feature of $\tilde{p}$ and $p$, and $\| \cdot \|$ denotes the $L^2$ norm. 

\end{defn}

First of all, note that the ``SIFT distance'' is not really a distance, but it is intimately related to the RID.
If $\tilde{p}=R(\phi).p$, where $\phi\in S^1$ and $R: S^1\to SO(2)$ is a diffeomorphic map, then it is clear that the orientation features of $\tilde{p}$ and $p$, denoted as $\tilde{\theta}$ and $\theta$ respectively, are related by $R(\tilde{\theta})R(\theta)^{-1} = R(\phi)$. However, the reverse is not true. If two patches have the same orientation features, it does not imply that they are the same. In practice, suppose the orientation features of $\tilde{p}$ and $p$ are $\tilde{\theta}$ and $\theta$ respectively, we have
\be
d_{\texttt{RID}}(\tilde{p},p)\leq \|\tilde{p}-R(\tilde{\theta})R(\theta)^{-1}.p\|_{L^2}.
\ee
In other words, two patches that are determined to be neighbors by the RID will be determined to be neighbors by the $d_{\texttt{SIFT}}$ defined in (\ref{approximate_accurate_RID}). Note that the SIFT distance could be further improved to better approximate the RID, like (\ref{approximate_accurate_RID}).
Thus, orientation features under SIFT can be used to estimate rotation angles between patches, or to ``nonlinearly filter out'' the impossible neighbors.

Next, we show why the orientation feature is robust to noise. We first study gradients of Gaussian smoothed patches. While the noise is modeled by the generalized random process $\psi\Phi$ in (\ref{definition:ContinuousNoisyPatch}), to make the calculation succinct, we consider the following simplified model:
\be
\label{noisy_patch}
p^{(\texttt{n})}_i(\mathbf{x}) = p^{(\texttt{c})}_i(\mathbf{x}) + \sigma \xi_i(\mathbf{x})
\ee
be a noisy patch defined on $D_r$, where $\sigma>0$, $p^{(\texttt{c})}_i$ denotes the patch of the clean image supported on $D_r$, and $\xi_i$ are i.i.d. standard random normal variables for all $\mathbf{x}\in D_r$.\footnote{Note that in the general model with the patch defined in (\ref{definition:ContinuousNoisyPatch}), the calculation is the same while the cut-off function $\psi$ will come into play and the calculation will be tedious. For example, in (\ref{Calculation:Proof:SIFT:Robust}) the noise term becomes $\Phi(\psi G_x)$ and $\Phi(\psi G_y)$, and the expectation and variance will be similar to the result, but the expression will not be explicit.}
To further simplify notation, we denote the Gaussian smoothed noisy and clean patches, $p^{(\texttt{n})}_i$ and $p^{(\texttt{c})}_i$ defined in (\ref{definition:ContinuousNoisyPatch}), by 
\be
 L^{(\texttt{n})}(\mathbf{x}) := L_{p^{(\texttt{n})}_i}(\mathbf{x}), L^{(\texttt{c})}(\mathbf{x}) := L_{p^{(\texttt{c})}_i}(\mathbf{x}),  \mbox{   and   }
\xi = \xi_i 
\ee
respectively.

\begin{lem}
\label{lem:grad}
We have 
\begin{align}
\mathbb{E}\left[ \| \nabla L^{(\texttt{n})}(\mathbf{x}) - \nabla L^{(\texttt{c})}(\mathbf{x}) \|^2 \right] = \frac{\sigma^2}{4\pi} 
\end{align}
and
\begin{align}
\mathrm{Var}\left[ \| \nabla L^{(\texttt{n})}(\mathbf{x}) - \nabla L^{(\texttt{c})}(\mathbf{x}) \|^2 \right] = \frac{\sigma^4}{16\pi^2} \left(1+\frac{27}{16\pi}\right).
\end{align}
\end{lem}

\begin{proof} 
Note that 
\be
\| \nabla L^{(\texttt{n})}(\mathbf{x}) - \nabla L^{(\texttt{c})}(\mathbf{x}) \|^2 = \sigma^2\left( \left( \iint_{\RR^2} \xi(\mathbf{s}) G_x(\mathbf{x-s})d\mathbf{s} \right)^2 + \left( \iint_{\RR^2} \xi(\mathbf{s}) G_y(\mathbf{x-s})d\mathbf{s} \right)^2\right), \label{Calculation:Proof:SIFT:Robust}
\ee
where we denote $\nabla L=[L_x\,\, L_y]^T$. 
Since $\xi(\mathbf{s}) \sim \mathcal{N}(0,1)$  are i.i.d., we have 
\begin{align}
 &\mathbb{E} \left[\| \nabla L^{(\texttt{n})}(\mathbf{x}) - \nabla L^{(\texttt{c})}(\mathbf{x}) \|^2 \right] \nonumber \\
 =\,& \sigma^2\mathbb{E}\left[ \iint_{\RR^2} G_x(\mathbf{x-s})\xi(\mathbf{s}) d\mathbf{s}  \iint_{\RR^2} G_x(\mathbf{x-t})\xi(\mathbf{t}) d\mathbf{t} \right] \nonumber \nonumber\\
 &\quad +  \sigma^2\mathbb{E}\left[ \iint_{\RR^2} G_y(\mathbf{x-s})\xi(\mathbf{s}) d\mathbf{s}  \iint_{\RR^2} G_y(\mathbf{x-t})\xi(\mathbf{t}) d\mathbf{t} \right] \nonumber\\
 =\,& \sigma^2\left( \iint_{\RR^2} \EE[\xi(\mathbf{s})^2] G_x(\mathbf{x-s})^2 + \EE[\xi(\mathbf{s})^2] G_y(\mathbf{x-s})^2 d\mathbf{s}\right) \nonumber\\
 =\,& \sigma^2 \left( \iint_{\RR^2} \frac{x^2}{4\pi} e^{-(x^2+y^2)} +  \frac{y^2}{4\pi} e^{-(x^2+y^2)} d\mathbf{s} \right)
 = \frac{\sigma^2}{4\pi}.
\end{align}
To obtain the variance, we evaluate $\mathbb{E}\left( \| \nabla L^{(\texttt{n})}(\mathbf{x}) - \nabla L^{(\texttt{c})}(\mathbf{x}) \|^4 \right)$.
\begin{align}
&\EE \left[ \| \nabla L^{(\texttt{n})}(\mathbf{x}) - \nabla L^{(\texttt{c})}(\mathbf{x}) \|^4 \right]   \nonumber\\
=\,& \sigma^4 \EE\left[ \left( \iint_{\RR^2}  G_x(\mathbf{x-s})\xi(\mathbf{s}) d\mathbf{s}  \right)^4\right] + \,  \sigma^4 \EE\left[ \left( \iint_{\RR^2} G_y(\mathbf{x-s})\xi(\mathbf{s}) d\mathbf{s}  \right)^4\right]  \nonumber\\
& + 2 \sigma^4 \EE\left[\left(\iint_{\RR^2} G_x(\mathbf{x-s})\xi(\mathbf{s}) d\mathbf{s}  \right)^2 \left(\iint_{\RR^2} G_y(\mathbf{x-s})\xi(\mathbf{s}) d\mathbf{s}  \right)^2 \right] \nonumber \\
=\,& \sigma^4 \left[ \left( 3 \iint_{\RR^2} \EE[\xi(\mathbf{s})^2] G_x(\mathbf{x-s}) ^2 d\mathbf{s} \right)^2  + \iint_{\RR^2} \EE[\xi(\mathbf{s})^4] G_x(\mathbf{x-s})^4 d\mathbf{s}\right] \nonumber \\
&  +  \sigma^4 \left[ \left( 3 \iint_{\RR^2} \EE[\xi(\mathbf{s})^2] G_y(\mathbf{x-s}) ^2 d\mathbf{s} \right)^2  + \iint_{\RR^2} \EE[\xi(\mathbf{s})^4] G_y(\mathbf{x-s})^4 d\mathbf{s}\right]  \nonumber\\
&  +2 \sigma^4 \left(\iint_{\RR^2} \EE[\xi(\mathbf{s})^2]  G_x(\mathbf{x-s}) ^2 d\mathbf{s} \right) \left(\iint_{\RR^2} \EE[\xi(\mathbf{s})^2]  G_y(\mathbf{x-s}) ^2 d\mathbf{s} \right)  \nonumber \\
&+ 4\sigma^4  \left(\iint_{\RR^2} \EE[\xi(\mathbf{s})^2] G_x(\mathbf{x-s}) G_y(\mathbf{x-s}) d\mathbf{s} \right)
+ 2 \sigma^4 \iint_{\RR^2} \EE[\xi(\mathbf{s})^4]G_x(\mathbf{x-s})^2 G_y(\mathbf{x-s})^2 d\mathbf{s}  \nonumber\\
=\,&\sigma^4\left( \frac{1}{8\pi^2} + \frac{27}{256 \pi^3}\right).  
\end{align}
Therefore, 
\be
\mathrm{Var}\left[ \| \nabla L^{(\texttt{n})}(\mathbf{x}) - \nabla L^{(\texttt{c})}(\mathbf{x}) \|^2 \right] =\frac{\sigma^4}{16\pi^2} \left(1+\frac{27}{16\pi}\right).
\ee
\end{proof}

\begin{cor}
\label{cor:grad}
With high probability, we have that $\| \nabla L^{(\texttt{n})}(\mathbf{x}) \| \approx \| \nabla L^{(\texttt{c})}(\mathbf{x})\|$. More precisely, for $k>0$
\be
\label{approx:grad}
\mathrm{Pr}\left( \left| \|\nabla L^{(\texttt{n})}(\mathbf{x}) \| - \| \nabla L^{(\texttt{c})}(\mathbf{x}) \| \right| < \sigma \sqrt{\frac{1+1.25k}{4\pi}} \right) > 1 - \frac{1}{1+k^2}.
\ee
\end{cor}

\begin{proof}
By Lemma \ref{lem:grad} and the one-sided Chebyshev inequality, we can obtain
\be
\mathrm{Pr}\left(  \|\nabla L^{(\texttt{n})}(\mathbf{x}) - \nabla L^{(\texttt{c})}(\mathbf{x}) \|^2 < \frac{\sigma^2}{4\pi} + k\frac{\sigma^2}{4\pi}\left(1+\frac{27}{16\pi} \right)\right) > 1 - \frac{1}{1+k^2}.
\ee
Since $\left| \|\nabla L^{(\texttt{n})}(\mathbf{x}) \| - \| \nabla L^{(\texttt{c})}(\mathbf{x}) \| \right| < \| \nabla L^{(\texttt{n})}(\mathbf{x}) - \nabla L^{(\texttt{c})}(\mathbf{x})\|$, and $\sqrt{1+\frac{27}{16\pi}} <1.25$, we obtain (\ref{approx:grad}).
\end{proof}

\begin{prop}
Let $p^{(\texttt{n})}$ be the associated noisy patch of a clean patch $p^{(\texttt{c})}$. Denote their orientation assignments by $\theta^{(c)*}\in S^1$ and $\theta^{(n)*}\in S^1$, respectively. Suppose $\sigma$ is small and $k$ satisfies $\sigma \sqrt{\frac{1+1.25k}{4\pi}}< \epsilon$, where $\epsilon>0$. With the probability higher than $1-\frac{1}{1+k^2}$, we have
\be
d_{S^1}(\theta^{(c)*}, \theta^{(n)*}) <\delta,
\ee
where $\delta$ is the number in Definition \ref{def:orient}.
\end{prop}

\begin{proof}
By Corollary \ref{cor:grad}, the probability that
\be
\| \nabla L^{(\texttt{c})}\| - \sigma \sqrt{\frac{1+1.25k}{4\pi}} < \| \nabla L^{(\texttt{n})}\| < \| \nabla L^{(\texttt{c})}\| - \sigma \sqrt{\frac{1+1.25k}{4\pi}}
\ee
is higher than $1-\frac{1}{1+k^2}$.

Suppose that $d_{S^1}(\theta^{(c)*}, \theta^{(n)*})> \delta$. Then
\begin{align}
\iint_{\Psi^{-1}(N_{\theta^{(n)*}})} \|\nabla L^{(\texttt{n})}(\mathbf{x})\| d\mathbf{x} 
>& \iint_{\Psi^{-1}(N_{\theta^{(c)*}})} \|\nabla L^{(\texttt{n})}(\mathbf{x})\| d\mathbf{x}  +  \epsilon \pi r^2 \nonumber\\
>& \iint_{\Psi^{-1}(N_{\theta^{(c)*}})} \Big(\|\nabla L^{(\texttt{c})}(\mathbf{x})\| -\sigma \sqrt{\frac{1+1.25k}{4\pi}}\Big)  d\mathbf{x}  + \epsilon \pi  r^2,
\end{align}
where $\Psi$ is defined as (\ref{def:psi}).
On the other hand,
\be
\iint_{\Psi^{-1}(N_{\theta^{(n)*}})} \|\nabla L^{(\texttt{n})}(\mathbf{x})\| d\mathbf{x}
< \iint_{\Psi^{-1}(N_{\theta^{(c)*}})} \Big(\|\nabla L^{(\texttt{c})}(\mathbf{x})\| + \sigma \sqrt{\frac{1+1.25k}{4\pi}}\Big)  d\mathbf{x}.
\ee
Combining the two inequalities and the assumption $\sigma \sqrt{\frac{1+1.25k}{4\pi}}< \epsilon$, we have 
 \be
 \iint_{\Psi^{-1}(N_{\theta^{(n)*}})} \|\nabla L^{(\texttt{c})}(\mathbf{x})\| 
 < \iint_{\Psi^{-1}(N_{\theta^{(c)*}})} \|\nabla L^{(\texttt{c})}(\mathbf{x})\|   d\mathbf{x} +  \epsilon \pi r^2
 \ee
 which leads to a contradiction to Assumption \ref{Assumption:OneOrientation}. Hence $d_{S^1}(\theta^{(c)*}, \theta^{(n)*}) < \delta$.
\end{proof}

\begin{rmk}
Definition \ref{def:orient} corresponds to the case where there is only one orientation feature in the SIFT algorithm. We may generalize the definition that allows two (or more) orientation features as the following: 

For fixed small positive numbers $\delta$ and $\epsilon$,  angles $\theta _1^*, \theta _2^* \in S^1$  and $d_{S^1}(\theta _1^*, \theta _2^*) >\delta$ are orientations if 
\be
\iint_{\Psi^{-1}(N_{\theta^*_i})} \|\nabla L(\mathbf{x})\| d\mathbf{x} > \iint_{\Psi^{-1}(N_{\theta'})} \|\nabla L(\mathbf{x})\| d\mathbf{x} + \epsilon  \pi r^2
\ee
for any $\theta ' \in S^1$ outside of $N_{\theta_1^{*}}$ and $N_{\theta_2^{*}}$. 
For a patch with two orientations, we can prove that with high probability, the orientations of the associated noisy patch will be close to the ones of the clean patch.
\end{rmk}

\section{Image Quality Assessment}
\label{sec:img_quality}

Image quality assessment (IQA) is an important subfield in image processing. The goal is to find an index quantifying ``how good'' an image is, which is suitable for different scenarios. We consider measures of two major categories in this paper to evaluate the VNLEM algorithm. The first category consists of objective measures based on a chosen theoretical model without taking the human visual system (HVS) into account. The second category consists of objective measures based on models taking the HVS into account. Below we summarize these measures.
Denote the clean image as $I\in\mathbb{R}^{N\times N}$. We are concerned with how close the noisy observation $I+\sigma\xi\in\mathbb{R}^{N\times N}$ is to $I$, or the denoised image $\tilde{I}\in\mathbb{R}^{N\times N}$ is to $I$. 

The signal-to-noise ratio (SNR) belongs to the first category, and is given in decibels. By denoting 
\be
E:=\left[\sum_{i=1,\ldots,N^2}(\tilde{I}(i)-I(i))^2\right]^{1/2}, 
\ee
the SNR is defined as
\begin{align}
\texttt{SNR}=20\log_{10}\big(\frac{\sigma_{I}}{E}\big)\,,
\end{align}
where $\sigma_I$ is defined in (\ref{Definition:SigmaI}) and assumed to be $1$.
Clearly, if the denoising algorithm can fully recover the clean image; that is $\tilde{I}=I$, then the SNR is $\infty$. 
The peak-signal-to-noise ratio (PSNR) also belongs to the first category, which is given in decibels:
\begin{align}
\texttt{PSNR}=20\log_{10}\big(\frac{p_I}{E}\big)\,,
\end{align}
where 
\be
p_I:=\max_{i=1,\ldots,N^2}|I(i)|.
\ee
The SNR gives us a sense of how strong the signal and the noise are, but if the image is rather homogenous, the SNR is not very informative. The PSNR is a lot more content dependent, and it gives us a sense of how well the high-intensity regions of the image is coming through the noise i.e. the contrast. Since the denoising filter can adjust the contrast of the image, the PSNR can be rather helpful in demonstrating the performance of the various denoising filters. While SNR and PSNR are widely applied IQA's in the field, they do not necessarily tell us all aspects of how well the denoising methods performed. For example, they do not readily capture the edge preserving capability of an algorithm.

To capture the edge preservation performance, we consider the third measurement, the \textit{Sobolev index} \cite{Wilson1997}, which also belongs to the first category. Let $\hat{I}$ and $\hat{\tilde{I}}$ denote the discrete Fourier transforms of $I$ and $\tilde{I}$, respectively. The Sobolev index of order $\kappa$ is then defined by the Sobolev norm, and is given by
\begin{align}
\texttt{SOB} =\left[\frac{1}{\vert \Omega \vert ^ 2} \sum_{\omega \in \Omega} (1 + \vert \eta_\omega \vert ^ 2) ^ \kappa  \vert \hat{I}(\omega) - \hat{\tilde{I}}(\omega) \vert ^ 2  \right]^{1 / 2}\,,
\end{align}
where $\Omega$ is the lattice of the frequency domain and $\eta_\omega$ is the two-dimensional frequency vector associated with $\omega\in\Omega$.

The SNR, PSN, and the Sobolev norm aim to evaluate how close the denoised image is to the clean image. We further consider the earth mover's distance (EMD) to measure how well we could recover the noise \cite[Section 2.2]{VillaniEMD}. 
The EMD between two probability distributions $\mu$ and $\nu$ on $\mathbb{R}$ is defined as
\[
d_{\texttt{OT}}(\mu,\nu) := \int_{\mathbb{R}}\, |f_\mu(x)-f_\nu(x)|\, d x \,,
\]
where $f_\mu(x):=\int_{-\infty}^x d \mu $ is the cumulative distribution function of $\mu$ and similarly for $f_\nu$. 
We will evaluate the EMD to compare how close the distribution of the estimated noise is to the added noise.

The above measurements are designed mainly around the idea of ``how well the error is captured'', or ``error sensitivity'' \cite{Wang_Bovik_Sheikh_Simoncelli:2004}. While they have been widely applied in different problems and provide useful information, it has been well accepted that they do not capture all aspects from the perspective of image quality. Particularly, generally it is not statistically consistent with human observers \cite{Zhang_Zhang_Mou_Zhang:2011}. Several metrics have been designed in the past decades to faithfully take the HVS into account, and they belong to the second category. These metrics emphasize the importance of luminance, the contrast, and the frequency/phase content. To further evaluate the performance of VNLEM, we consider the state-of-art measurement in this category, 
the Feature SIMilarity (FSIM) index \cite{Zhang_Zhang_Mou_Zhang:2011}. 
The FSIM is based on the model that the HVS perceives an image mainly based on its low-level features, such as edges and zero crossings, and it separates the similarity measurement task into phase congruency and gradient magnitude. 
Here we summarize the FSIM index. Suppose the dynamical range of the image is $\mathcal{R}$. 
The definition of FSIM depends on the definition of the phase congruency and gradient magnitude. The phase congruent of $I$ at $i$, denoted as $P_I(i)$, and the gradient magnitude of $I$ at $i$, denoted as $G_I(i)$, are defined in \cite[Equation (3) and Section II.B]{Zhang_Zhang_Mou_Zhang:2011}. Similarly, we could define $P_{\tilde{I}}(i)$ and $G_{\tilde{I}}(i)$. The FSIM between $I$ and $\tilde{I}$ is defined as
\be
\texttt{FSIM}(I,\tilde{I}):=\frac{\sum_{i=1}^{N^2}S_L(i)P_m(i)}{\sum_{i=1}^{N^2}P_m(i)},
\ee
where
\be
P_m(i)=\max\{P_I(i),P_{\tilde{I}}(i)\},\quad S_L(i):=S_P(i)S_G(i),
\ee
\be
S_P(i):=\frac{2P_I(i)P_{\tilde{I}}(i)+T_1}{P_I(i)^2+P_{\tilde{I}}(i)^2+T_1},\quad\mbox{and }S_G(i):=\frac{2G_I(i)G_{\tilde{I}}(i)+T_2}{G_I(i)^2+G_{\tilde{I}}(i)^2+T_2}.
\ee
Here, we follow \cite{Zhang_Zhang_Mou_Zhang:2011} and choose $T_1=0.85$ and $T_2=160$.
There are several other measures of this kind in the field, and we refer interested readers to \cite{Wang_Bovik_Sheikh_Simoncelli:2004,Zhang_Zhang_Mou_Zhang:2011} for a review of these indices.

\section{Numerical Result}
\label{sec:stats}

In our numerical experiments, we fix the following parameters for NLEM, VNLEM, and VNLEM-DD for a fair comparison. Fix $q=13$. We build $13 \times 13$ patches around each pixel of the noisy image. We chose $\epsilon = (16.5)^2$, the number of nearest neighbors as $N_1=100$, the size of the search window for creating the initial affinity matrix is determined by $N_2=10$; that is, $21 \times 21$ neighbours of each patch are chosen for the search window. The $\theta_l$ in (\ref{approximate_Rotation_RID}) is set to 30 degrees, the upsampling operator $U_k$ is implemented by the bicubic interpolation, and $k$ is set to $2$. After building the transition matrix, we choose $m=30$ to evaluate the DM and DD. Finally, we select $\gamma = 0.1$ for the final denoising step. The Matlab code is available via request.

To compare our results with those of the NLEM algorithm, we also preformed the NLEM denoising with $\epsilon = (6.5)^2$\footnote{The code is available in \url{https://www.mathworks.com/matlabcentral/fileexchange/40624-non-local-patch-regression}}, where the search window and patch sizes are chosen to be identical to those selected for our proposed schemes.  
The kernel bandwidth is chosen to give the best performance for the NLEM algorithm in terms of SNR and PSNR.  

In Table~\ref{stats} we report the different IQA metrics, including SNR, PSNR, RMS, SOB, and FSIM discussed previously as well as the computational time, by running the three denoising algorithms on 1,361 sample images of size $512\times 512$\footnote{The images are collected from : \\
$\bullet$ Caltech-UCSD Birds-200-2011 collection at :
\url{http://www.vision.caltech.edu/visipedia/CUB-200-2011.html}
\\ $\bullet$ Digital Image Processing, 3rd ed, by Gonzalez and Woods at : \url{http://www.imageprocessingplace.com/DIP-3E/dip3e_book_images_downloads.htm}
\\
$\bullet$ USC-SIPI image database at: \url{http://sipi.usc.edu/database/}}. 
There are $98$ images for animals, $143$ images for flowers, $52$ images for fruits, $115$ images for landscapes, $450$ images for faces, $419$ images for manmade structures, and $44$ miscellaneous images.
The SOB metric is applied to the image recovery error. This measure particularly reflects the amount of edge information wiped out due to the denoising process. Therefore, the scheme with a lower SOB index performs the better. For the other indices, the higher the index is, the better the performance is. 
Under the null hypothesis that the performance of two algorithms is the same, we reject the hypothesis by the Mann-Whitney U test with the $p$ value less than $10^{-4}$. Note that based on the overall statistics, VNELM and VNLEM-DD outperform NLEM statistically significantly on all IQA metrics. On the other hand, we cannot distinguish the performance of VNLEM and VNLEM-DD statistically, except on the FSIM index. This result suggests that VNLEM-DD could better recover features sensitive to HVS.

The execution times based on 17 images are $501.8\pm 203.3$s, $1489.8\pm 26.1$s, and $1619\pm 34.8$s for NLEM, VNLEM, and VNLEM-DD respectively. This execution time is obtained on a PC with 8 Gb of RAM using a single core from Intel Corei7 CPU with a clock speed of 3.7 GHz running on Microsoft Windows 7.

\begin{table}[ht]
\caption{Summary statistics over 1,361 images of different denoising algorithms evaluated by different image quality assessment metrics. ${}^{*\#}$: $p<10^{-8}$. ${}^{\dagger}$: $p<10^{-6}$. a.u.: the arbitrary unit.}
\label{times}
\begin{center}
\begin{tabular}{m{100 pt} m{100pt}  m{100pt}  m{100pt}}
\hline
\vspace{5pt}
 & \textbf{NLEM} & \textbf{VNLEM} & \textbf{VNLEM-DD}
\\
\hline
\vspace{5pt}
\textbf{PSNR} (dB) & $18.78\pm 2.92^{*\#}$ & $19.49\pm 2.72^{*}$ & $19.62\pm 2.81^{\#}$ 
\\
\textbf{SNR} (dB) & $13.33\pm 2.78^{*\#}$ & $14.04\pm 2.36^{*}$ & $14.18\pm 2.49^{\#}$ 
\\
\textbf{RMS}$\times 100$ (a.u.) & $5.77\pm 1.58^{*\#}$ & $5.35\pm 1.33^{*}$ & $5.24\pm 1.36^{\#}$
 \\
\textbf{SOB}$\times 100$ (a.u.) & $5.9\pm 1.63^{*\#}$ & $5.45\pm 1.37^{*}$ & $5.35\pm 1.4^{\#}$ 
\\
\textbf{OT}$\times 100$  (a.u.) & $0.59\pm 0.4^{*\#}$ & $0.32\pm 0.16^{*}$ & $0.35\pm 0.22^{\#}$ 
\\
\textbf{FSIM}$\times 100$  (a.u.) & $88.33\pm 2.98^{*\#}$ & $89.64\pm 2.13^{*\dagger}$ & $90.03\pm 2.13^{\#\dagger}$ 
\\
\hline
\end{tabular}
\end{center}
\label{stats}
\end{table}

\subsection{Comparison between NLEM, VNLEM, and VNLEM-DD}

Fig.~\ref{beans} depicts an example of noisy image recovery performed using three different denoising algorithms, namely the NLEM, the VNLEM and the VNLEM-DD algorithms. The original image is of size $512 \times 512$. In this figure, we have also presented the denoising error for each scheme. These errors help us identify the amount of details and desired features that are lost in the image recovery process. For this and the consequent examples, we have also reported the PSNR, SNR, SOB, and FSIM values achieved by the denoising process. 
The PSNR, SNR, SOB, and FSIM all suggest that by taking the rotational fiber structure into account, the proposed algorithm improves the NLEM scheme in terms of the amount of details preserved in the recovered image.
A close look at the results of the recovered image suggests that the VNLEM-DD algorithm preserves the most amount of texture features present in the image. The smaller SOB's of VNLEM and VNLEM-DD indicate a better edge preservation. The visual perception improvement is captured by the higher FSIM.

Similar observations can be made in our second example in Fig.~\ref{starfish}.\footnote{The original image can be found at \url{https://wall.alphacoders.com/big.php?i=109992}.} 
The visual perception improvement by reading Fig.~\ref{starfish} is supported by the higher FSIM. By looking at the denoising errors one can notice that the details of the edges are lost in all three schemes. However, while the VNLEM and VNLEM-DD algorithms lead to higher SNR and PSNR, and edge preservation; this fact is quantified by the larger SOB.

\begin{figure}[ht]
\centering
\subfigure[Original image]{
$\hspace{-0.5cm}$
\includegraphics[width=0.30\textwidth]{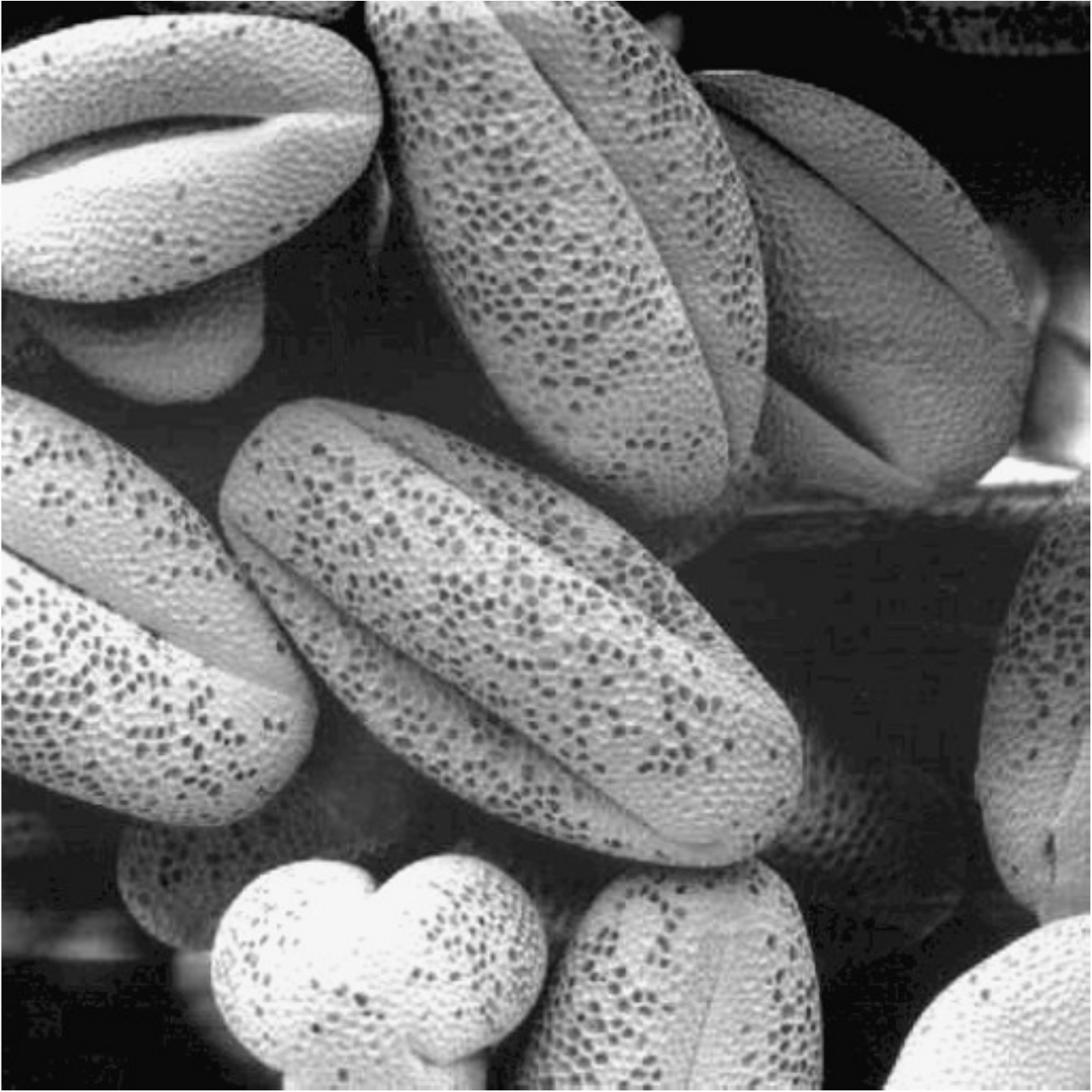}}
\subfigure[Noisy image, PSNR = $11.47$, SNR = $5.24$]{
\includegraphics[width=0.30\textwidth]{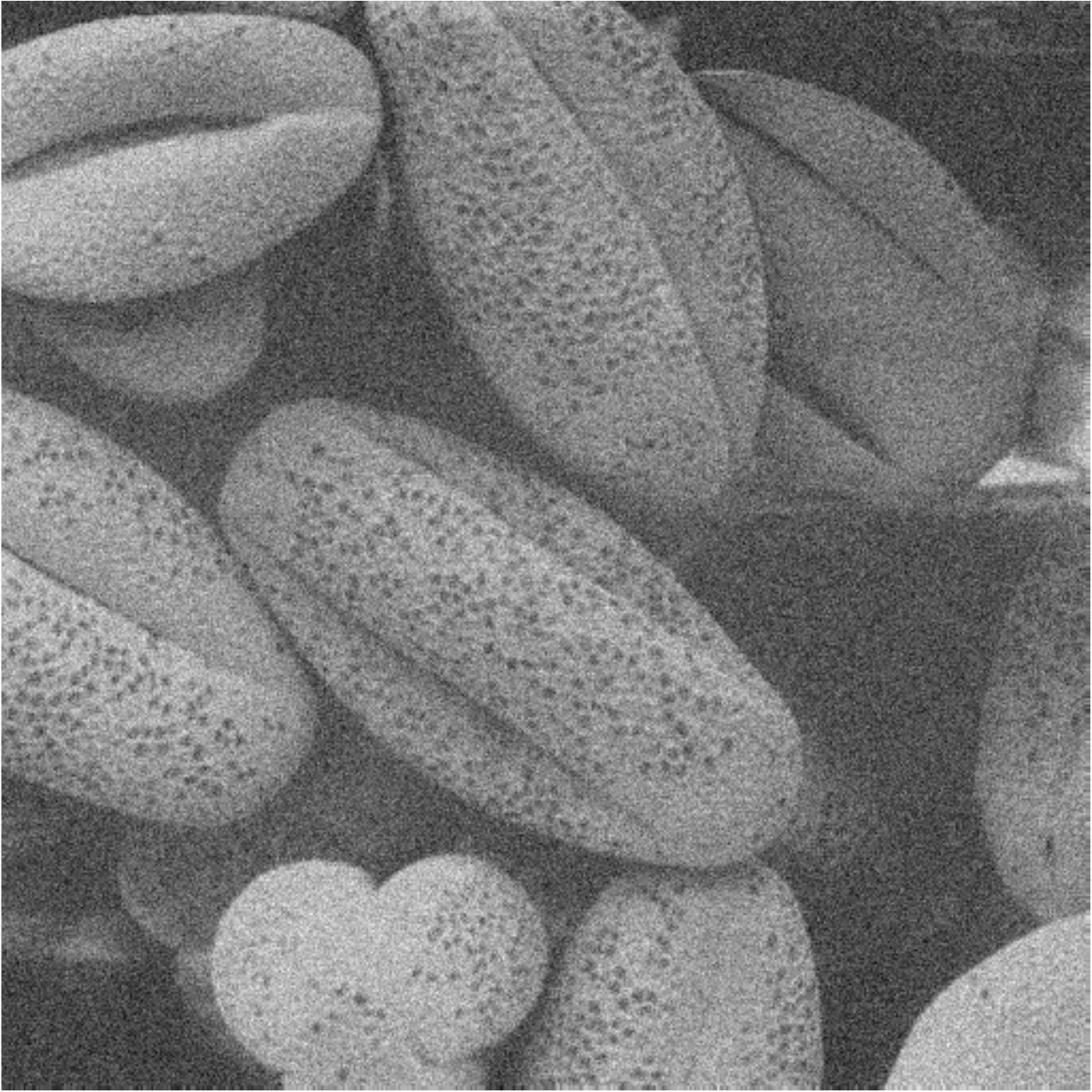}}

\subfigure[NLEM, PSNR = $19.37$, SNR = $13.14$, FSIM = $0.853$]{
\includegraphics[height=0.30\textwidth]{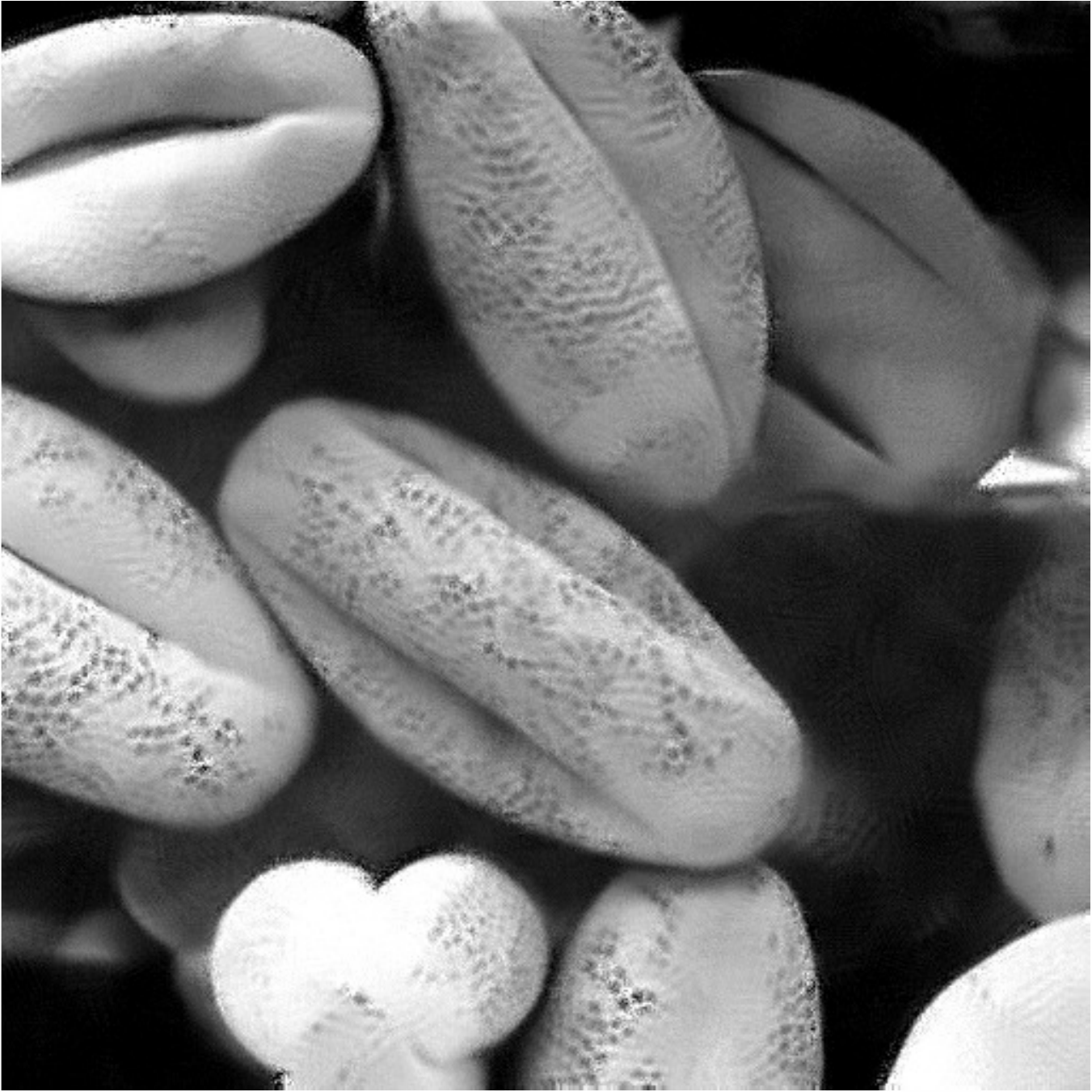}}
\subfigure[VNLEM, PSNR = $20.22$, SNR = $13.99$, FSIM = $0.883$]{
\includegraphics[height=0.30\textwidth]{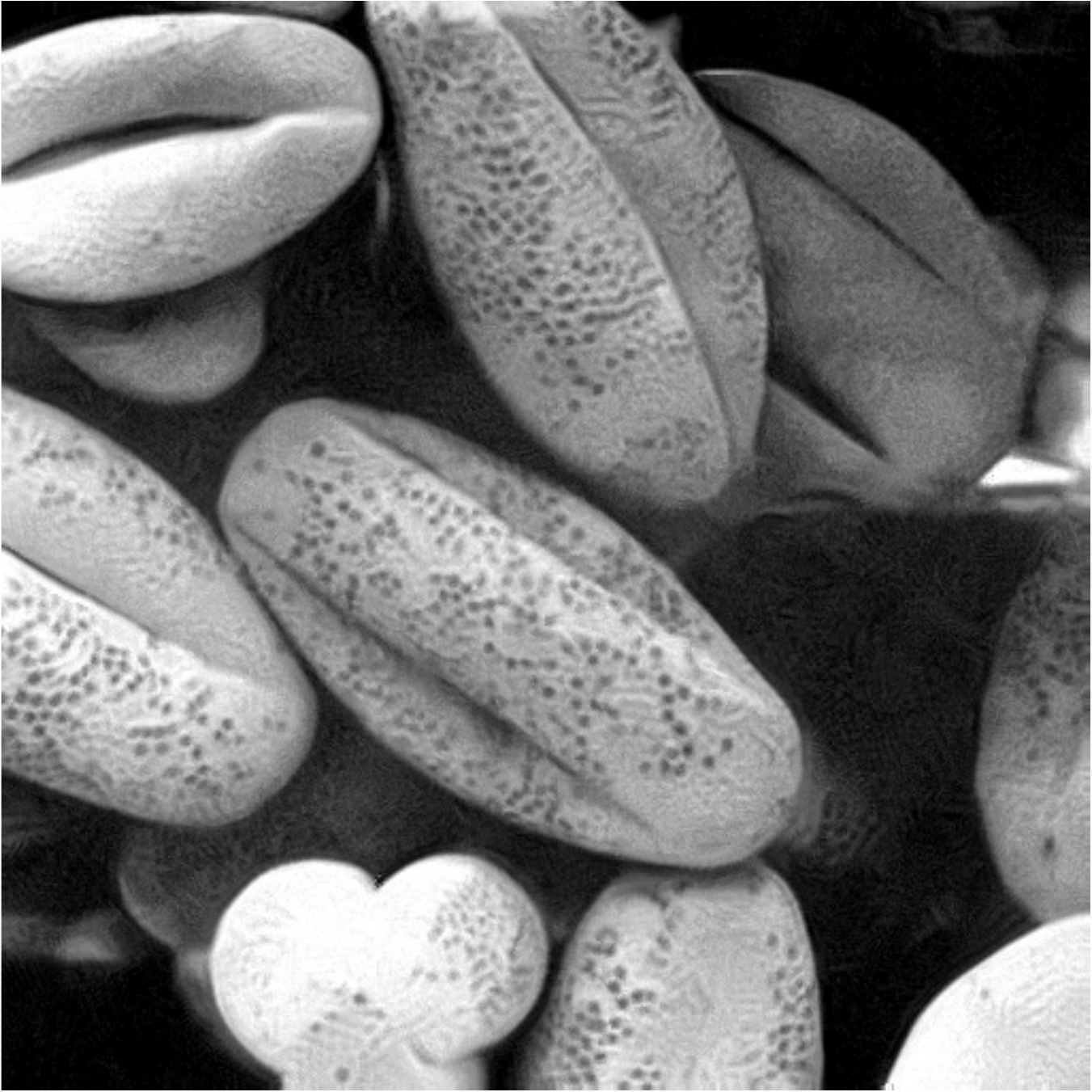}}
\subfigure[VNLEM-DD, PSNR = $20.07$, SNR = $13.84$, FSIM = $0.877$]{
\includegraphics[height=0.30\textwidth]{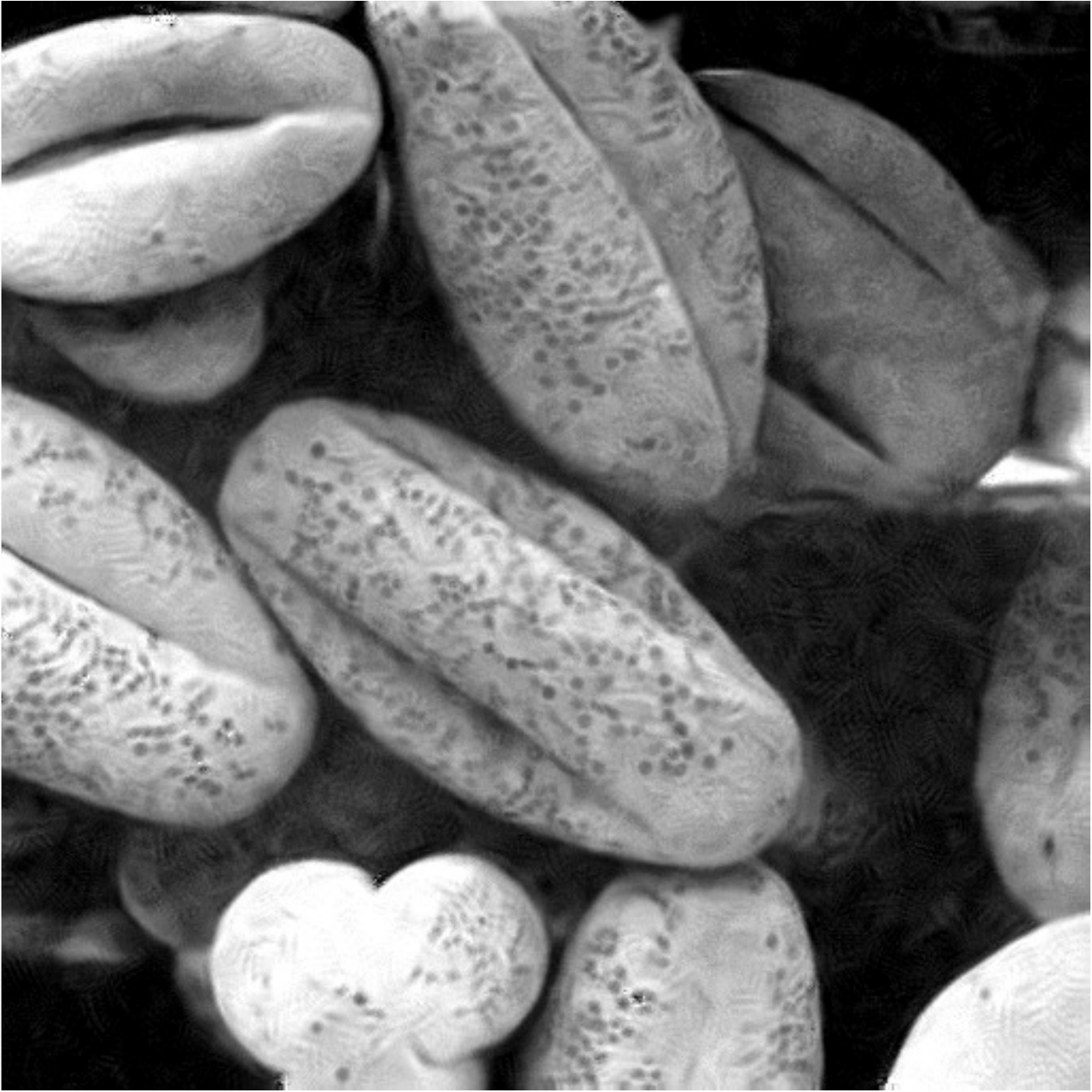}}

\subfigure[NLEM, difference, SOB = $0.166$]{
\includegraphics[height=0.28\textwidth]{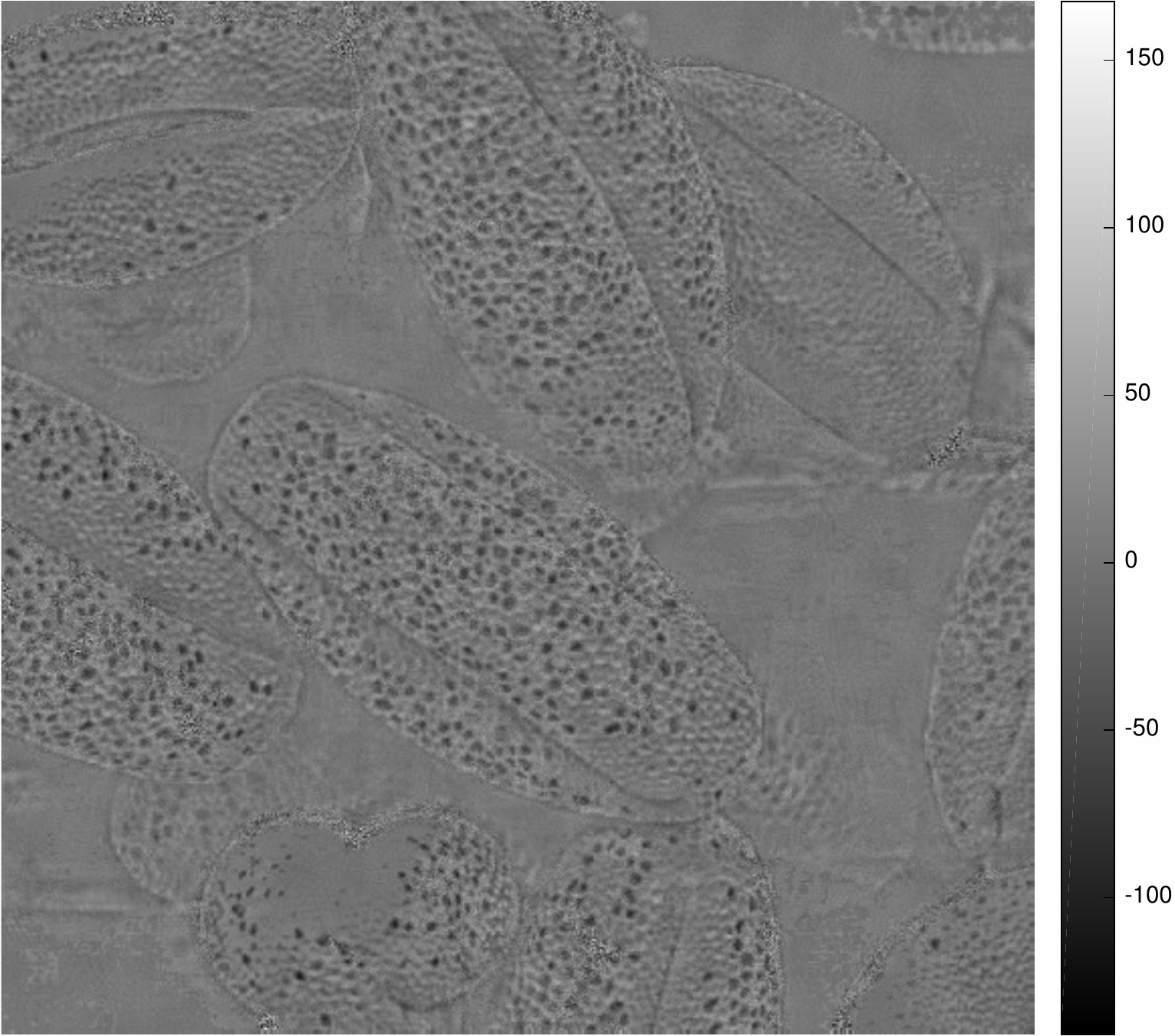}}
\subfigure[VNLEM, difference, SOB = $0.059$]{
\includegraphics[height=0.28\textwidth]{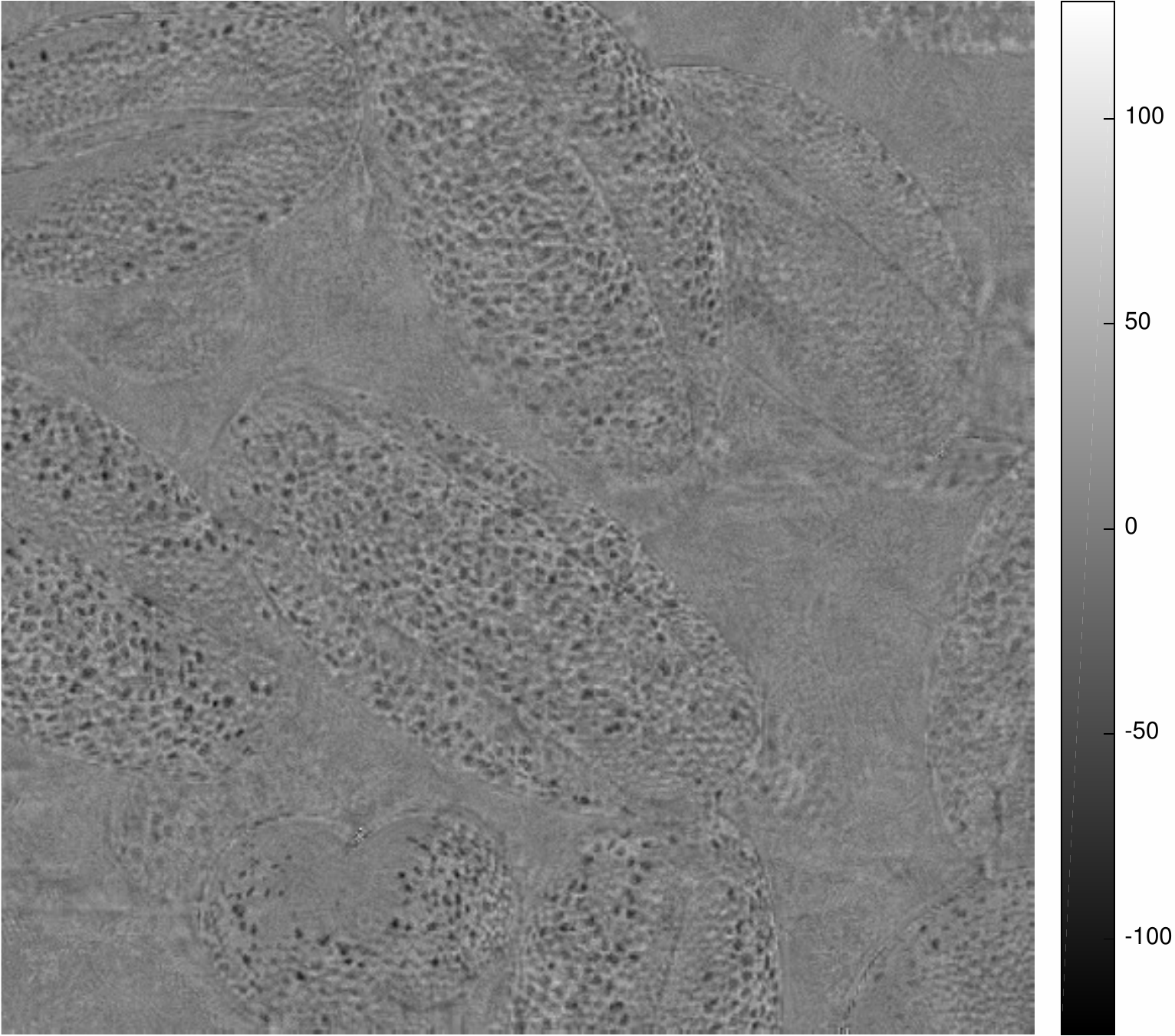}}
\subfigure[VNLEM-DD, difference, SOB = $0.06$]{
\includegraphics[height=0.28\textwidth]{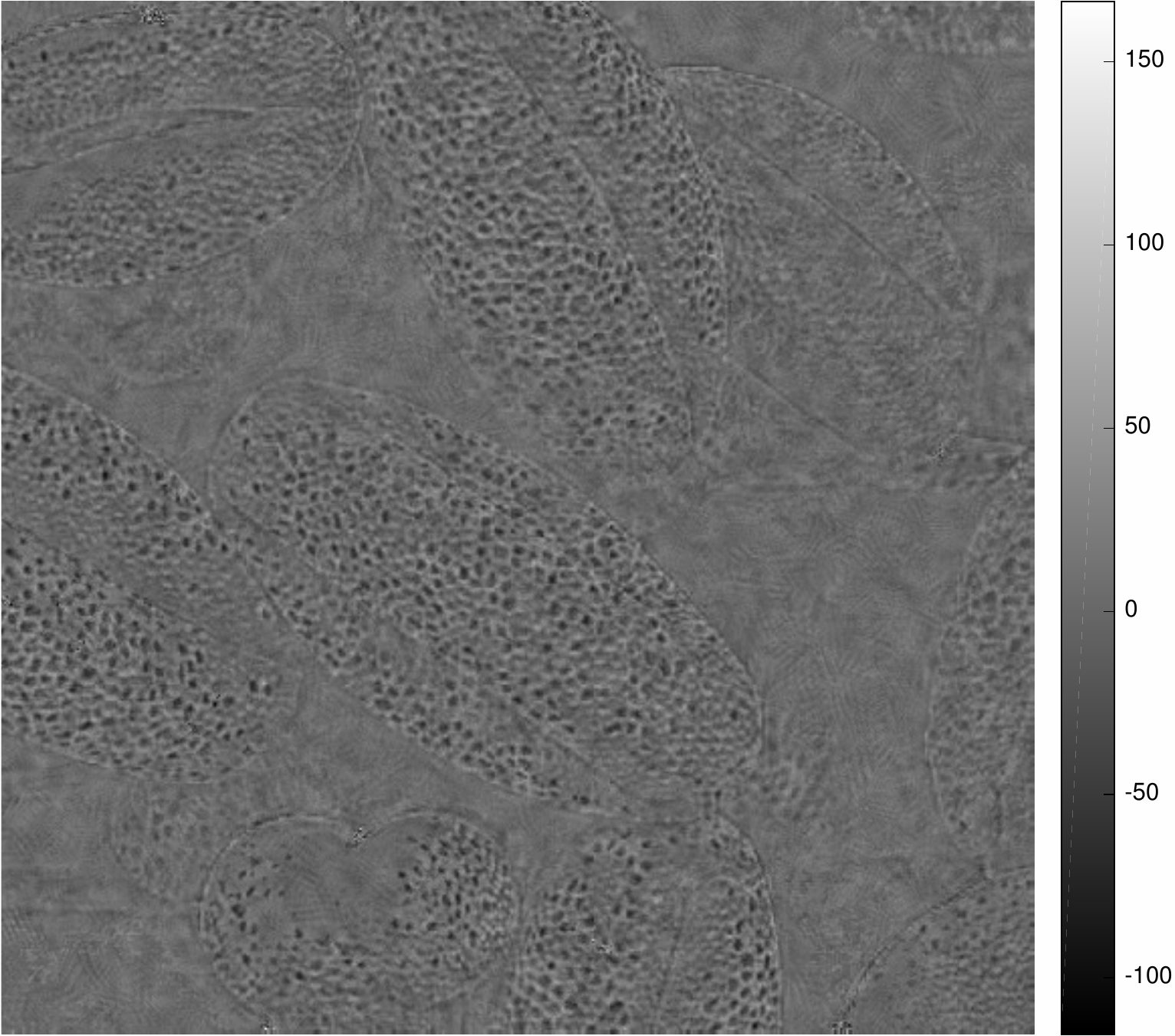}}

\caption{\label{beans}Example 1: the beans}
\end{figure}

\begin{figure}[ht]
\centering
\subfigure[Original image]{
$\hspace{-0.5cm}$
\includegraphics[height=0.30\textwidth]{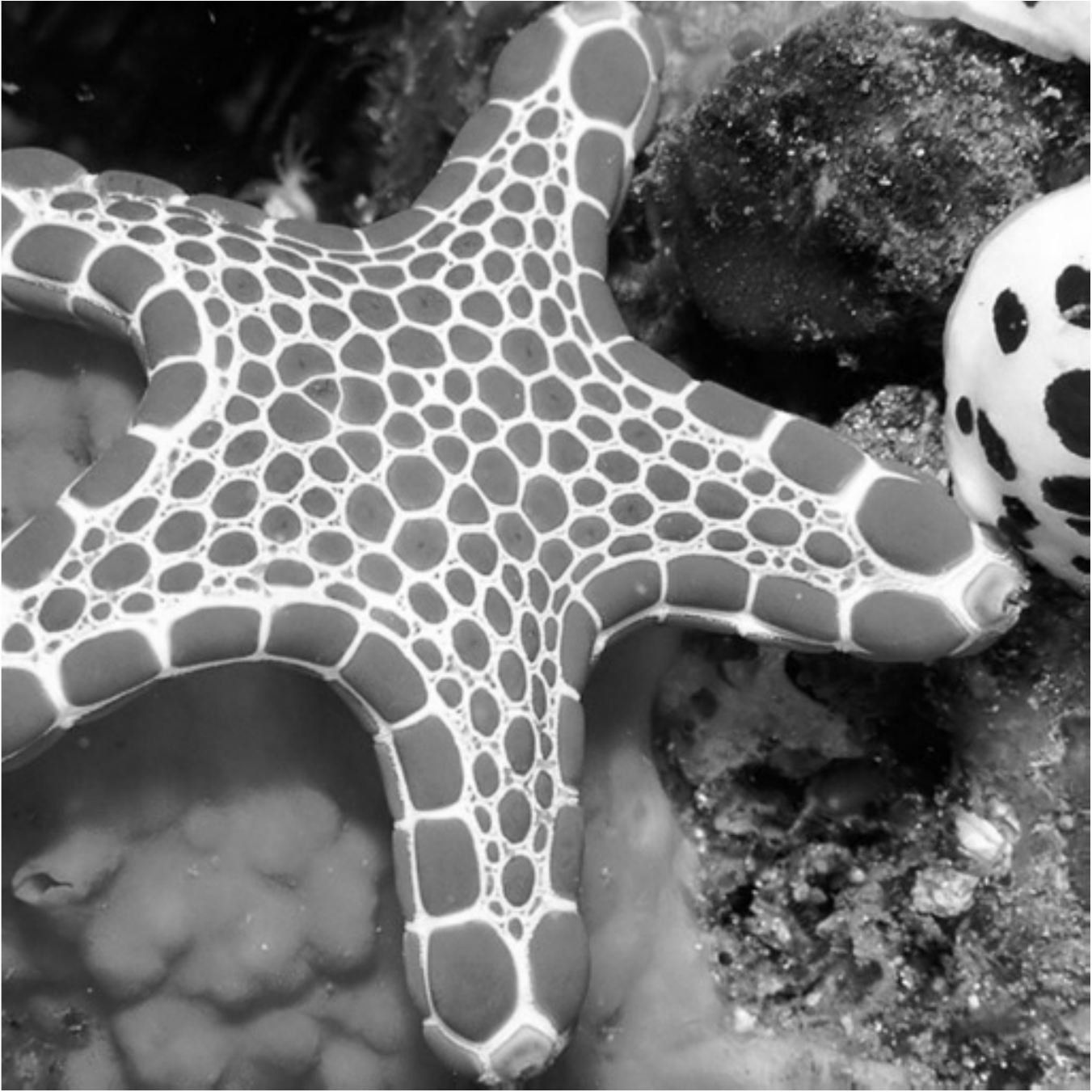}}
\subfigure[Noisy image, PSNR = $11.60$, SNR = $5.25$]{
\includegraphics[height=0.30\textwidth]{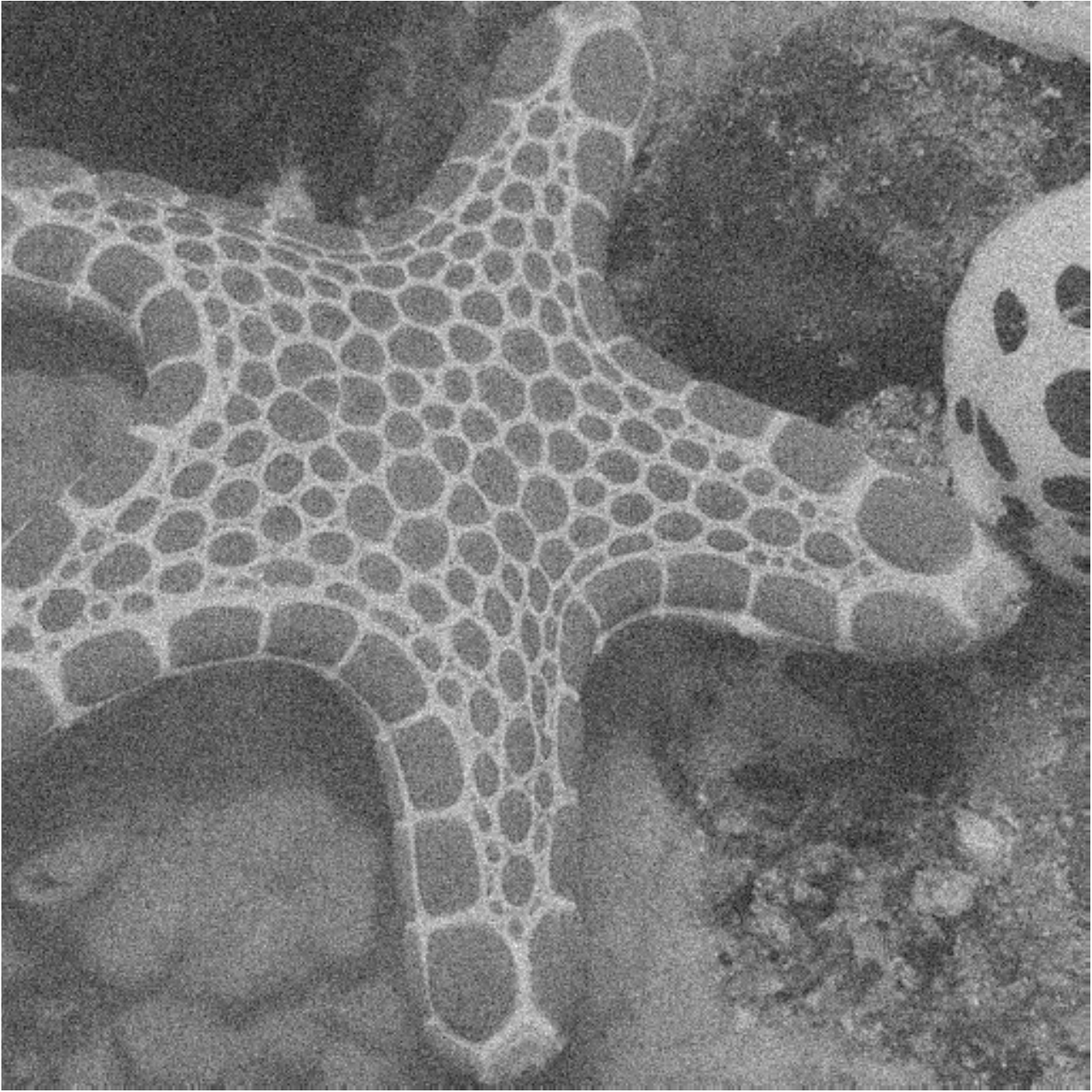}}

\subfigure[NLEM, PSNR = $16.79$, SNR = $10.42$, FSIM = $0.902$]{
\includegraphics[height=0.30\textwidth]{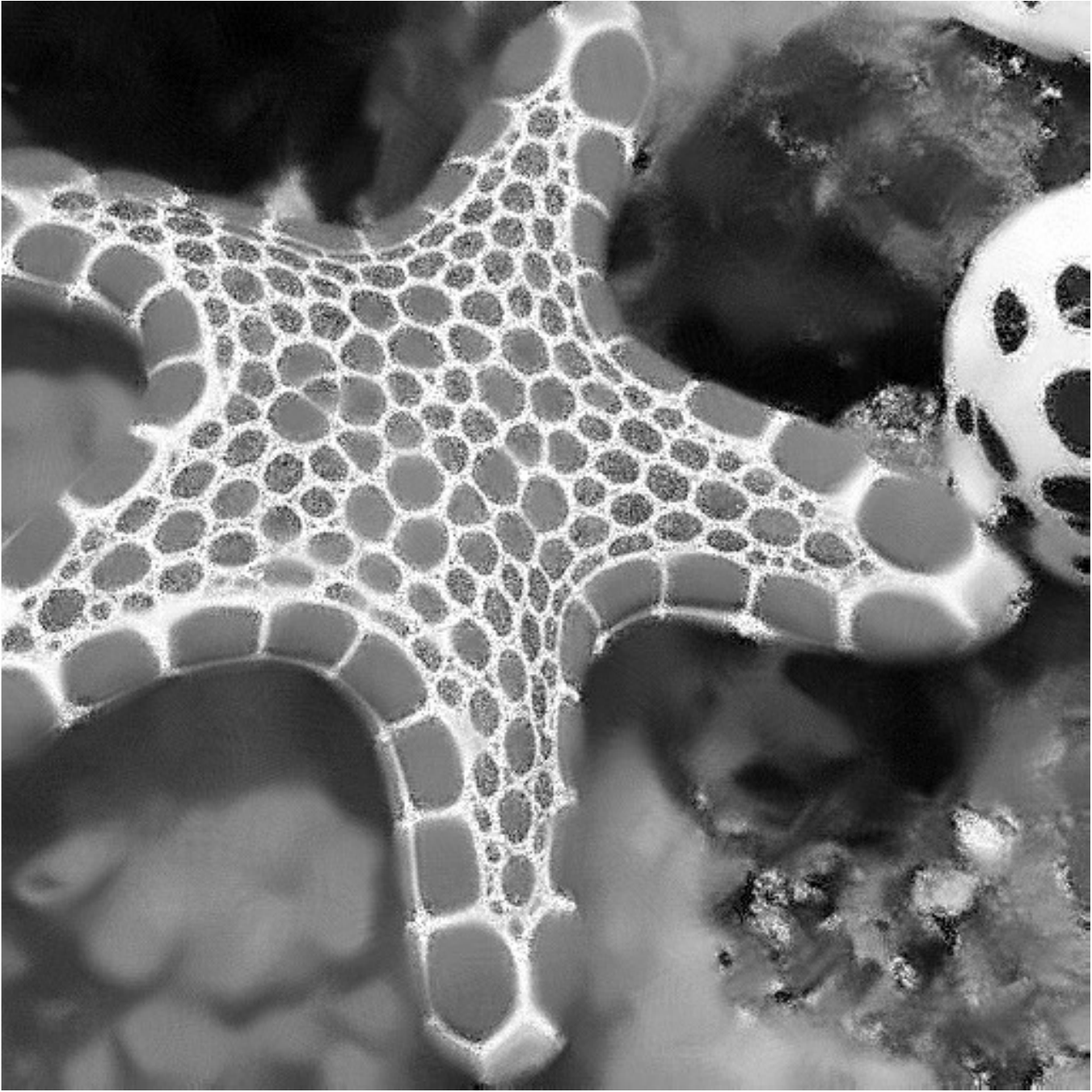}}
\subfigure[VNLEM, PSNR = $19.68$, SNR = $13.32$, FSIM = $0.914$]{
\includegraphics[height=0.30\textwidth]{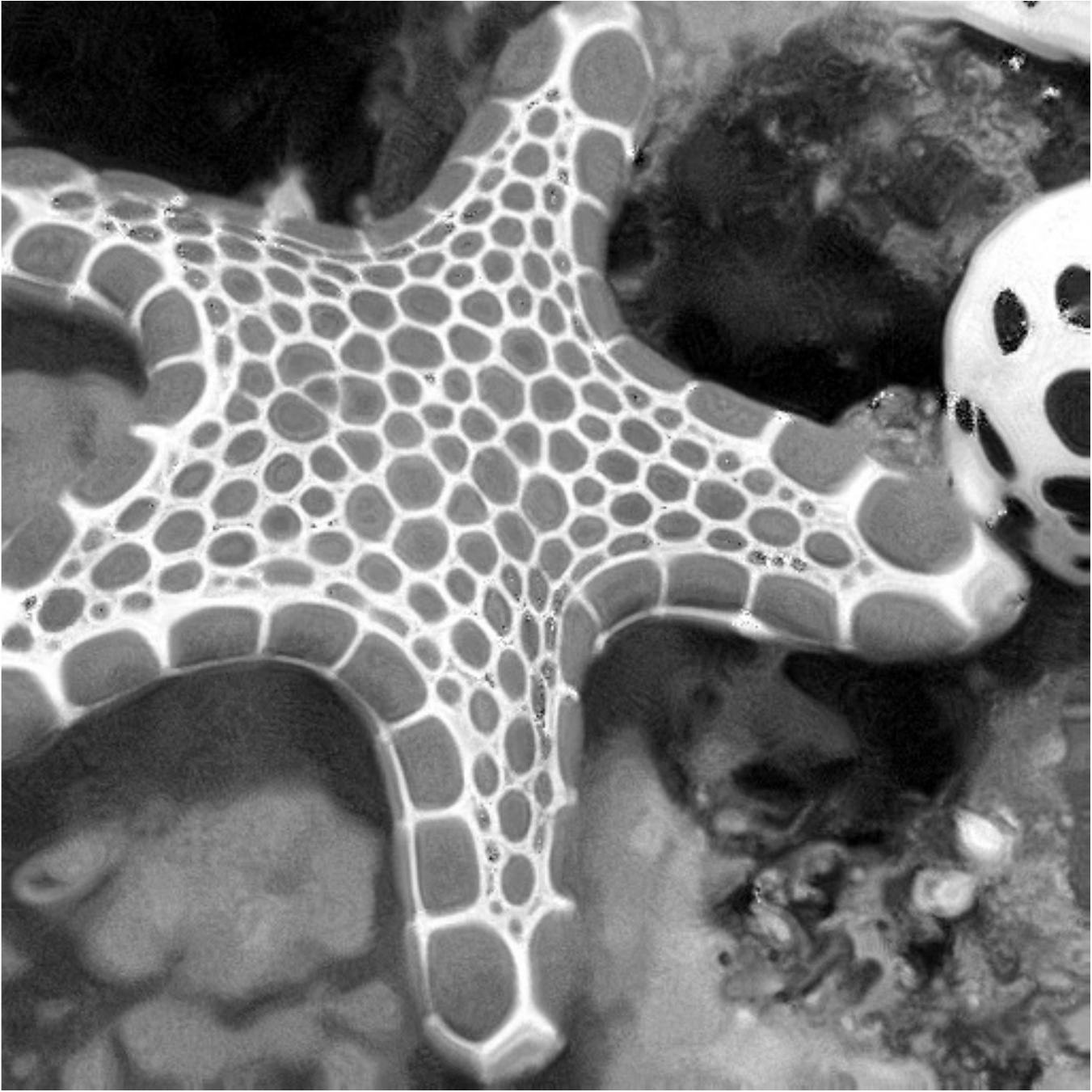}}
\subfigure[VNLEM-DD, PSNR = $19.49$, SNR = $13.13$, FSIM = $0.918$]{
\includegraphics[height=0.30\textwidth]{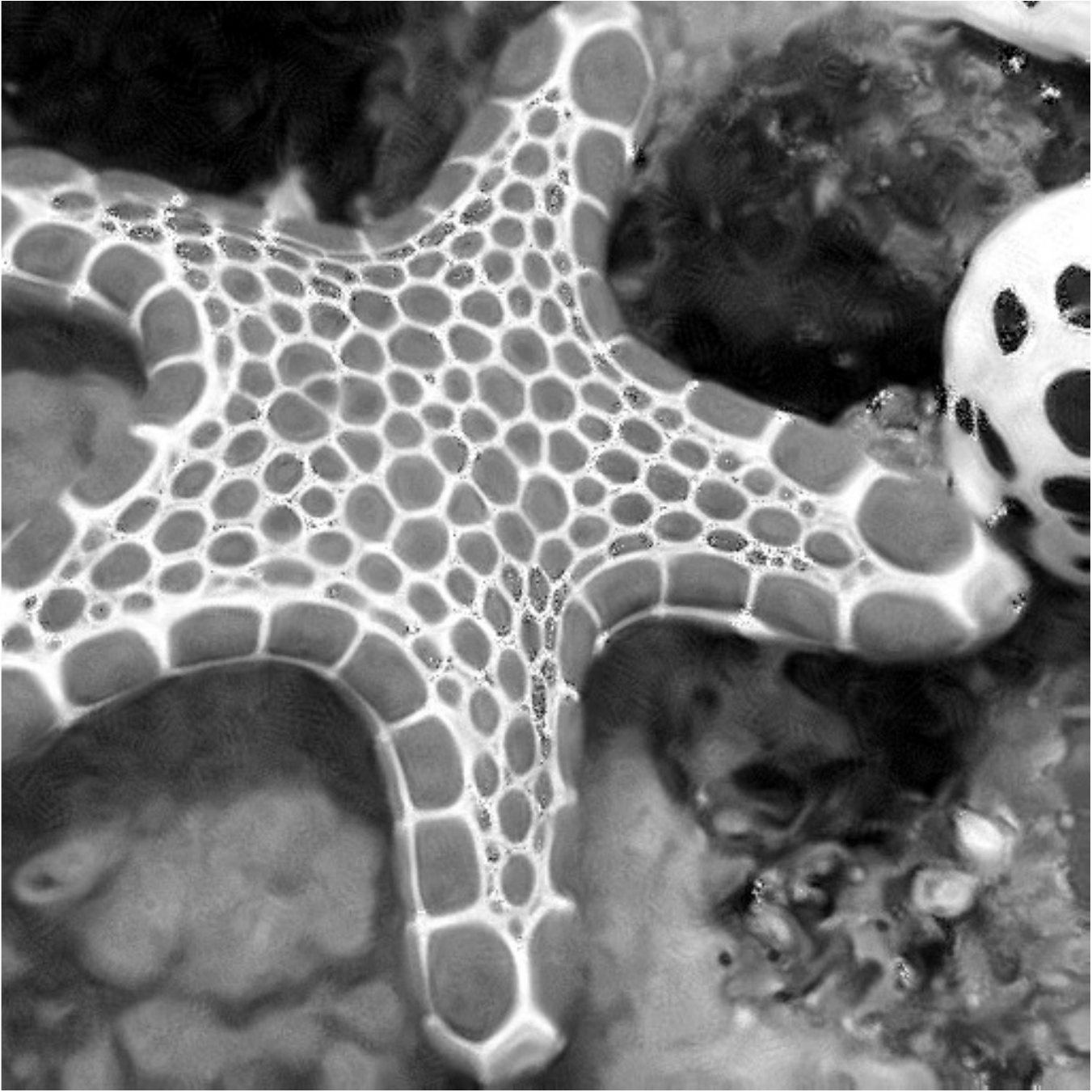}}

\subfigure[NLEM, difference, SOB = $0.047$]{
\includegraphics[height=0.28\textwidth]{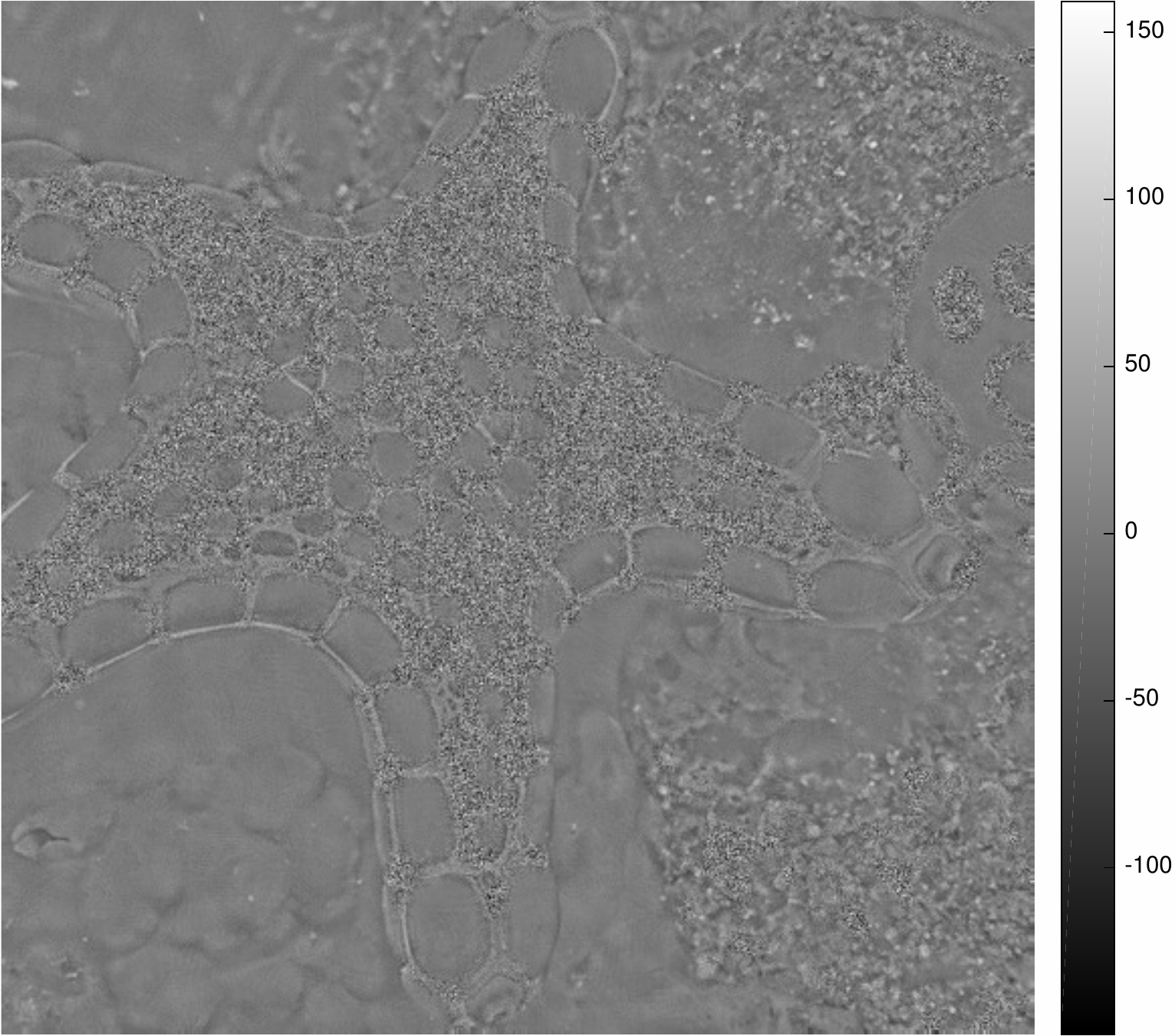}}
\subfigure[VNLEM, difference, SOB = $0.057$]{
\includegraphics[height=0.28\textwidth]{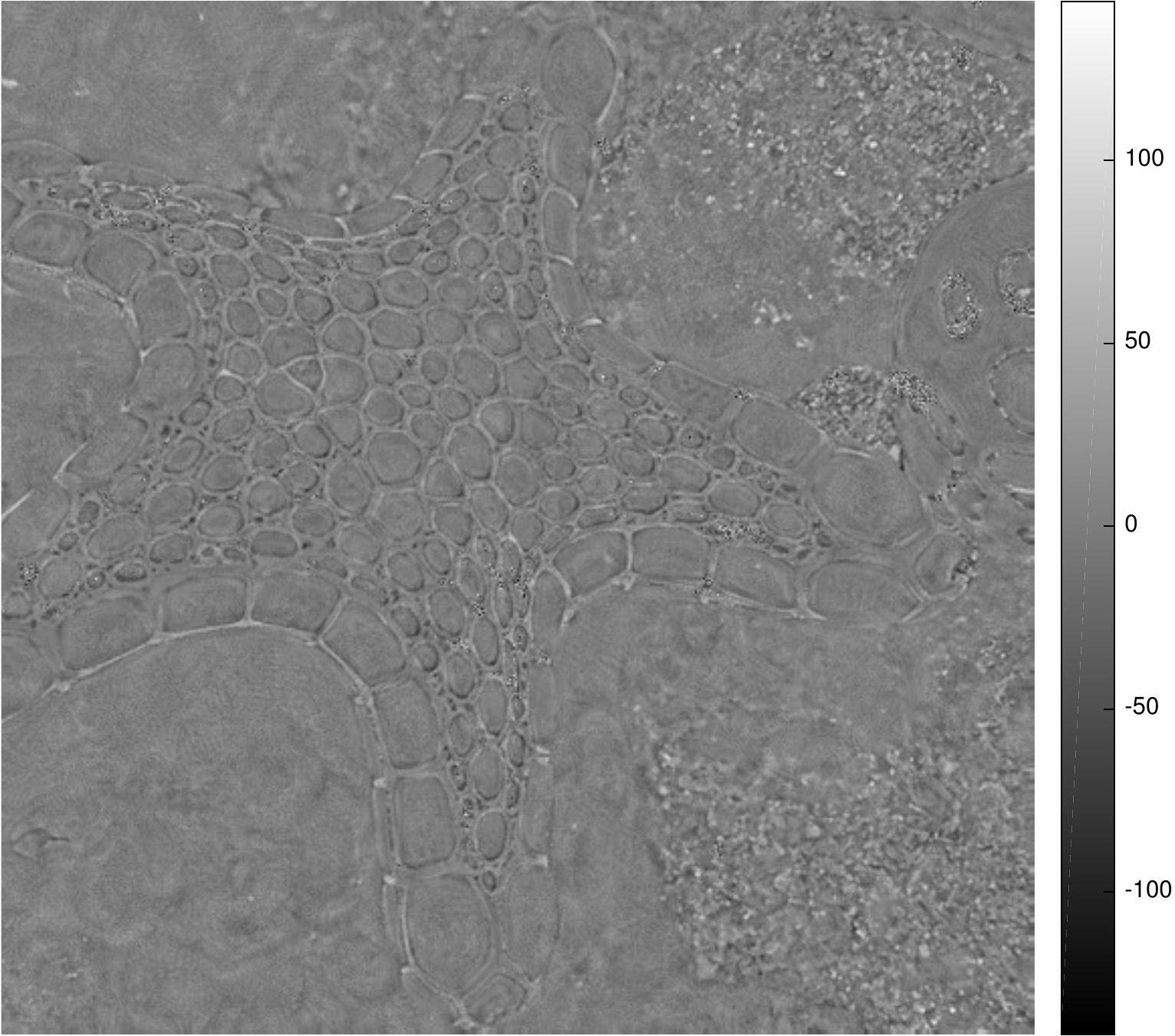}}
\subfigure[VNLEM-DD, difference, SOB = $0.057$]{
\includegraphics[height=0.28\textwidth]{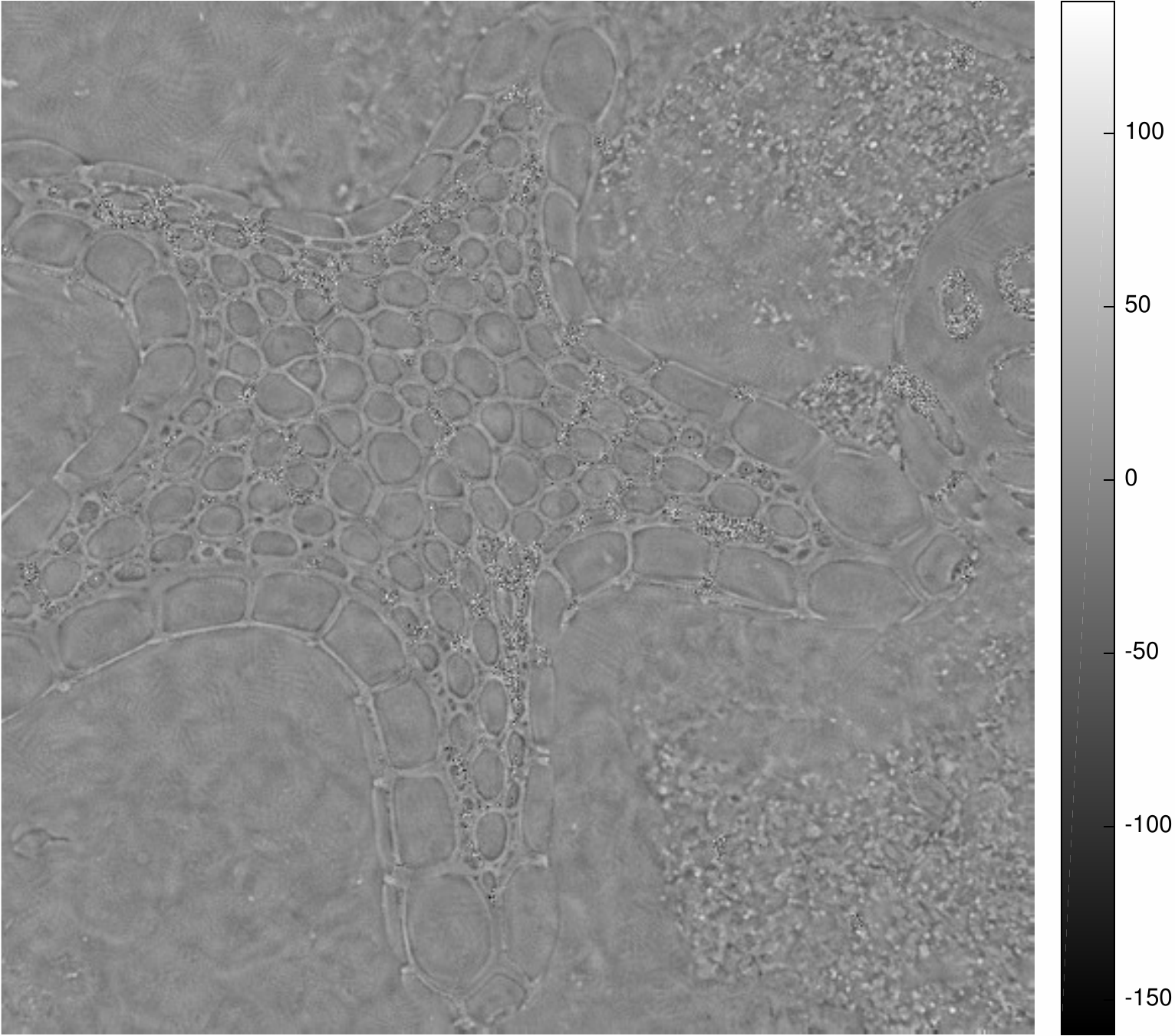}}

\caption{\label{starfish}Example 2: The starfish.}
\end{figure}

\subsection{Limitations of VNLEM and VNLEM-DD}
While statistically VNLEM and VNLEM-DD outperform NLEM, there are cases where NLEM outperforms. We now take a closer look into some of these examples. In Fig.~\ref{lady}, we see that the NLEM scheme achieves higher PSNR and SNR values in the recovered image. A second look at the recovery errors reveals that this superior performance comes at the cost of substantial loss of edge details in the results, and this is reflected in the higher SOB metric. Specifically, note that the teeth are better recovered in the VNLEM-DD. This result supports that the PSNR and SNR measures alone cannot throughly represent the performance of a denoising scheme, and IQA's from different perspectives are needed to better quantify the performance. 
It is also worth noting that the VNLEM-DD algorithm introduces ``texture-like artifices'' in the areas of the image that do not manifest any distinct feature, for example, the forehead of the portrait shown in Fig.~\ref{lady}. This comes from the following facts. Note that the clean patches associated with this region are concentrated at one point in the high dimensional space $\mathbb{R}^{q^2}$ and the added noise creates a geometric pattern (see, for example, the discussion in \cite{ElKaroui:2010}) that is irrelevant to the underlying image itself. Thus, the DD provides a deviated neighbors for the median filter. We thus have to be careful when applying VNLEM-DD on images with this type of ``flat region''.

\begin{figure}[ht]
\centering
\subfigure[Original image]{
$\hspace{-0.5cm}$
\includegraphics[width=0.30\textwidth]{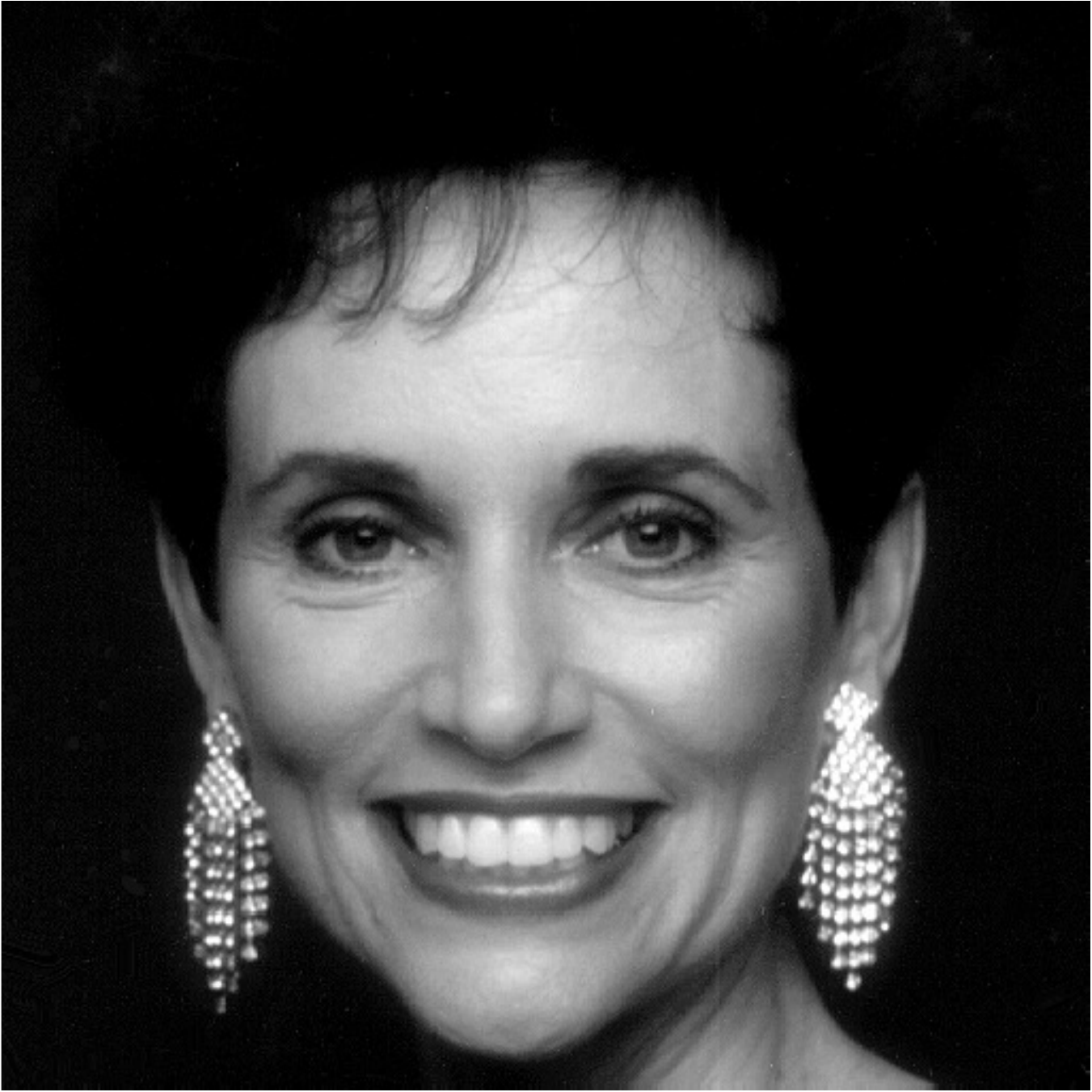}}
\subfigure[Noisy image, PSNR = $13.55$, SNR = $5.23$]{
\includegraphics[width=0.30\textwidth]{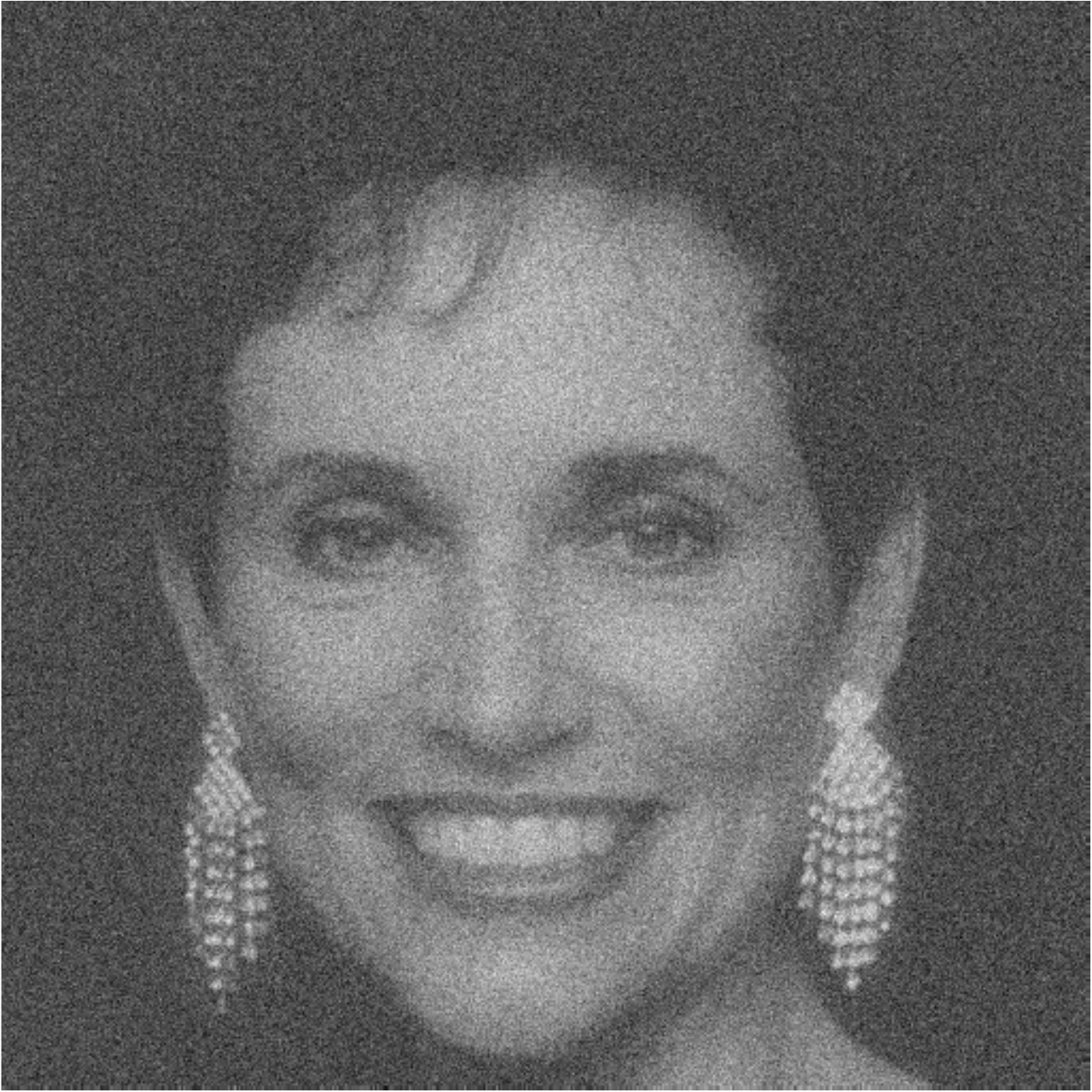}}

\subfigure[NLEM, PSNR = $25.68$, SNR = $17.36$, FSIM = $0.928$]{
\includegraphics[height=0.30\textwidth]{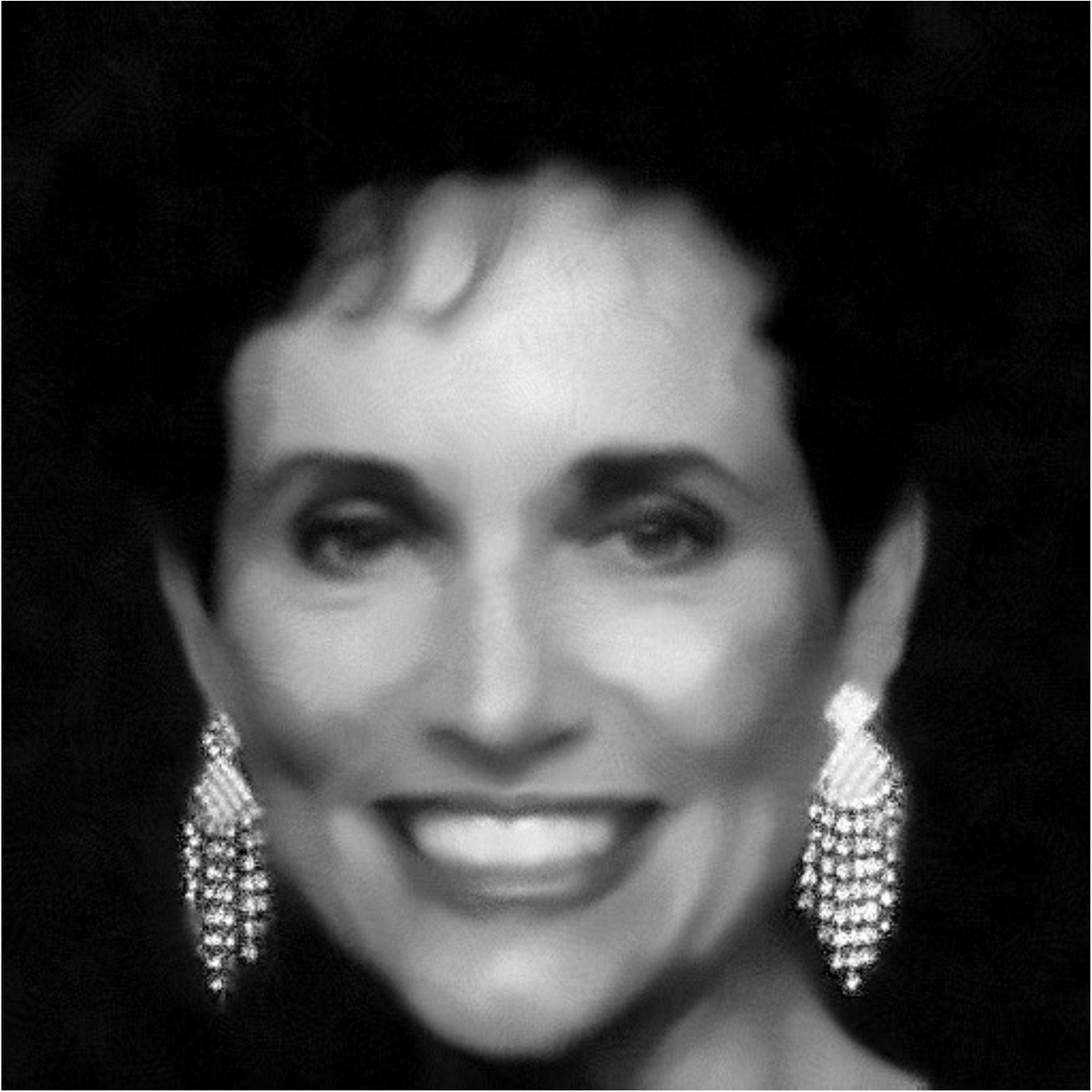}}
\subfigure[VNLEM, PSNR = $25.48$, SNR = $17.16$, FSIM = $0.897$]{
\includegraphics[height=0.30\textwidth]{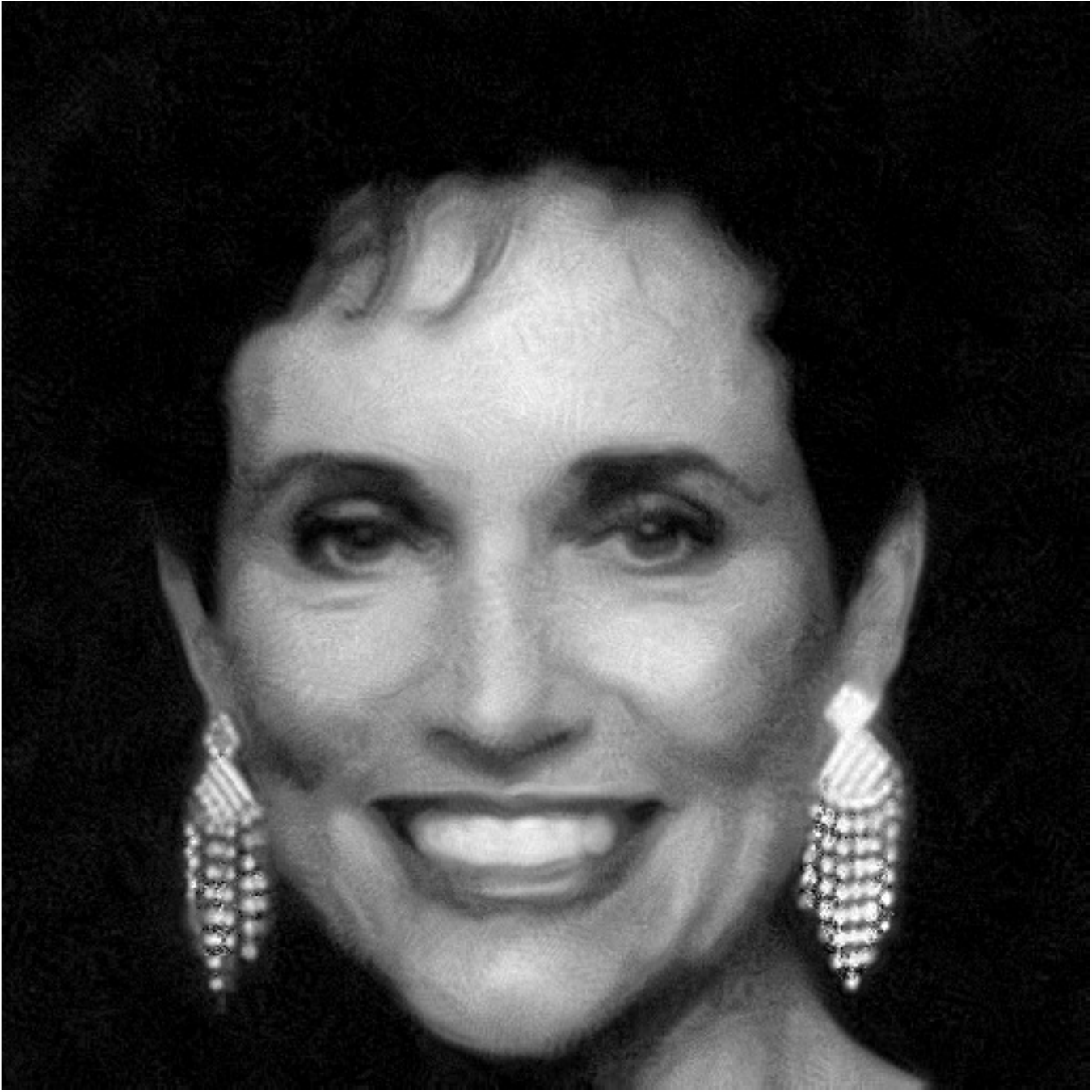}}
\subfigure[VNLEM-DD, PSNR = $25.54$, SNR = $17.21$, FSIM = $0.898$]{
\includegraphics[height=0.30\textwidth]{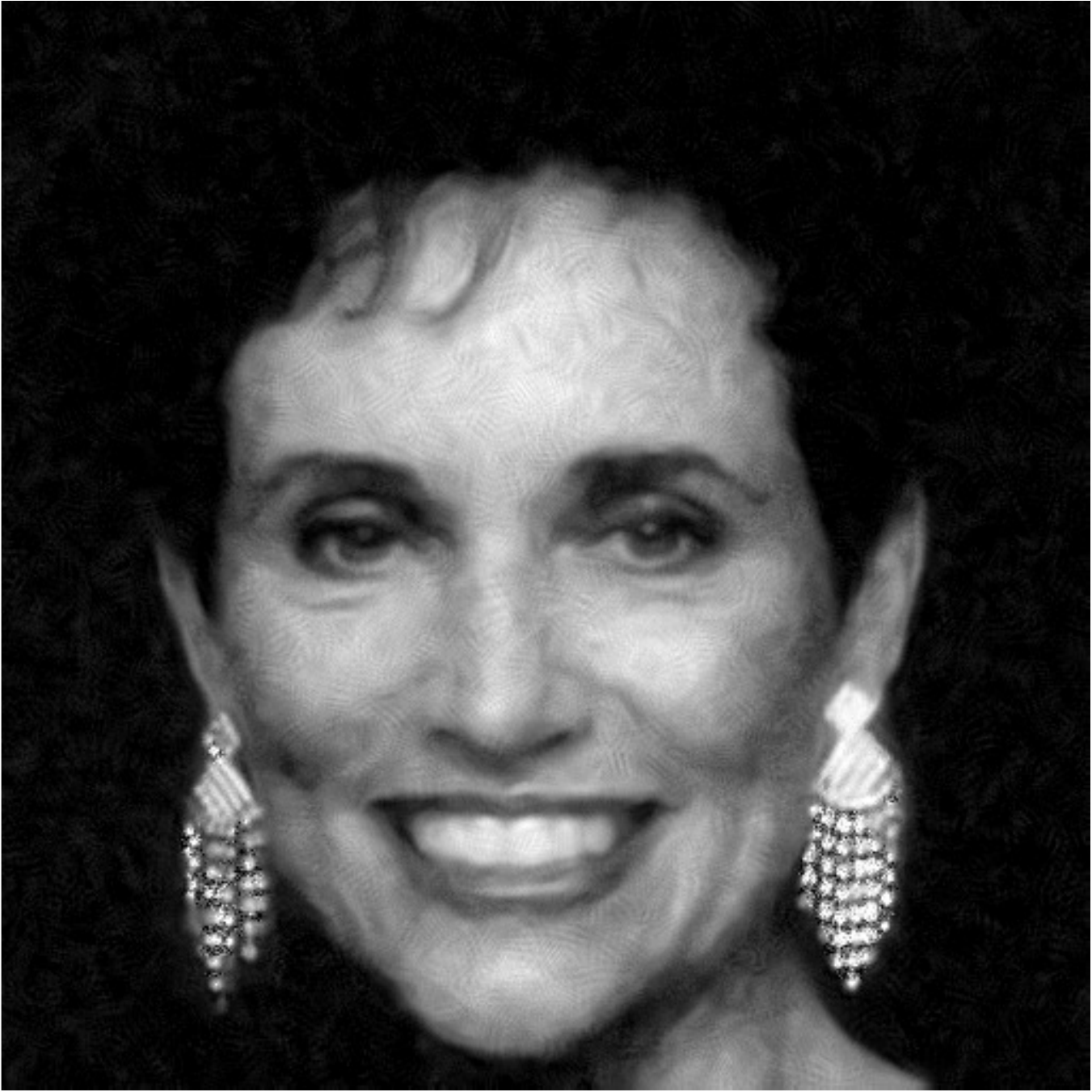}}

\subfigure[NLEM, difference, SOB = $0.043$]{
\includegraphics[height=0.28\textwidth]{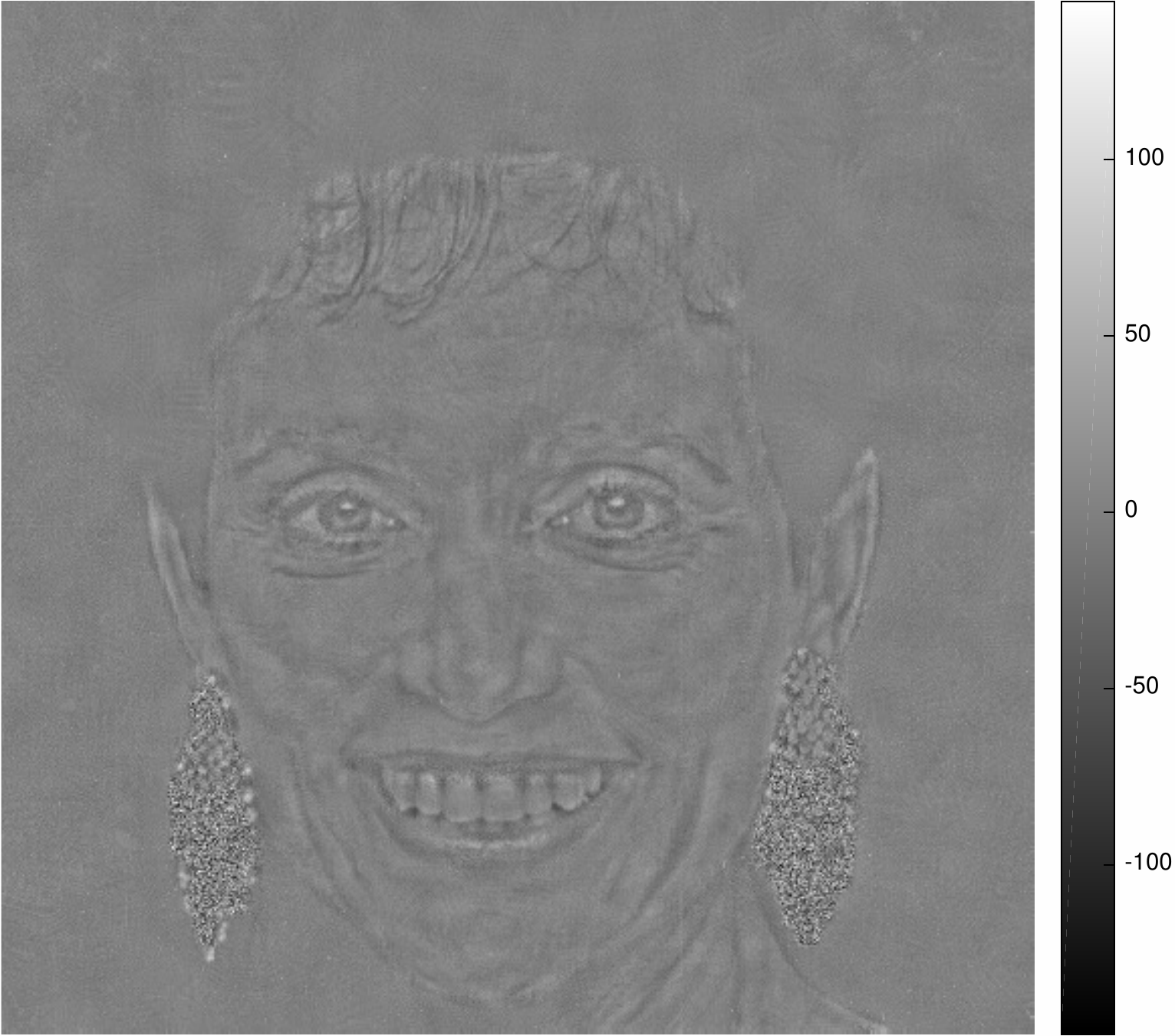}}
\subfigure[VNLEM, difference, SOB = $0.039$]{
\includegraphics[height=0.28\textwidth]{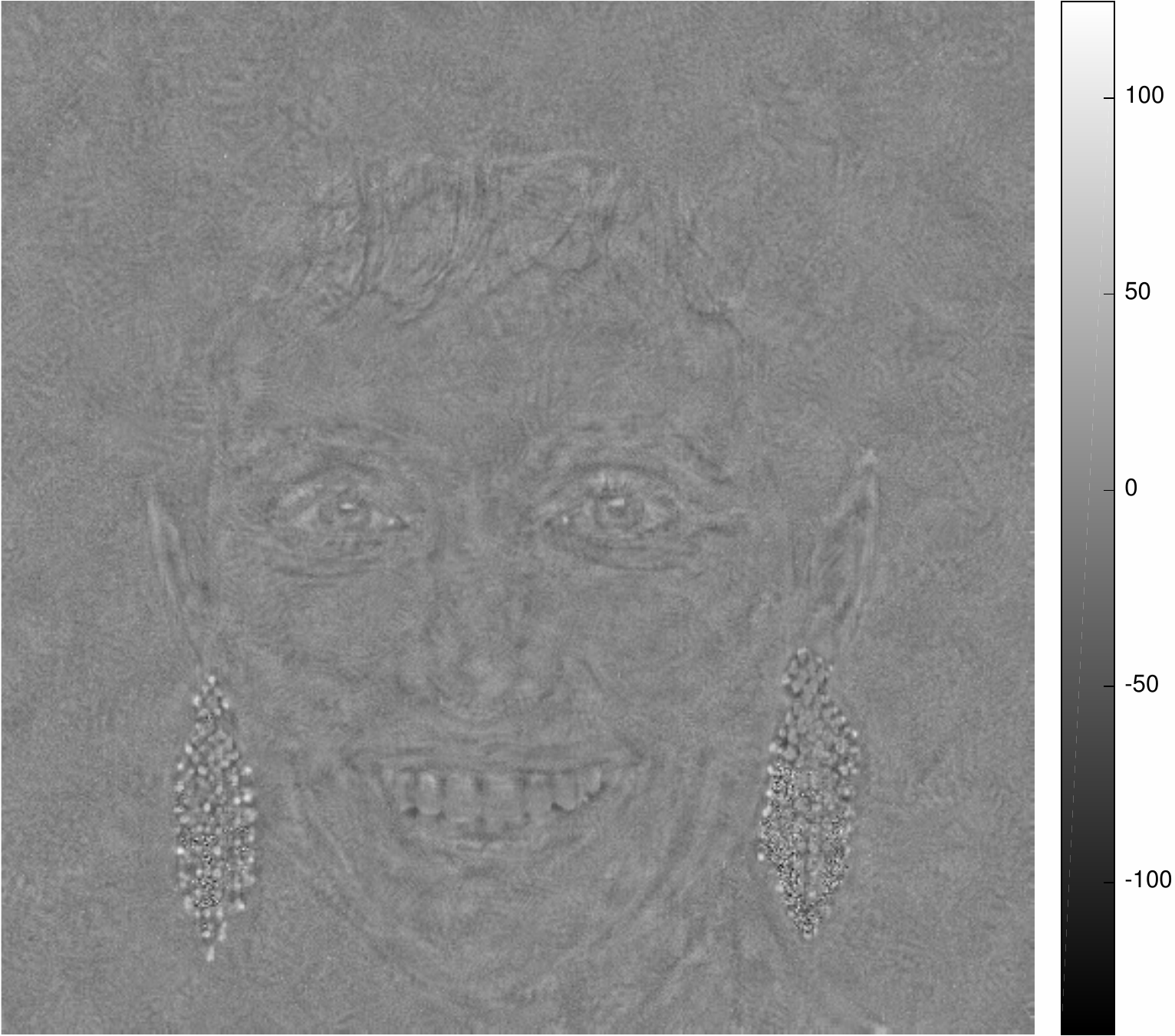}}
\subfigure[VNLEM-DD, difference, SOB = $0.038$]{
\includegraphics[height=0.28\textwidth]{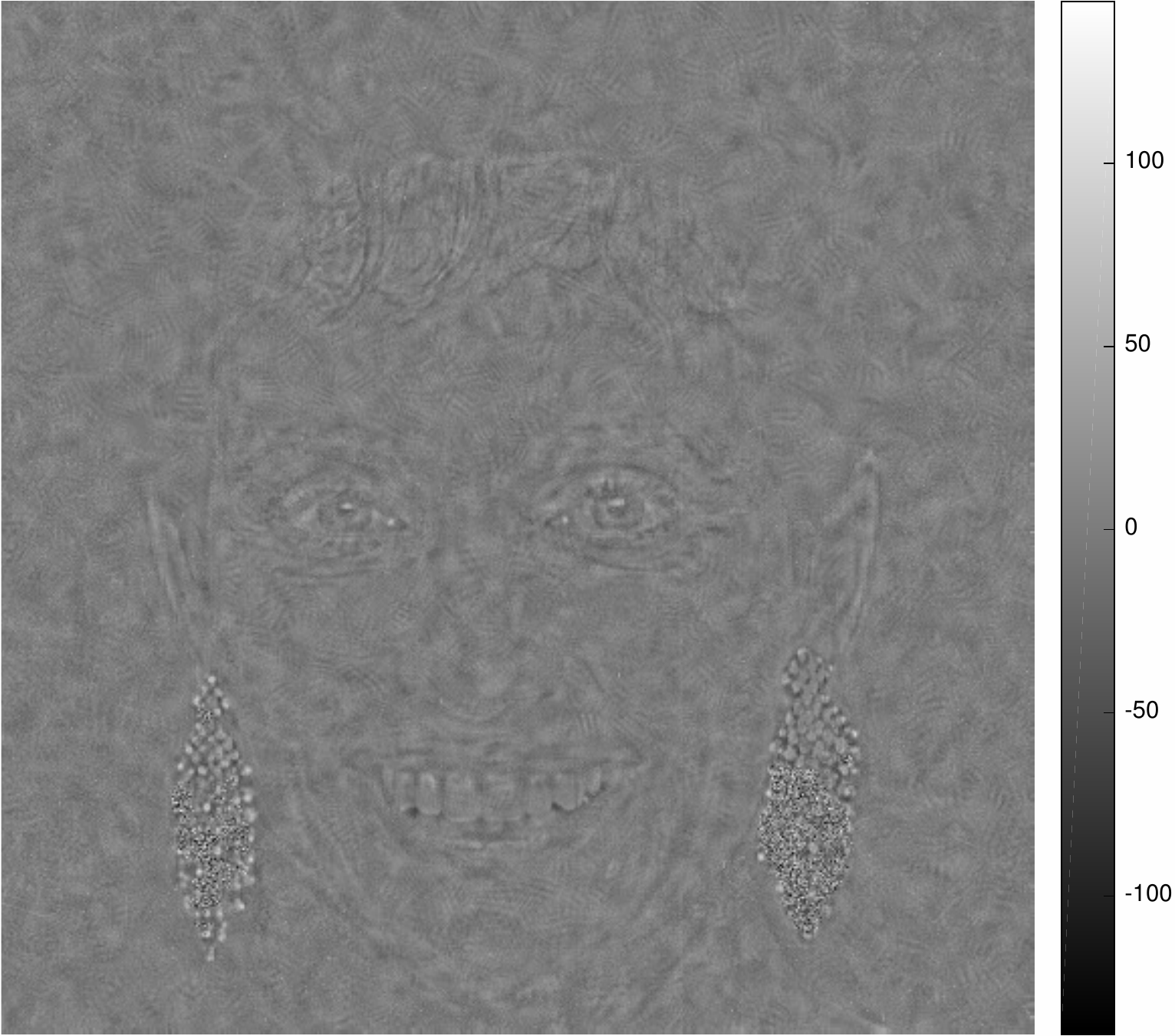}}

\caption{\label{lady}Example 3: the lady.}
\end{figure}

Another example worth looking into is presented in Fig.~\ref{clock}. This image contains substantial fine details that should be preserved during the recovery. Looking into results of the three different denoising schemes, one can see that in such an image, the NLEM scheme outperforms the VNLEM and VNLEM-DD in terms of SNR and PSNR as well as the level of details kept in the process, like the SOB and FSIM. This example shows that while the VNLEM and VNLEM-DD overall outperforms the NLEM statistically, there are examples where the NLEM performs better.

\begin{figure}[ht]
\centering
\subfigure[Original image]{
$\hspace{-0.5cm}$
\includegraphics[width=0.30\textwidth]{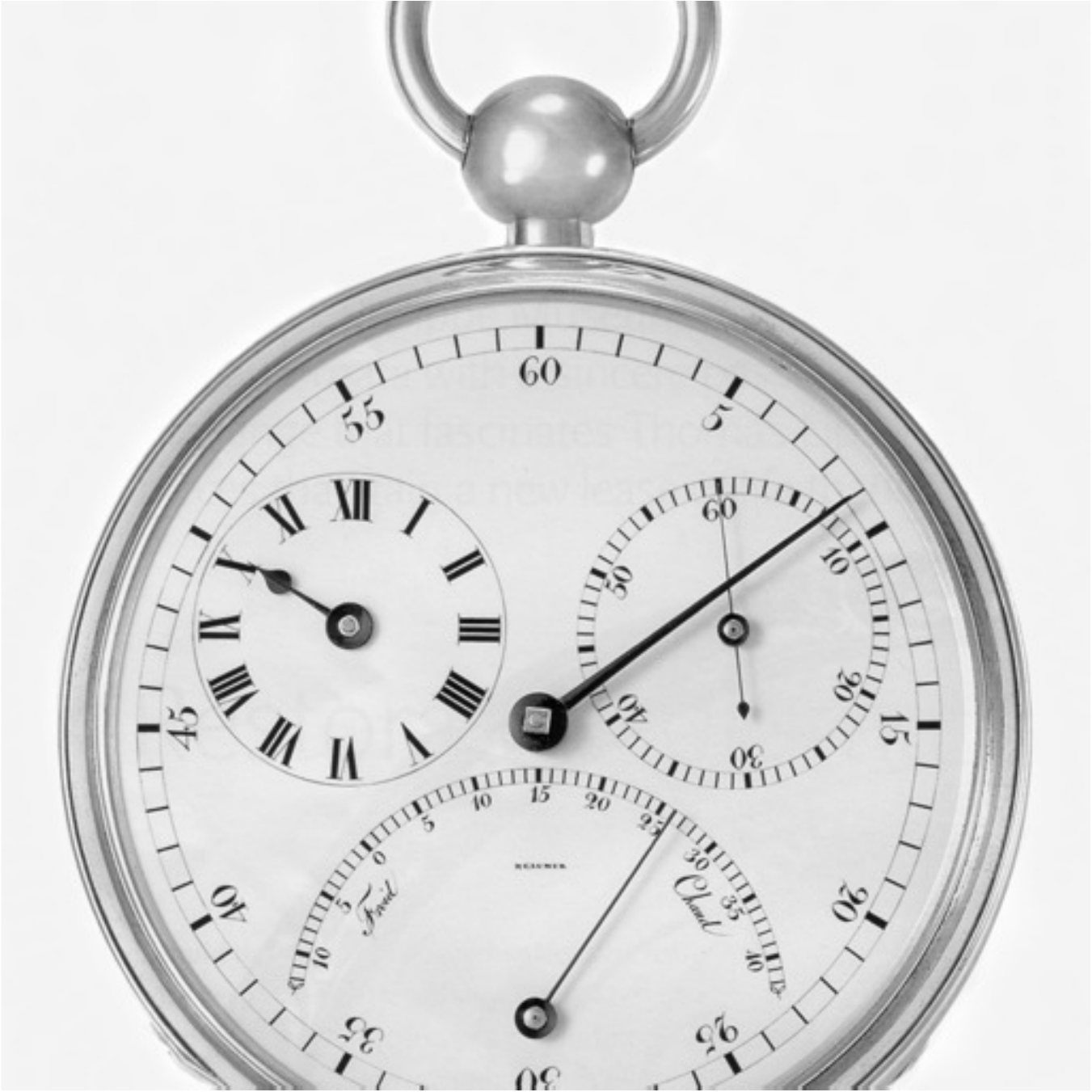}}
\subfigure[Noisy image, PSNR = $13.55$, SNR = $5.23$]{
\includegraphics[width=0.30\textwidth]{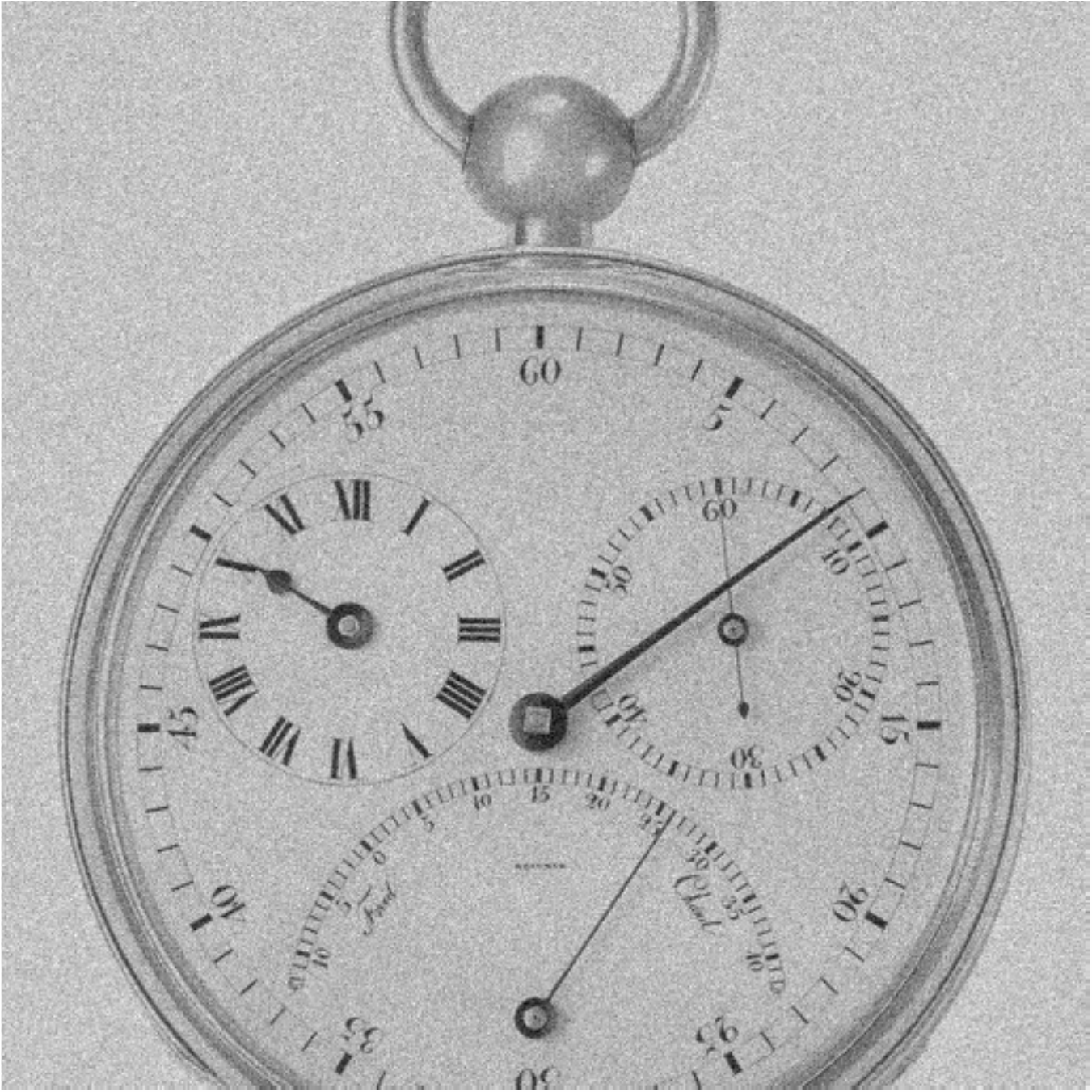}}

\subfigure[NLEM, PSNR = $25.68$, SNR = $17.36$, FSIM = $0.952$]{
\includegraphics[height=0.30\textwidth]{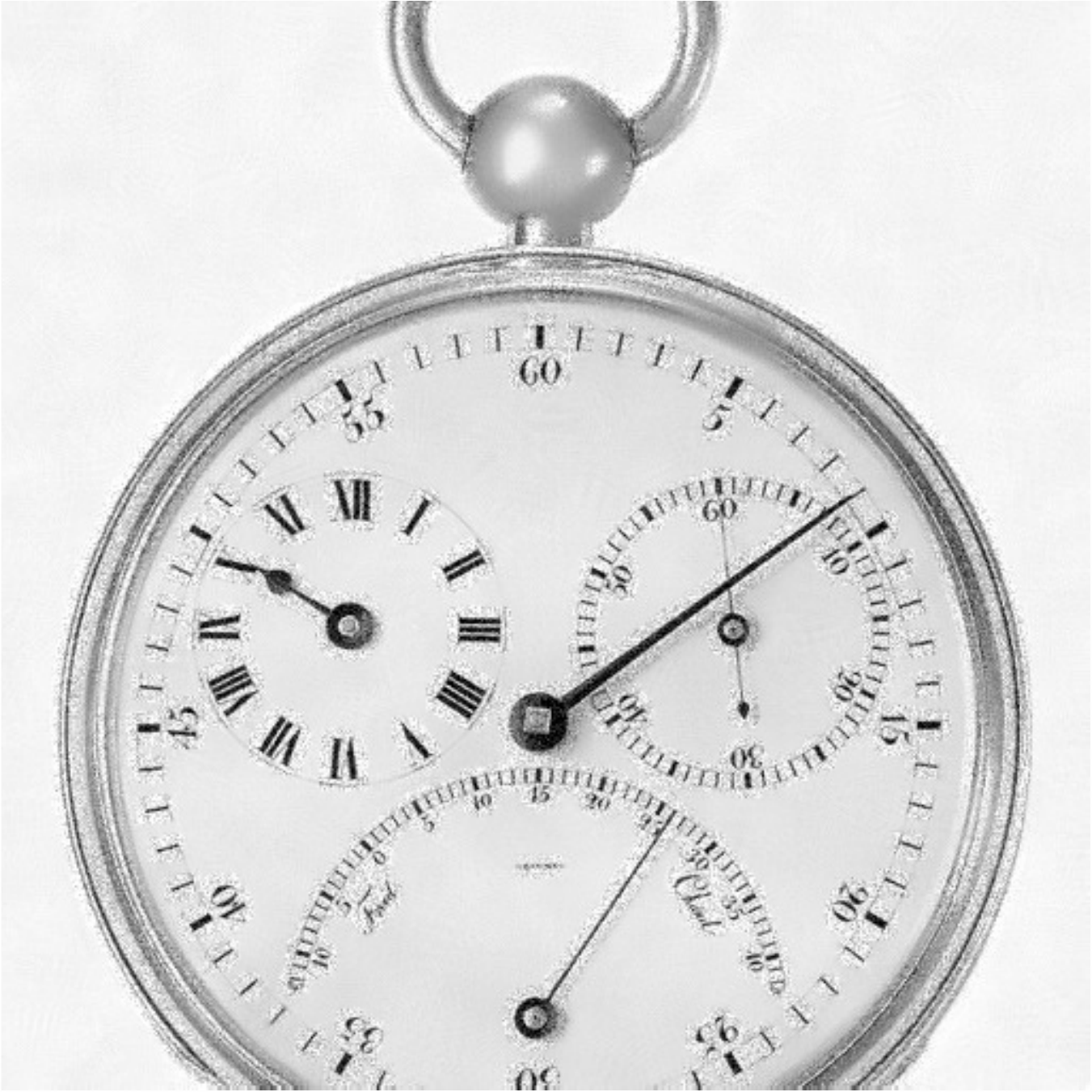}}
\subfigure[VNLEM, PSNR = $25.48$, SNR = $17.16$, FSIM = $0.944$]{
\includegraphics[height=0.30\textwidth]{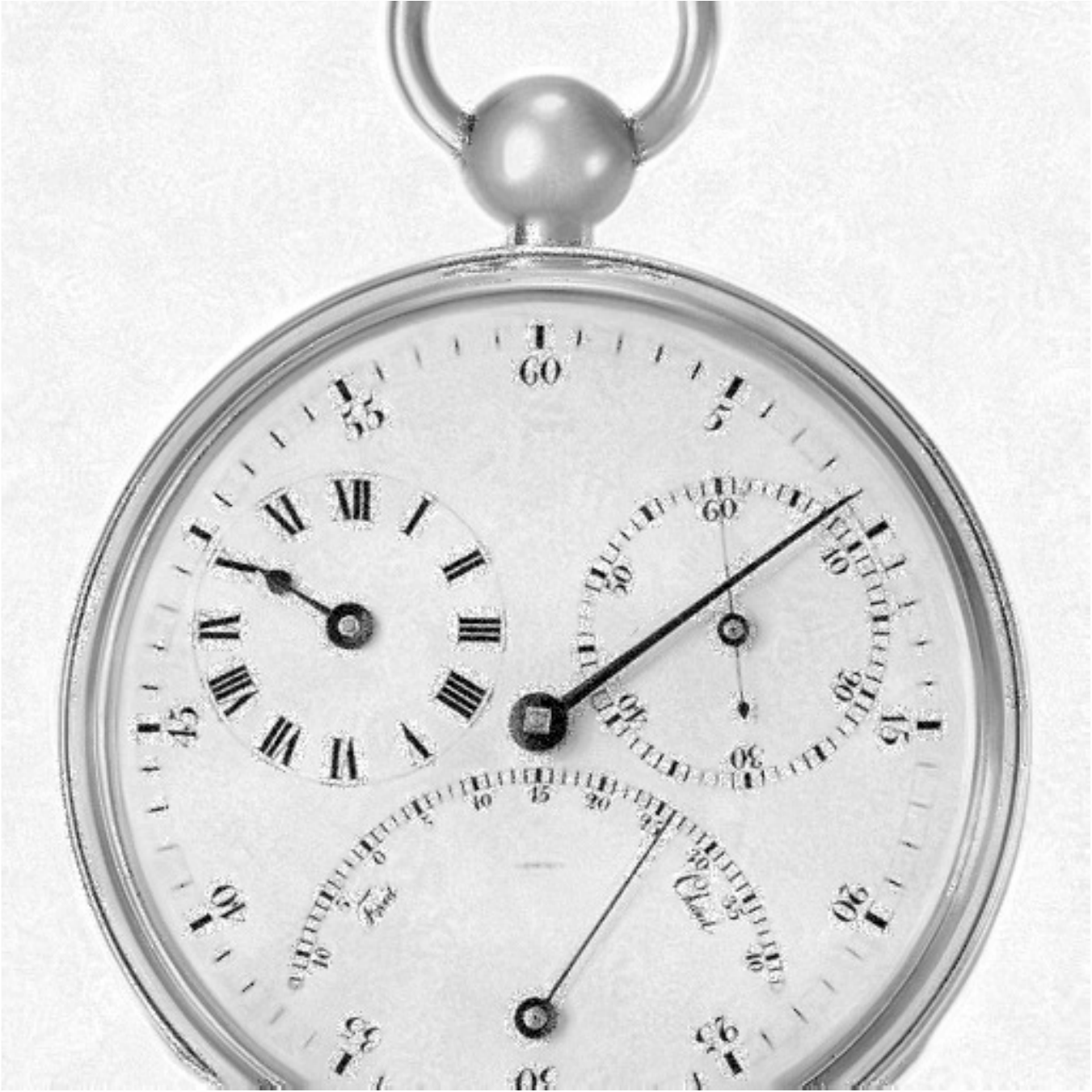}}
\subfigure[VNLEM-DD, PSNR = $25.54$, SNR = $17.21$, FSIM = $0.943$]{
\includegraphics[height=0.30\textwidth]{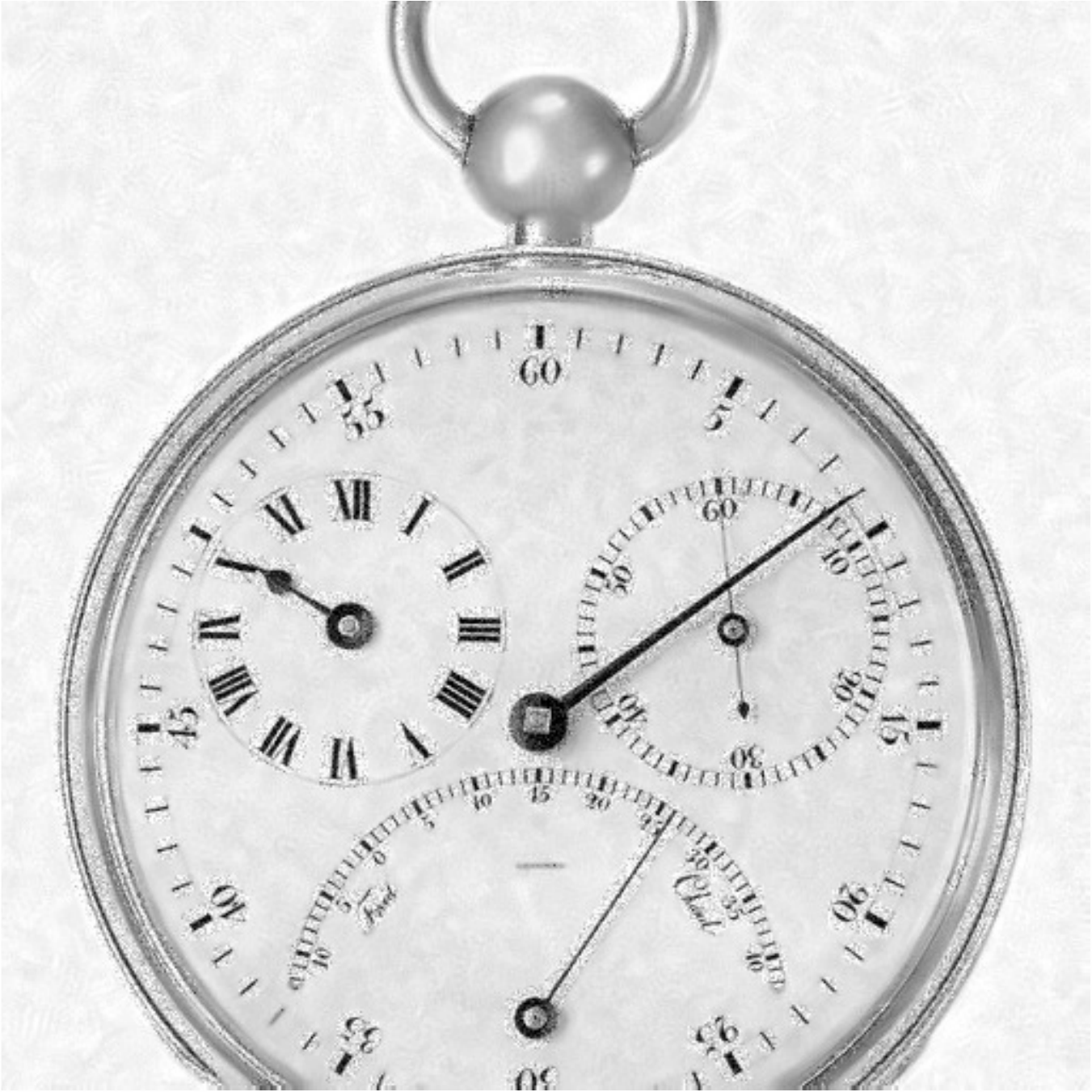}}

\subfigure[NLEM, difference, SOB = $0.033$]{
\includegraphics[height=0.28\textwidth]{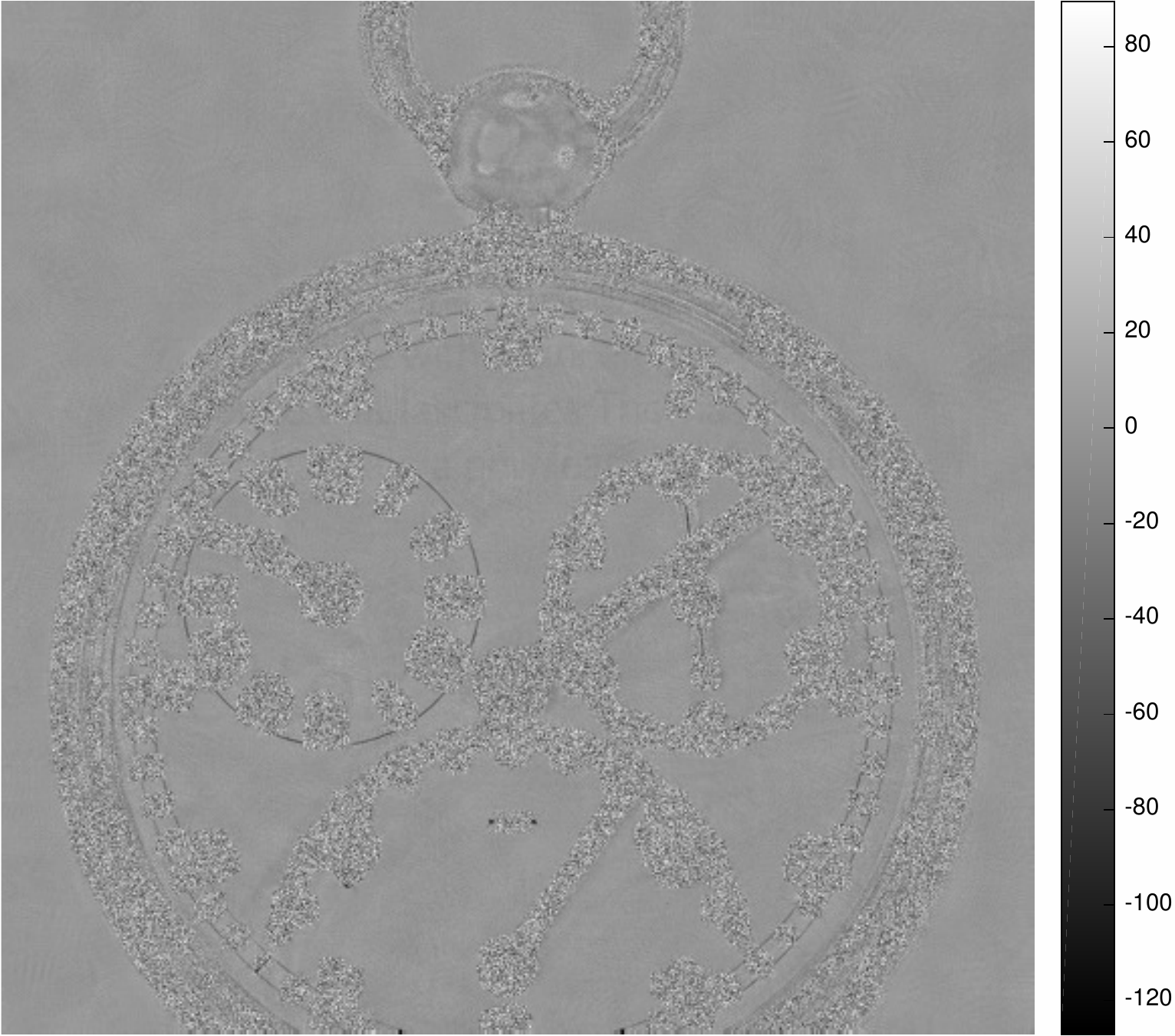}}
\subfigure[VNLEM, difference, SOB = $0.042$]{
\includegraphics[height=0.28\textwidth]{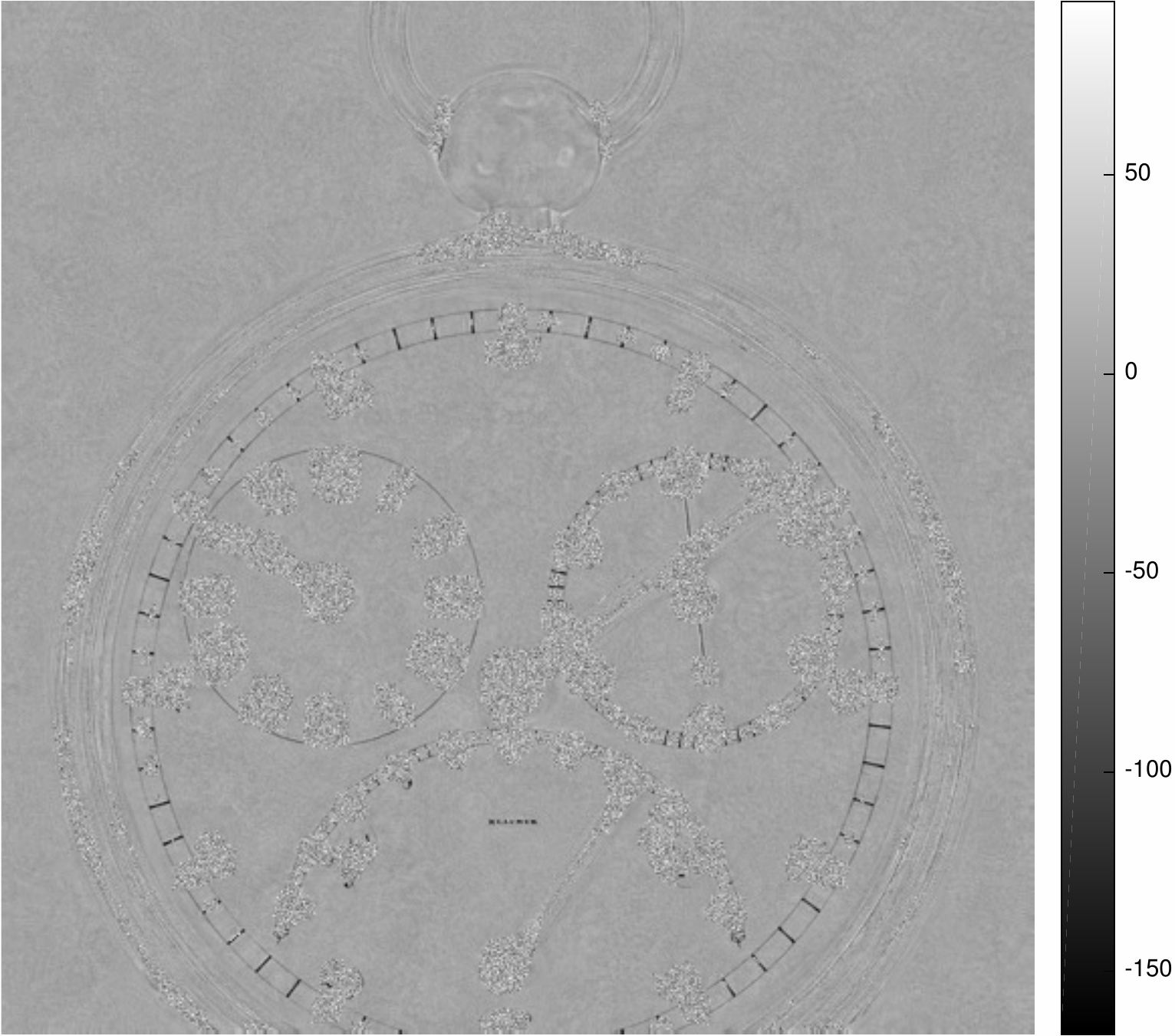}}
\subfigure[VNLEM-DD, difference, SOB = $0.044$]{
\includegraphics[height=0.28\textwidth]{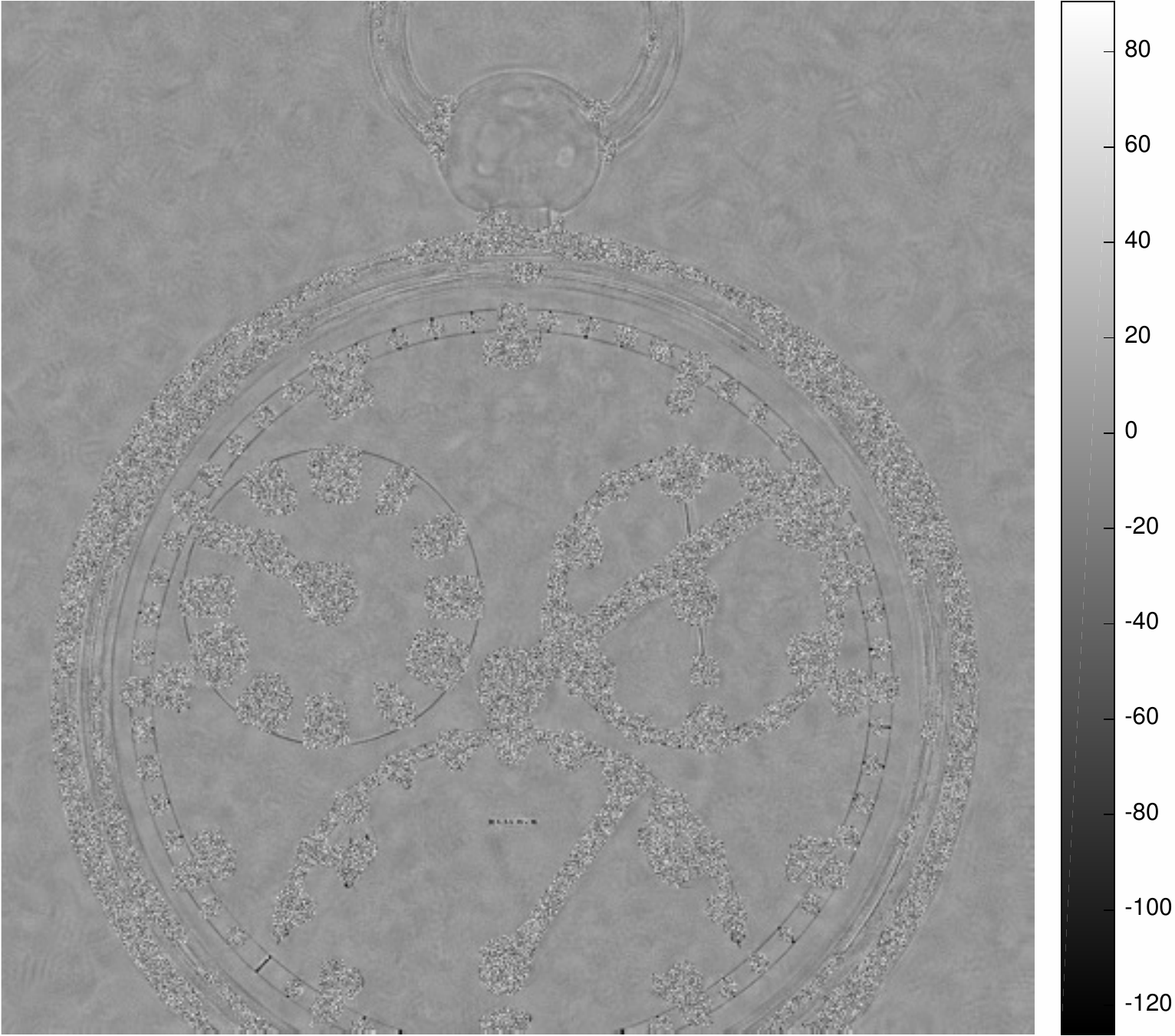}}

\caption{\label{clock}Example 4: the clock.}
\end{figure}

\subsection{The image resolution issue}
One interesting parameter that influences the image denoising performance is the ``image resolution''. Note that the ``image resolution'' is not a well-defined term, and in this example it means the number of pixels in the image -- the more pixels there are in an image, the higher the image resolution is. Equivalently, an image with a higher resolution means a denser sampling of the image function. In Fig.~\ref{resolution}, we take the starfish image shown in Fig.~\ref{starfish} and show how the image resolution affects the final result. In this figure, we present the outputs of the VNLEM-DD algorithm (left column) and the NLEM algorithm (right columns) for $N=200,512,1024$. It can be clearly seen that in all three cases, the VNLEM-DD algorithm produces a more clean image compared to the NLEM scheme, and the performance of each algorithm increases as the resolution increases.

\begin{figure}[ht]
\centering

\subfigure[NLEM, $N = 200$, PSNR = $14.90$, SNR = $8.05$, FSIM = $0.812$]{
\includegraphics[width=0.28\textwidth]{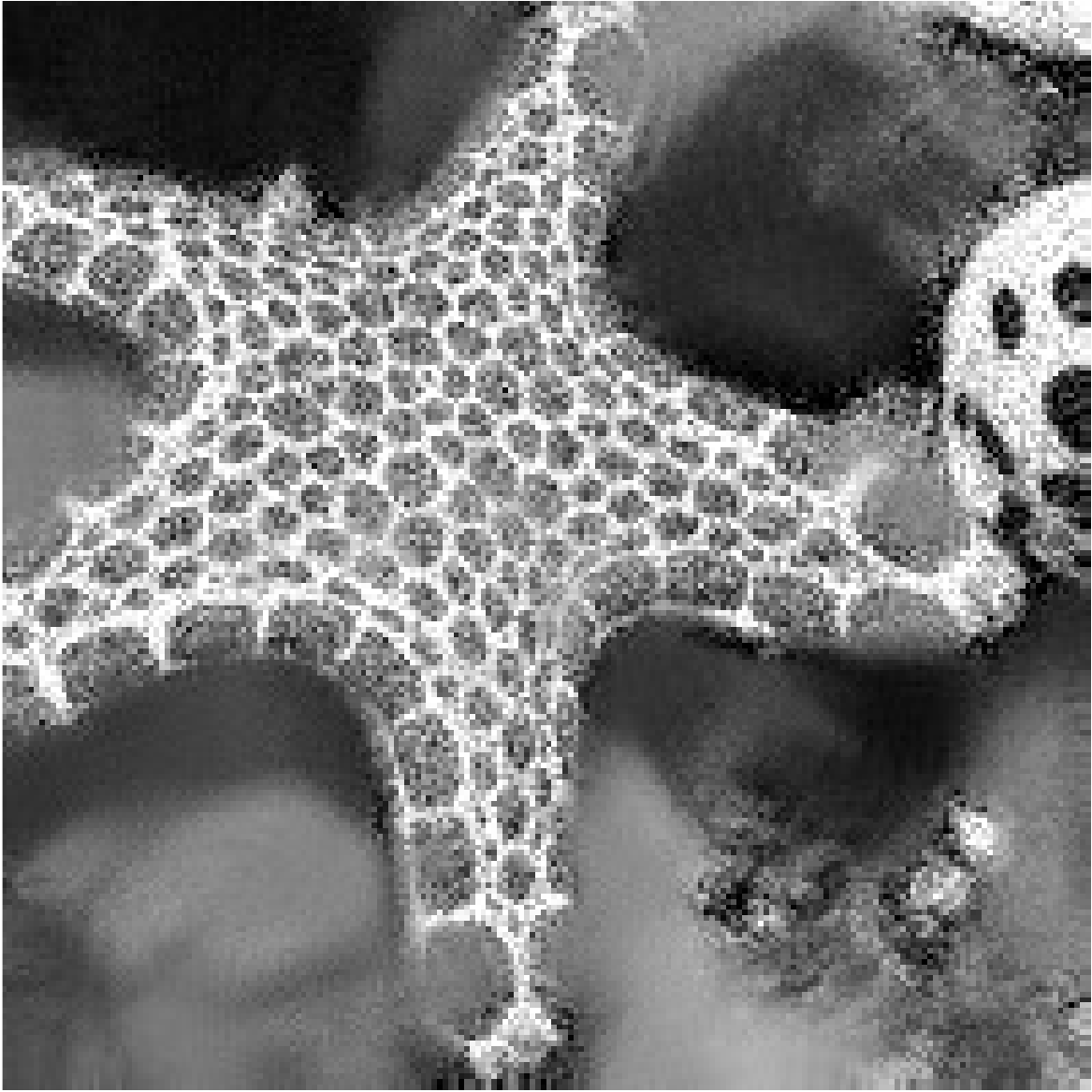}}
\subfigure[VNLEM, $N = 200$, PSNR = $16.79$, SNR = $9.93$, FSIM = $0.829$]{
\includegraphics[width=0.28\textwidth]{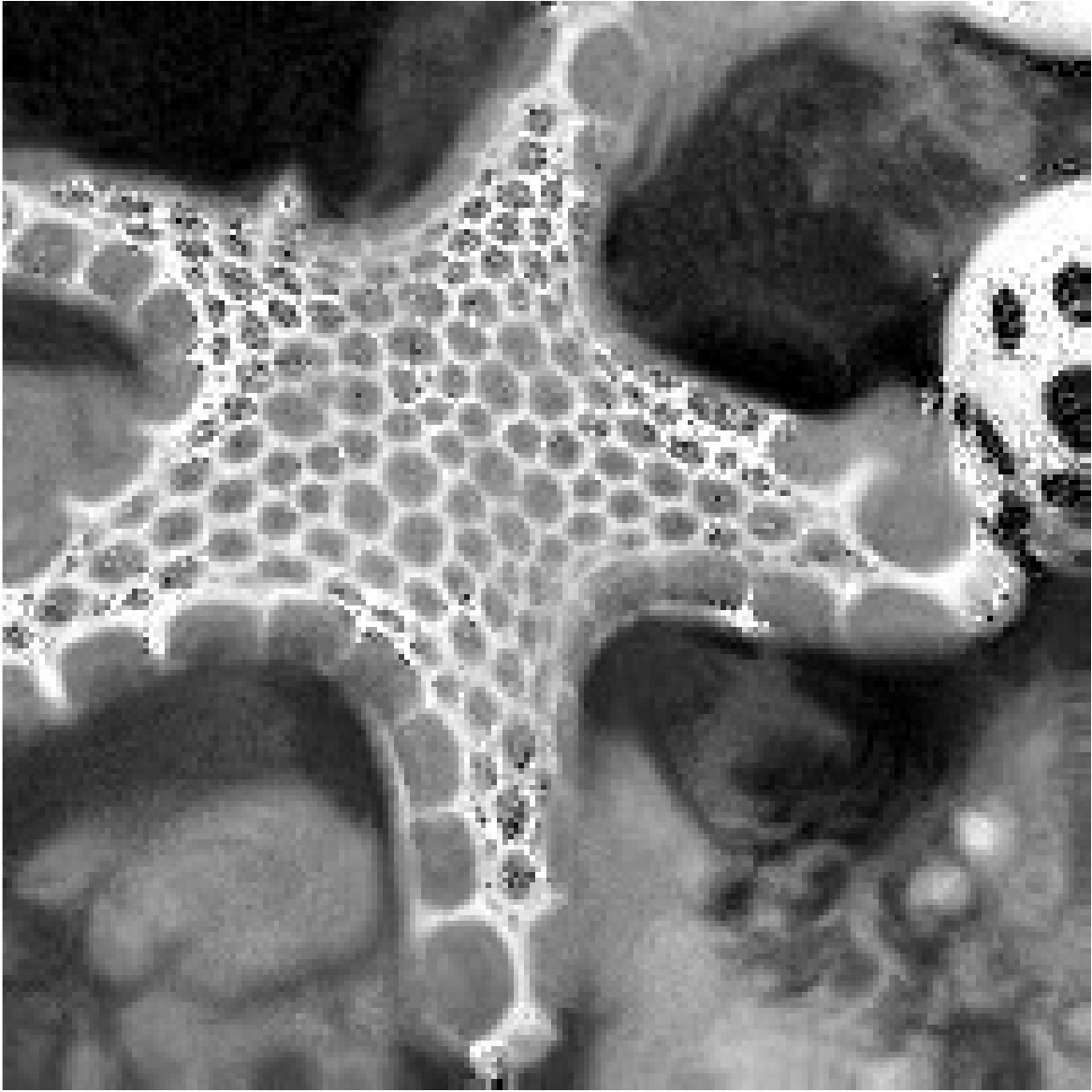}}
\subfigure[VNLEM-DD, $N = 200$, PSNR = $16.19$, SNR = $9.33$, FSIM = $0.819$]{\includegraphics[width=0.28\textwidth]{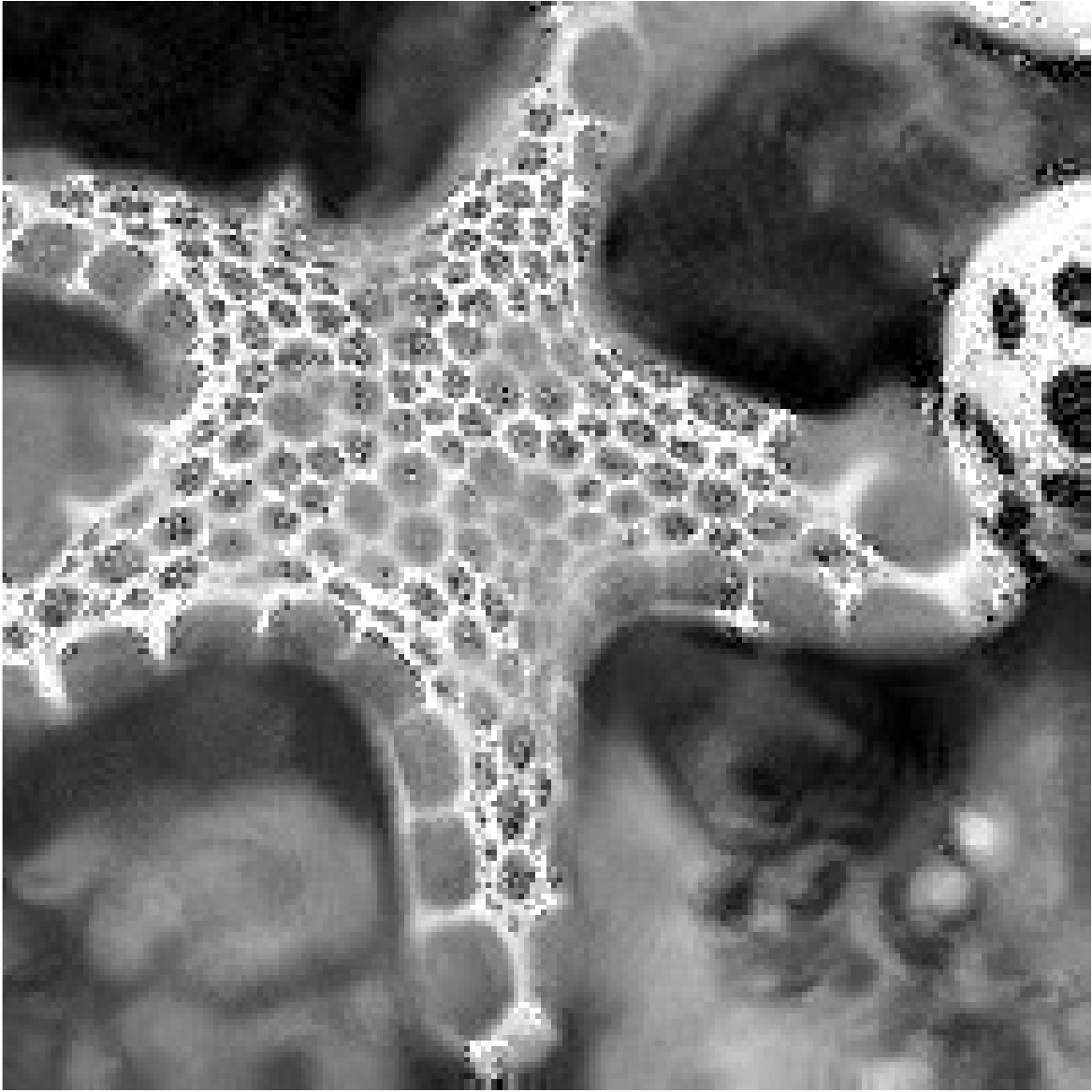}}

\subfigure[NLEM, $N=512$, PSNR = $16.79$, SNR = $10.43$, FSIM = $0.902$]{
\includegraphics[width=0.28\textwidth]{Img14_512_NLEM-eps-converted-to.pdf}}
\subfigure[VNLEM, $N=512$, PSNR = $19.68$, SNR = $13.32$, FSIM = $0.914$]{
\includegraphics[width=0.28\textwidth]{Img14_512_VNLEM-eps-converted-to.pdf}}
\subfigure[VNLEM-DD, $N=512$, PSNR = $19.49$, SNR = $13.13$, FSIM = $0.918$]{\includegraphics[width=0.28\textwidth]{Img14_512_DD-eps-converted-to.pdf}}

\subfigure[NLEM, $N=1024$, PSNR = $21.11$, SNR = $14.78$, FSIM = $0.953$]{
\includegraphics[width=0.28\textwidth]{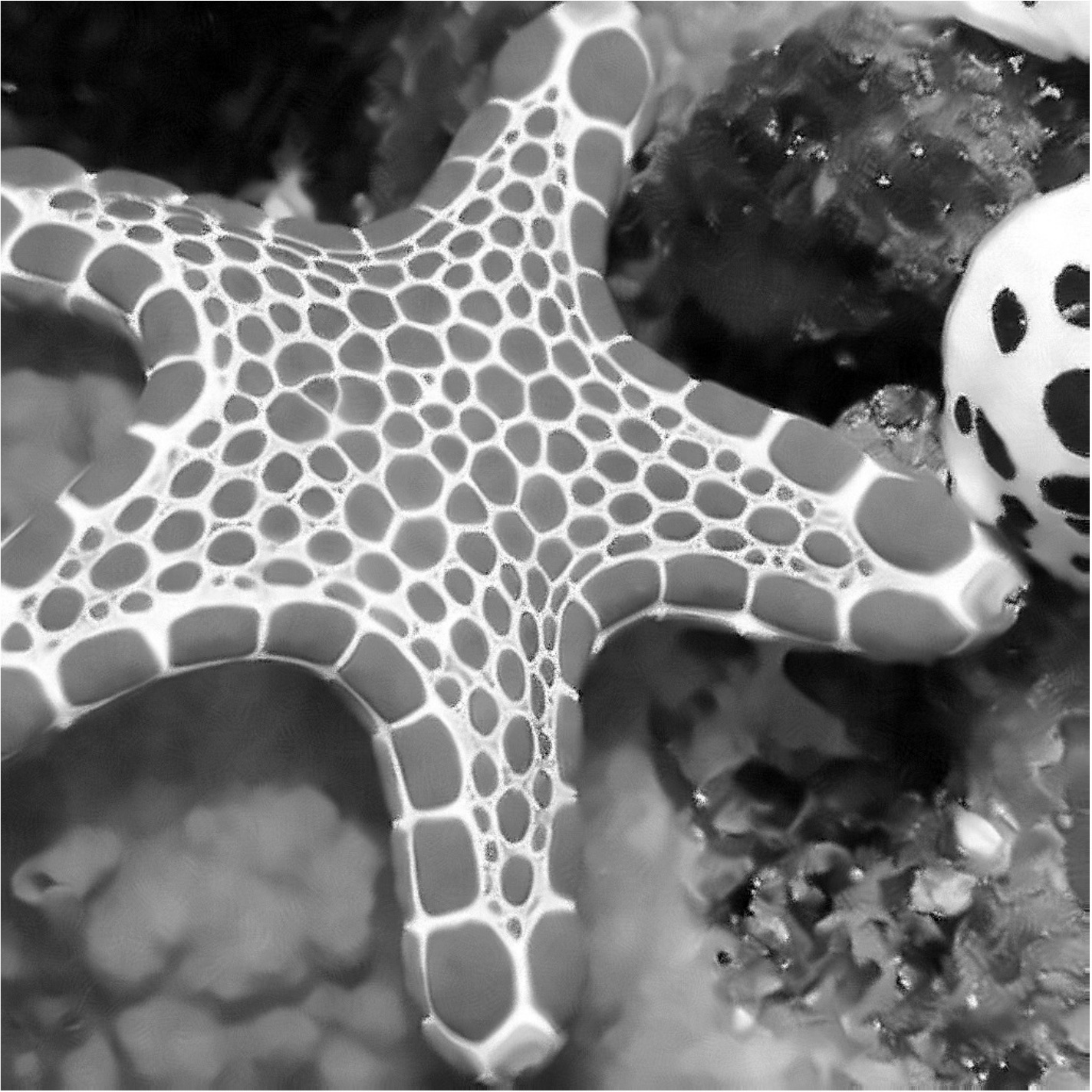}}
\subfigure[VNLEM, $N=1024$, PSNR = $22.50$, SNR = $16.16$, FSIM = $0.959$]{\includegraphics[width=0.28\textwidth]{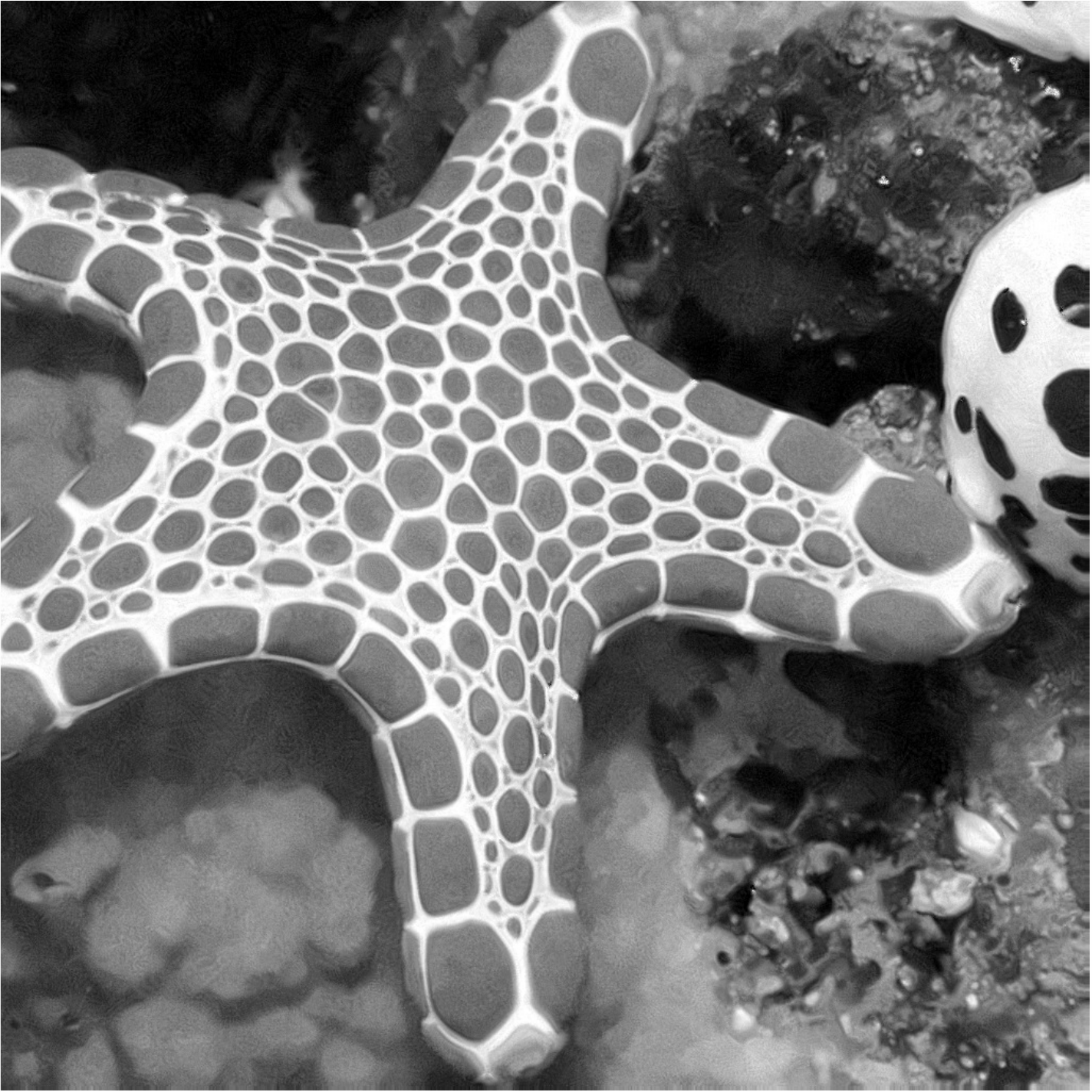}}
\subfigure[VNLEM-DD, $N=1024$, PSNR = $22.94$, SNR = $16.60$, FSIM = $0.961$]{\includegraphics[width=0.28\textwidth]{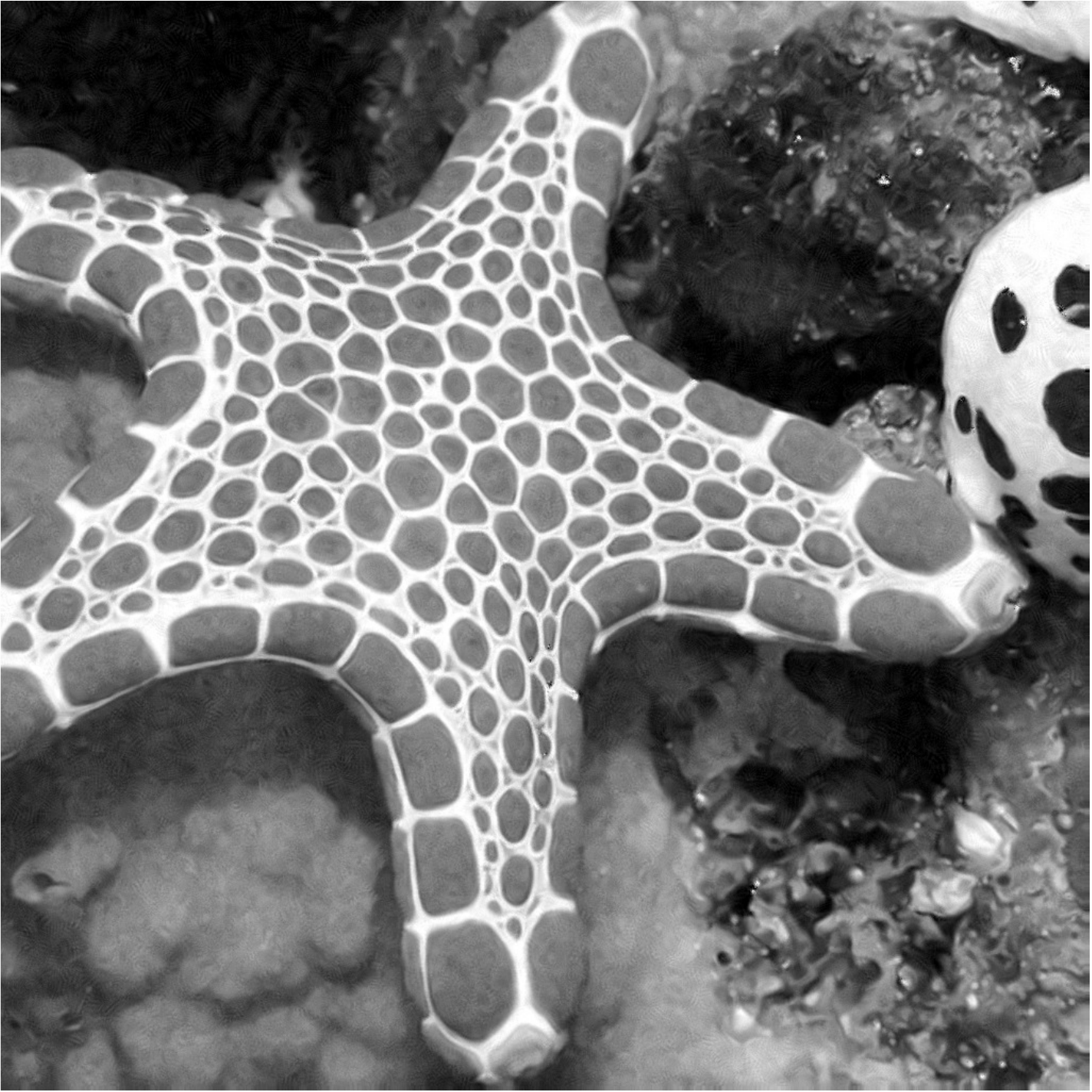}}

\caption{\label{resolution}The starfish in Figure \ref{starfish} with the same noise level but different resolutions.}
\end{figure}

\subsection{Application to the cytometry problem}

In this last example, we apply the developed VNLEM to the third-harmonic-generation (THG) microscopy image. The goal of the THG microscopy-based imaging cytometry is to automatically differentiate and count different types of blood cells with less blood \textit{ex vivo}, or even \textit{in vivo} \cite{Wu_Wang_Hsieh_Huang_Lin:2016}. One of the many strengths of THG is reflecting the granularity of leukocytes, which allows us to apply image processing techniques for the automatic classification. However, the raw data is noisy most of time, and a denoising technique is needed. We now apply the NLEM, VNLEM, and VNLEM-DD to the THG sectioning image of the whole blood smear at 1 hour post blood sampling. The data is provided by Professor Tzu-Ming Liu, Faculty of Health Sciences, University of Macau. The result is shown in Figure \ref{cytometry}. Note that since we do not have the ``ground truth'' for a comparison, we only show the FSM for the quality evaluation purpose. Note that while the result is encouraging, a systematical study of the problem, and a systematic comparison of the proposed algorithm with existing algorithms is needed. The result will be reported in a future work.

\begin{figure}[ht]
\centering

\subfigure[Original image]{
\includegraphics[width=0.3\textwidth]{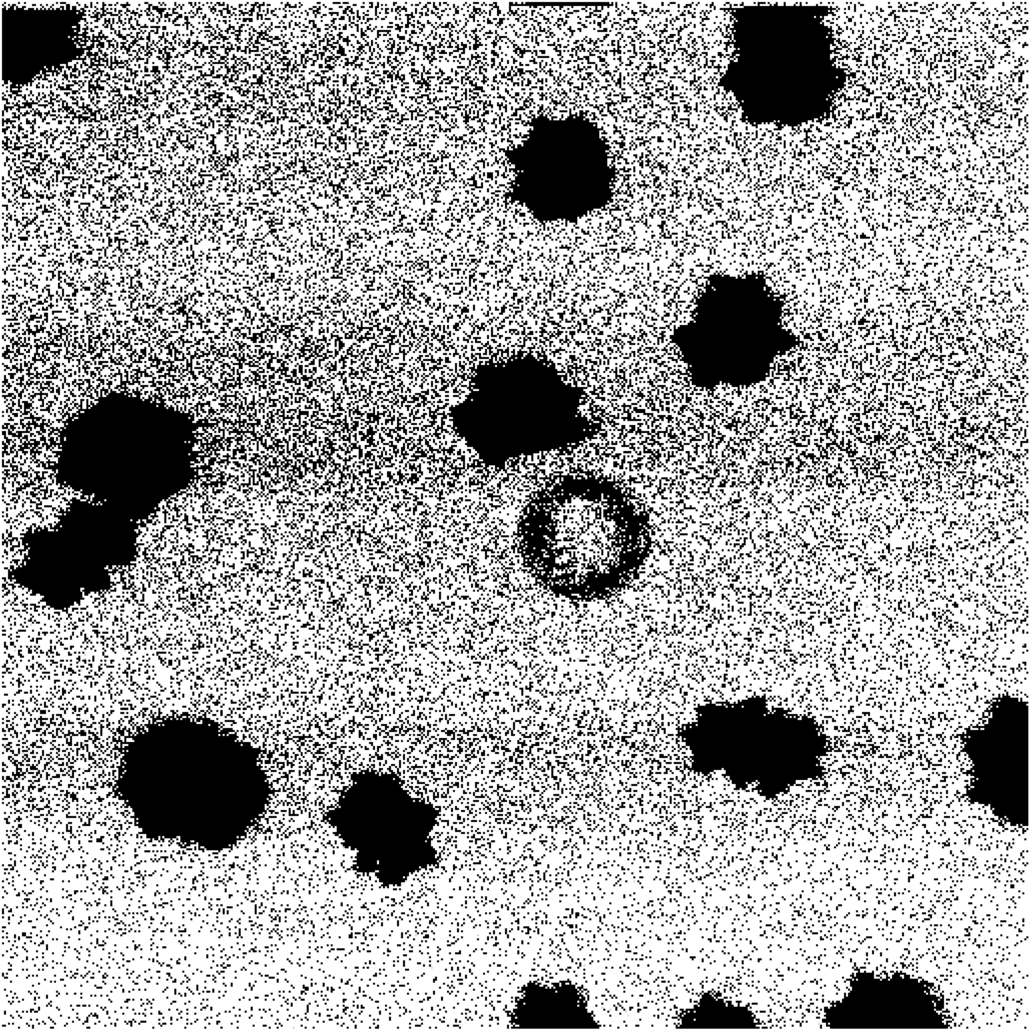}}

\subfigure[NLEM, FSIM = $0.947$]{
\includegraphics[width=0.3\textwidth]{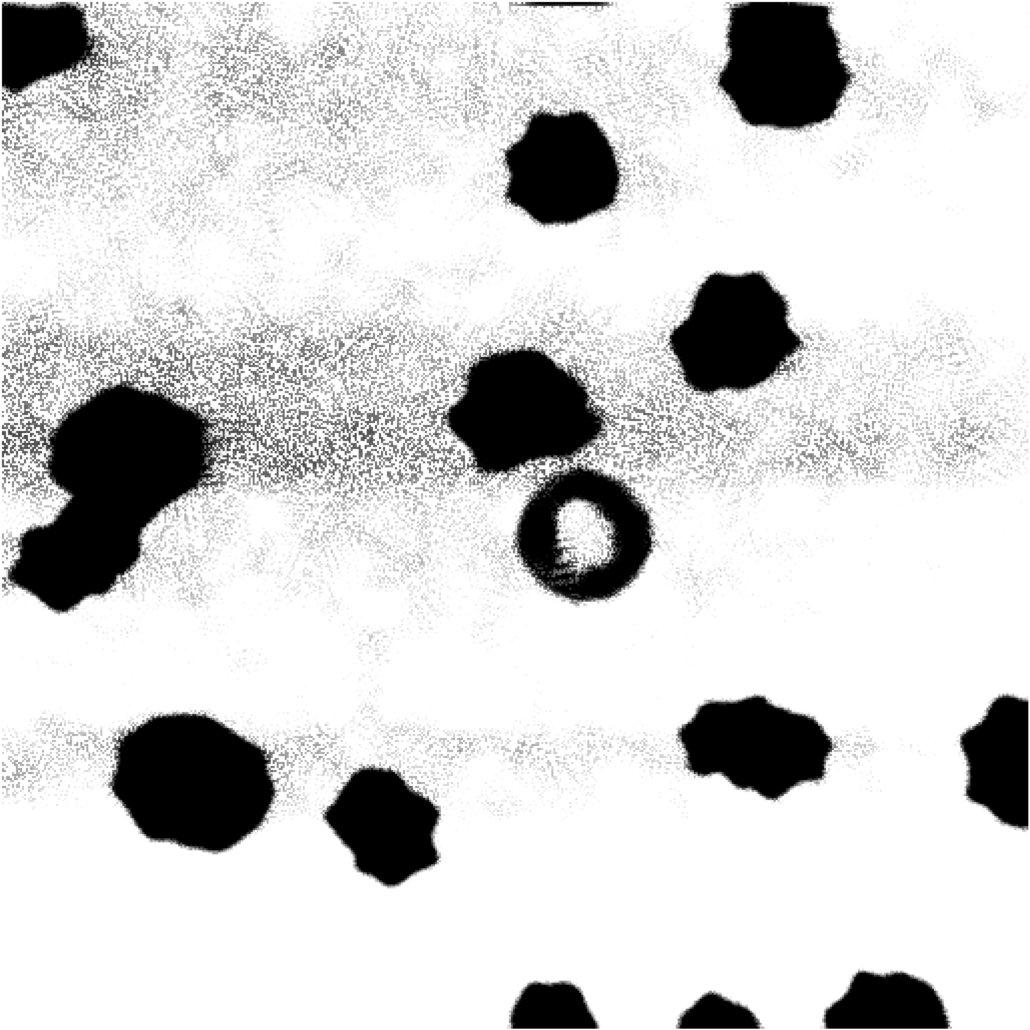}}
\subfigure[VNLEM, FSIM = $0.967$]{
\includegraphics[width=0.3\textwidth]{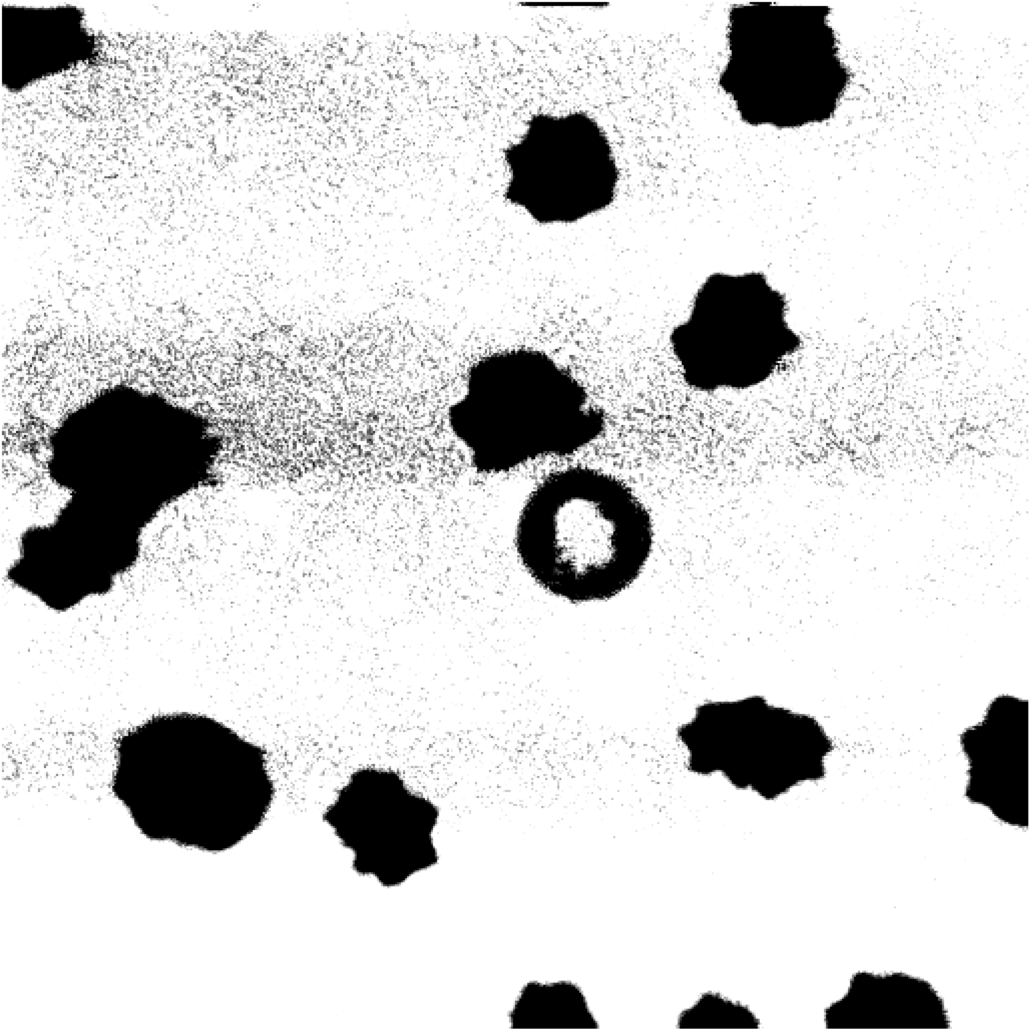}}
\subfigure[VNLEM-DD, FSIM = $0.97$]{
\includegraphics[width=0.3\textwidth]{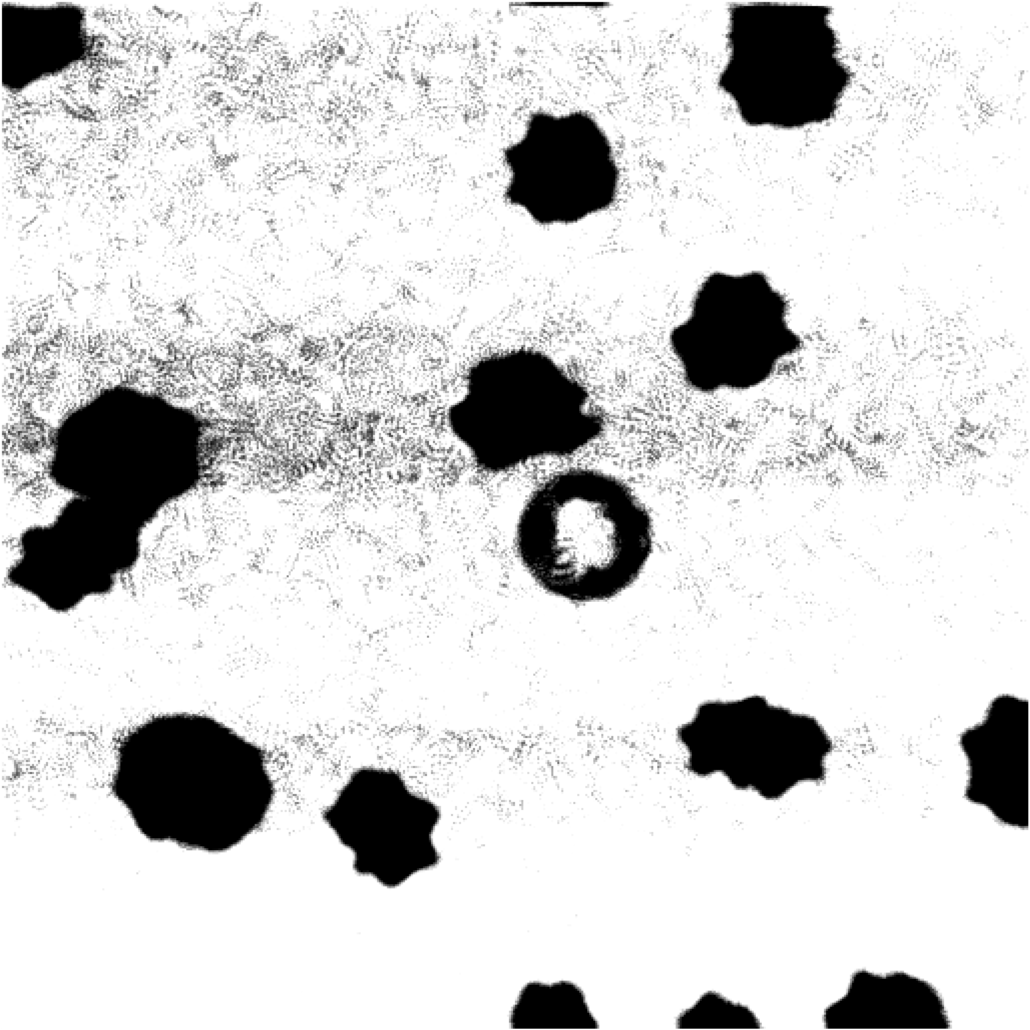}}

\caption{\label{cytometry}The cytometry image. Since the ``ground truth'' is not available for a comparison, we only show the FSIM for the quality evaluation purpose.}
\end{figure}

\section{Conclusion and Discussion}

In this work, we propose a fiber bundle structure to model the patch space, and take the fiber structure to generalize the commonly used NLEM algorithm to the VNLEM/VNLEM-DD algorithm. One main benefit of introducing the fiber structure is the dimension reduction. To speed up the VNLEM algorithm and stabilize the numerical rotation on a small patch, different numerical techniques are applied, including the search window and the SIFT features. The numerical simulation provides positive evidence of the potential of the proposed algorithm. In addition to providing the theoretical justification of how the VNLEM and the NLEM work, particularly why we could accurately find nearest neighbors from the noisy patches, we study the stability of the widely applied SIFT algorithm. Both theoretical results support how the proposed VNLEM algorithm works. The potential of the proposed model, algorithms, and the associated theory are statistically supported by a large image database composed of 1,361 images. 
Below, we discuss the limitations of the current work and several future works. 

First, the computational complexity needs to be further improved. Note that the main difference between the VNLEM and the NLEM algorithms is the chosen metric. In the NLEM, since the $L^2$ distance is chosen to evaluate the similarity of two patches, there are several fast algorithms available to evaluate the nearest neighbors. However, in the VNLEM, there does not exist a fast algorithm to determine the nearest neighbors with respect to the RID, to the best of our knowledge. Although we have delegated the problem of evaluating the RID distance to that of evaluating the SIFT distance, the numerical performance still has a significant room for improvement. 

Second, although the manifold model has been widely accepted in the field, and our algorithm is also based on the manifold structure, 
it is certainly arguable if in general a patch space could be well approximated by a manifold.
On one hand, we need to consider a more general model than the fiber bundle; on the other hand, for different problems we may want to better understand its associated manifold structure, if there is any. In other words, we might need different models, and hence different metrics, for different kinds of images. For example, while the RID helps reduce the dimension of the patch space of a ``structured'' image, its deterministic nature might render it unsuitable for analyzing a ``texture'' image, since the texture features are stochastic in nature. In short, it might be beneficial to take the metrics designed for the texture analysis into account. 
On the other hand, we could consider to segment the given noisy image into different categories, and run the VNLEM on each category. This segmentation step is related to the ``multi-manifold model'' considered in the literature \cite{Yang_Sun_Zhang:2011,Wang_Slavakis_Lerman:2014}, and could be understood as a generalization of the search window method used in this paper.

Third, we should consider different structures in the denoising procedure. In addition to taking the rotation group to fibrate the patch space, it is an intuitive generalization to further consider other groups, like the dilation group or even the general linear group. Also, while the current work focuses on grayscale images, the proposed algorithm has the potential to be generalized to colored images. In colored images, more structures, like the color space, will be taken into account. Furthermore, in practice we would expect to have more than one image from the practical problem. Under the assumption that the noise behavior is similar, it is of great interest to see if we could further improve the denoise performance by denoising multiple available images simultaneously.

Fourth, note that the proposed algorithm could be understood as aiming to reduce the error introduced to the clean image. However, it has been widely argued in the IQA society that simply reducing the error might not lead to the optimal result in all scenarios. It might be more important to take the human perception into account, if the images are meant to be watched by a human being. While the proposed VNLEM provides a satisfactory result by the FSIM evaluation, note that the ``features'' considered in the FSIM are not used in the algorithm. It is reasonable to expect that by taking these features into account, we could further improve the result.

Fifth, in this paper we focus only on comparing our algorithm with the NLEM to study the corresponding diffusion property and the geometric structure of the underlying patch space model. For the image denoising purpose, there are several other image denoising algorithms available in the field, and we will do a systematic comparison in a upcoming report. For example, while not specifically indicated, the widely used algorithm block-matching and 3-D filtering (BM3D) \cite{Dabov2007} and its generalizations, for example \cite{Katkovnik2010}, are also based on the patch space model. We could view the sparsity structure used in BM3D as a different way to design a ``metric'' to compare different patches. 

Last but not least, although we compared the algorithm on a big image database and reported the statistical significance, note that statistical significance does not imply practical significance. Particularly, the included images are not exhaustive. A more systematic comparison is thus needed. In practice, the overall performance might depend on the problems encountered, and the specific applications, like the cytometry problem, will be discussed in a upcoming research report.

\section*{Acknowledgement}

Hau-tieng Wu's research is partially supported by Sloan Research Fellow FR-2015-65363 and partially by Connaught New Researcher grant 498992. He would like to thank the valuable discussions with Professor Ingrid Daubechies and Professor Amit Singer. Chen-Yun Lin would like to thank Professor Chiahui Huang for her helpful discussions. The authors would like to thank Professor Tzu-Ming Liu for sharing the cytometry image.

\bibliography{reference}
\bibliographystyle{plain}

\appendix

\section{Diffusion Map}
\label{sec:DM}

To make the paper self-contained, we summarize the DM algorithm here. 
DM were initially introduced in~\cite{Coifman20065} as a means to extract feature and reduce the dimensionality. This mapping embeds the points from the original data set, which might be high-dimensional, into a low-dimensional Euclidean space so that the geometric properties of the original dataset are less distorted. The coordinates of the embedded points are derived from the eigenvectors and eigenvalues of the \textit{transition matrix} of the graph Laplacian associated with the data set. Below we summarize the embedding procedure. For a detailed algorithm description and a summary of the existing theorems describing the asymptotical behavior of DM, we refer the interested reader to, for example, the online supplementary of  \cite{ElKaroui_Wu:2015b}.

For a give point cloud $\mathcal{X}=\{x_i\}_{i=1}^n\subset \mathbb{R}^n$, we construct an affinity graph $(V, E,w)$, where $V := \{x_1, x_2, \ldots, x_n\}$, $E$ is the set of edges that is determined by the user, and $w:E\to \mathbb{R}^+$ is the affinity function defined by the user. Usually $w$ is defined as $w_{ij}=K(\|x_i-x_j\|)$ when $(i,j)\in E$, where $K$ is a chosen kernel, and $w_{ij}=0$ when $(i,j)\notin E$.
With the affinity graph, 
we have an equivalent expression of the affinity function as the $n \times n$ {\it affinity matrix}, defined as
\begin{align}
W_{i j} = \left \{
\begin{array}{ll}
w(i,j) & \mbox{ if } (x_i, x_j) \in E
\\
0 & \text{otherwise}
\end{array}
\right..
\end{align}
We then consider the {\it transition matrix} 
\be
\label{transitionMatrix}
A = D^{-1}W
\ee
where
$D$ is the \textit{degree matrix} defined as a $n \times n$ diagonal matrix defined as 
\be
D_{ii} = \sum_{j=1}^n W_{ij}.
\ee 
Note that though $A$ may not be symmetric in general, when $D$ is not singular, $A$ is similar to $D^{-1/2} W D^{-1 / 2}$ which is symmetric and thus diagonalizable. More specifically, there exists a diagonal matrix $\Lambda\in \mathbb{R}^{n\times n}$ and an orthogonal matrix $Q\in O(n)$ such that $D^{-1/2} W D^{-1 / 2} = Q \Lambda Q^T$, where $\Lambda = \text{diag}\{\lambda_1, \lambda_2, \ldots, \lambda_n\}$ is the matrix of eigenvalues such that $\lambda_1  \geq \lambda_2  \geq \ldots  \lambda_n \geq 0$ as $W \geq 0$.
Therefore, we can write $A$ as 
\begin{align}
A =  U \Lambda V^T
\end{align}
where $U = D^{-1/2} Q$ and $V = D^{1 / 2} Q$ and their column vectors are called right and left eigenvectors of $A$, respectively. In this work, we assume that $D$ is not singular. 
With the above preparation, we could define the DM and diffusion distance (DD). Take a diffusion time $t>0$. The DM $\Phi_t: V \rightarrow \RR^{m}$ is defined as
\be
\Phi^{(m)}_t(i) = (\lambda_2^t \phi_2(i), \lambda_3^t \phi_3(i), \ldots, \lambda_{m+1}^t \phi_{m+1}(i)),
\ee
where $\phi_1,\phi_2, \cdots, \phi_n$ are the column vectors of $U$ and $m\in \mathbb{N}$ is determined by the user. 
We could view $\Phi^{(m)}_t(i)$ as a new feature representing $x_i$.
The DD between $x_i$ and $x_j$ in $\mathcal{X}$ with diffusion time $t>0$ is then defined as 
\be
D^{(m)}_t(i,j) := \| \Phi^{(m)}_t(i) - \Phi^{(m)}_t(j)\|.
\ee
The DD could be view as a new metric on the dataset. It has been shown in \cite{ElKaroui:2010a,ElKaroui_Wu:2015b} that the DD is robust to ``big'' noise, and hence suitable for us to suppress the influence of inevitable noise in our denoising problem.

\section{Why could we approximate the patch space by a manifold?}\label{sec:manifold_model}

While we follow the convention and assume that the patch space could at least be well approximated by a manifold, this assumption certainly deserves more discussion. While this is not the focus of this paper, we mention that the same patch space idea could be applied to study the one dimensional signal; particularly the time series. For example, the same nonlocal median filter idea has been applied to decompose the fetal electrocardiogram signal from the single-lead maternal abdominal electrocardiogram signal \cite{Su_Wu:2016}. 

In this section, we provide a review of another viewpoint of ``getting a manifold'' inside the one dimensional time series. Precisely, we discuss a set of theorems provided in \cite{Takens:1981} and an associated embedding algorithm in the time series framework, which is exactly the patch space of the one dimensional image. The algorithm is well known as the lag map, and has been extensively applied in several fields, for example, the heart rate variability analysis in the bio-medical field. 

From now on, denote $M$ to be a $d$-dim compact manifold without boundary. For the
sake of self-containedness, we recall the following definitions.

\begin{defn}[Discrete time dynamics]
By a discrete time dynamics, we mean a diffeomorhism $\varphi:\, M\rightarrow M$
with the time evolution $i\mapsto\varphi^{i}(x_{0})$, $i\in\mathbb{N}$, where
$x_{0}$ is the starting status.
\end{defn}

\begin{defn}[Continuous time dynamics]
By a continuous time dynamics, we mean a smooth vector field $X\in\Gamma(M)$ with
the time evolution $t\mapsto\gamma_{t}(x_{0})$, where $\gamma_{t}$ is the integral
curve with respect to $X$ via $x_{0}$.
\end{defn}

To simplify the discussion, in both cases, we denote $\Phi_{t}(x_{0})$ to be the
time evolution with time $t\in\mathbb{N}$ or $\mathbb{R}$ with the starting point
$x_0$.

\begin{defn}[Observed time series] Let $\Phi_{t}(x_{0})$ be a dynamics on $M$. The
observation is modeled as a function $f:\, M\rightarrow\mathbb{R}$ and the observed
time series is $f(\Phi_{t}(x_{0}))$.
\end{defn}

The question we have interest in with respect to the patch space formation is that if we have an observed time series $f(\Phi_{t}(x_{0}))$, whether we can recover $M$. Moreover,
can we even recover the dynamics $\Phi_{t}$?
The positive answer and the precise statements are provided in the following two theorems. The proof of these theorems can be found in \cite{Takens:1981}, and the noise analysis could be found in \cite{Stark_Broomhead_Davies_Huke:1997}. Below, by generic, we mean an open dense subset of all possible $(\varphi,f)$. We mention that the theorems hold for non-compact manifolds if $f$ is proper.

\begin{thm}[discrete time dynamics] For a pair $(\varphi,f)$, $\varphi:\,
M^{d}\rightarrow M^{d}$ is the $C^{2}$-diffeomorphism and $f\in C^2$, it is generic
that the map $\Psi:\, M\rightarrow\mathbb{R}^{2d+1}$ given by
$$
\Psi:\, x\mapsto(f(x),\, f(\varphi(x)),\, f(\varphi^{2}(x))\ldots
f(\varphi^{2d}(x)))^{T}\in\mathbb{R}^{2d+1}
$$
is an embedding. 
\end{thm}

\begin{thm}[Continuous time dynamics]
When $X\in C^{2}(\Gamma M)$ and $f\in C^{2}(M)$, it is generic that $\Psi:\,
M\rightarrow\mathbb{R}^{2d+1}$ given by
$$
\Psi:\, x\mapsto(f(x),\, f(\gamma_{1}(x)),\, f(\gamma_{2}(x))\ldots
f(\gamma_{2d}(x)))^{T}\in\mathbb{R}^{2d+1}
$$
is an embedding, where $\gamma_{t}(x)$ is the flow of $X$ of time $t$ via $x$.
\end{thm}

These theorems tell us that we could embed the manifold into a $(2d+1)$ dimensional Euclidean space if we have access to all dynamical behaviors from all points on the manifold. 
However, in practice the above model and theorem cannot be applied directly. Indeed, for a given dynamical system, most of time we may only have one or few experiments that are sampled at discrete times; that is, we only have access to one or few $x\in M$. We thus ask the following question. Suppose we have the time series  
$$
\left\{ f(\Phi_{\ell \alpha}(x))\right\} _{\ell=0}^{N},
$$ 
where $x\in M$ is fixed and inaccessible to us, $\alpha>0$ is the sampling period, and $N\gg1$ is the number of samples, what can we do? We first give the following definition.

\begin{defn} 
The {\it positive limit set (PLS) of $x$} of a vector field $X\in C^{2}\left(\Gamma
M\right)$ is defined as
$$
L_{c}^{+}(x):=\left\{ x'\in M|\,\exists\, t_{i}\rightarrow\infty,\;
t_{i}\in\mathbb{R}\text{ such that }\gamma_{t_{i}}(x)\rightarrow x'\right\}
$$
and the PLS of $x$ of a diffeomorphism $\varphi:\, M\rightarrow M$ is defined as
$$
L_{d}^{+}(x):=\left\{ x'\in M|\,\exists\, n_{i}\in\mathbb{N}\rightarrow\infty,\text{
such that }\varphi^{n_{i}}\left(x\right)\rightarrow x'\right\}.
$$
\end{defn}

It turns out that in this case, we should know whether under generic assumptions the
topology and dynamics in the PLS of $x$ is determined by $\{ f(\Phi_{\ell
\alpha}(x))\} _{\ell=0}^{\infty}$. Precisely, we have the following theorem

\begin{thm}[Continuous dynamics with 1 trajectory]
Fix $x\in M$. When $X\in C^{2}(\Gamma M)$ with flow $\gamma_{t}$ passing $x$, then
there exists a residual subset $C_{X,x}\subset\mathbb{R}^{+}$ such that for all
$\alpha\in C_{X,x}$ and diffeomorphism $\varphi:=\gamma_{\alpha}$, the PLS
$L_{c}^{+}(x)$ for flow $\gamma_{t}$ and $L_{d}^{+}(x)$ for $\varphi$ are the same;
that is, for all $\alpha\in C_{X,x}$ and for all $q\in L_{c}^{+}(p)$, there exists
$n_{i}\in\mathbb{N}\rightarrow\infty$ such that $\varphi^{n_{i}}(x)\rightarrow q$.
\end{thm}

This theorem leads to the following corollary, which is what we need to analyze the time series. 

\begin{cor}
Take $x\in M$, generic $X\in C^{2}(\Gamma M)$ and $f\in C^{2}(M)$, and
$a\in\mathbb{R}^{+}$ satisfying generic conditions depending on $X$ and $x$. Denote
the set
\[
\mathcal{P}:=\left\{f(\gamma_{k\alpha}(x)),\,f(\gamma_{k\alpha}(x)),\ldots,f(\gamma_{(k+2d)\alpha}(x))
\right\}_{k=0}^\infty.
\] 
Then there exists a smooth embedding of $M$ into $\RR^{2d+1}$ mapping PLS
$L_{c}^{+}$ bijectively to the set $\mathcal{P}$.
\end{cor}

While the above model and theorems work well for the one dimensional ``image'' or time series, to the best of our knowledge, there is no parallel theorem for the higher dimensional statement.
In the image processing setup we have interest and the patch space, we could parallel the above setup by viewing an image as an observation of a random field; that is, the temporal one-dimensional axis in the above theorems is replaced by the spatial two dimensional ``time''. Precisely, given a random field defined on $M$, an image could be viewed as an observation of the random field on $M$. 
Now, the patch space could be viewed as a ``two dimensional lag map'' defined on the observation, and we would expect that for a suitably chosen patch size, the patch space could be well approximated by a manifold diffeomorphic to $M$. However, it is not clear at this moment how to justify this statement. We thus conjecture that if an image is generated by this an observation process, then the patch space could be well modeled by a manifold.

\end{document}